\documentclass[12pt,vi]{mitthesis}
\usepackage{lgrind}
\usepackage{cmap}
\usepackage[T1]{fontenc}
\usepackage[autostyle]{csquotes}
\usepackage{graphicx}
\usepackage{wrapfig}

\usepackage{amsmath,amsfonts,bm}

\def\eqref#1{equation~\ref{#1}}

\def\1{\bm{1}}

\def\rr{{\textnormal{r}}}

\def\rvf{{\mathbf{f}}}

\def\rvs{{\mathbf{s}}}

\def\rvw{{\mathbf{w}}}
\def\rvx{{\mathbf{x}}}

\def\vv{{\bm{v}}}

\DeclareMathAlphabet{\mathsfit}{\encodingdefault}{\sfdefault}{m}{sl}
\SetMathAlphabet{\mathsfit}{bold}{\encodingdefault}{\sfdefault}{bx}{n}

\DeclareMathOperator*{\argmin}{arg\,min}

\DeclareMathOperator{\sign}{sign}

\newcommand{\norm}[1]{\left\lVert#1\right\rVert}

\usepackage{booktabs}
\usepackage{url}
\usepackage{amsthm, graphicx, todonotes}
\usepackage[ruled,vlined]{algorithm2e}
\usepackage{lipsum}
\usepackage{todonotes,multirow,makecell}
\usepackage{amsfonts}       
\usepackage{nicefrac}       
\usepackage{microtype}      
\usepackage{xcolor}         
\usepackage[export]{adjustbox}
\usepackage{wrapfig,amsthm,amsmath, textcomp}

\newcommand{\xx}{\mathbf{x}}

\renewcommand{\rr}{\mathbf{r}}
\newcommand{\yy}{\mathbf{y}}
\newcommand{\bfy}{\mathbf{y}}
\newcommand{\ww}{\mathbf{w}}
\newtheorem{proposition}{Proposition}
\newtheorem*{definition*}{Definition}
\usepackage{hyperref}

\DeclareMathOperator{\rmsd}{RMSD}
\newcommand{\comment}[1]{}

\makeatletter

\newcommand*{\@rowstyle}{}

\newcommand*{\rowstyle}[1]{
  \gdef\@rowstyle{#1}%
  \@rowstyle\ignorespaces%
}

\newcolumntype{=}{
  >{\gdef\@rowstyle{}}%
}

\newcolumntype{+}{
  >{\@rowstyle}%
}

\makeatother
\newcommand{\new}[1]{\textcolor{blue}{#1}}

\DeclareMathOperator{\rmsdalign}{RMSDAlign}
\DeclareMathOperator{\uni}{Uni}
\newcommand{\PP}{\mathbb{P}}
\newcommand{\xlig}{\mathbf{x}}
\newcommand{\cc}{\mathbf{c}}
\renewcommand{\ss}{\mathbf{s}}
\newcommand{\dd}{\mathbf{d}}
\newcommand{\pp}{\mathbf{p}}
\newcommand{\bb}{\mathbf{b}}
\renewcommand{\vv}{\mathbf{v}}
\newcommand{\lnorm}{\left|\left|}
\newcommand{\rnorm}{\right|\right|}
\usepackage[ruled,vlined]{algorithm2e}

\def\all{all}
\ifx\files\all  \else 

\begin{document}

\title{Modeling Molecular Structures \\with Intrinsic Diffusion Models}

\author{Gabriele Corso}
\prevdegrees{B.A., University of Cambridge (2021)}
\department{Department of Electrical Engineering and Computer Science}

\degree{Masters of Science}

\degreemonth{February}
\degreeyear{2023}
\thesisdate{January 25, 2023}


\supervisor{Tommi S. Jaakkola}{Professor of Electrical Engineering and Computer Science}

\supervisor{Regina Barzilay}{Distinguished Professor for AI and Health}

\chairman{Leslie A. Kolodziejski}{Professor of Electrical Engineering and Computer Science\\Chair, Department Committee on Graduate Students}

\maketitle



\cleardoublepage
\setcounter{savepage}{\thepage}
\begin{abstractpage}
%
%
%
Since its foundations, more than one hundred years ago, the field of structural biology has strived to understand and analyze the properties of molecules and their interactions by studying the structure that they take in 3D space. However, a fundamental challenge with this approach has been the dynamic nature of these particles, which forces us to model not a single but a whole distribution of structures for every molecular system. 

This thesis proposes Intrinsic Diffusion Modeling, a novel approach to this problem based on combining diffusion generative models with scientific knowledge about the flexibility of biological complexes. The knowledge of these degrees of freedom is translated into the definition of a manifold over which the diffusion process is defined. This manifold significantly reduces the dimensionality and increases the smoothness of the generation space allowing for significantly faster and more accurate generative processes.

We demonstrate the effectiveness of this approach on two fundamental tasks at the basis of computational chemistry and biology: molecular conformer generation and molecular docking. In both tasks, we construct the first deep learning method to outperform traditional computational approaches achieving an unprecedented level of accuracy for scalable programs.

\end{abstractpage}


\cleardoublepage

\section*{Acknowledgments}

First, I would like to thank my advisors Tommi Jaakkola and Regina Barzilay without whom this work would have never been possible. They took a chance on me as an undergrad they never met and gave me full freedom from day one to explore my curiosity. I am sure that for the rest of my Ph.D. and life journey, I will continue to be inspired by Regina's strength and enthusiasm and learn from Tommi's incredible technical and creative insight to become the researcher and mentor I aspire to be.

I am also very thankful to all my collaborators and labmates for the incredible help and support they have given me. In particular, to Octavian-Eugen Ganea (1987-2022), dear colleague, mentor, and friend without whom this work would have never been possible; and to Bowen Jing and Hannes St\"ark with whom I have shared over the past year the research journey that has led to the work presented in this thesis and was made of countless whiteboard discussions, failed experiments, and draft rewrites. 

I would also like to thank Professor Pietro Liò and all the mentors without whose kindness and support I would never be where I am and Renato Berlinghieri, Theo Olausson, Sara Pidò, and all the friends with whom I share this period of my journey through life. 

Finally, this thesis is dedicated to my family, in particular, my parents, Luisella and Mariano, and my fiancée Maëlle-Marie. I am deeply grateful for their unwavering love and sacrifices that have always allowed me to follow my passions and dreams.


\pagestyle{plain}
\tableofcontents

\chapter{Introduction} \label{chapter:intro}


Many of the functions that small molecules and proteins have depend on the 3D structures their atoms take in space. Over the past century, since the development of X-ray crystallography by Max Von Laue in 1912, the field of structural biology has flourished and has been the base of many scientific discoveries and biological models such as the double helical structure of DNA \cite{watson1953molecular}. Since the initial development of computers, researchers have been trying to use algorithms to directly model the structure formed by different molecular complexes without the need for crystallography or other expensive experimental methods. 

One of the fundamental tasks in structural biology, referred to as molecular docking, consists of predicting the position, orientation, and conformation of a ligand when bound to a target protein. The development of accurate docking computational methods in this effort would have a huge impact on drug discovery where researchers look for molecules that are able to bind and inhibit certain protein functions.
Traditional approaches for docking \cite{trott2010autodock,halgren2004glide, koes2013smina} rely on scoring functions that estimate the correctness of a proposed structure and an optimization algorithm that searches for the global maximum of the scoring function. However, since the search space is vast and the landscape of the scoring functions rugged, these methods tend to be too slow and inaccurate.


Recently, the deep learning method AlphaFold2 \cite{jumper2021highly} revolutionized the field of structural biology by being able to accurately (median RMSD below 1\AA{}) predict the folded structure of proteins. AlphaFold2 outperformed by a very large margin existing methods, often based on expensive searches, in the CASP14 competition \cite{moult14critical}, and, since then, has had a significant impact on a large number of downstream applications. Researchers have tried to apply similar ideas and methods \cite{equibind, Lu2022TankBind} to molecular docking without, however, achieving any substantial improvement in accuracy over established search-based methods.

In this thesis, we identify the underlying issue with these existing deep learning methods for molecular docking to be their regression-based training paradigm. This approach fails to capture the flexibility present in molecular structures and to account for model uncertainty. To deal with these two factors, the aleatoric and epistemic uncertainty, that characterize most computational structural biology open challenges, we propose to frame structure prediction as a generative problem.

In recent years, the intersection of generative modeling and deep learning has seen tremendous success with large models now able to generate very realistic text \cite{brown2020language} and images \cite{song2021score}. Deep generative models could hold the key to a solution to the problem of modeling molecular flexibility, however, the direct application of the methods developed for images and natural language fails due to the issues of very high dimensionality and data scarcity. 

It is, therefore, crucial to use scientific insights to build the right degrees of freedom into the generative processes and the right symmetries and inductive biases into the models. This thesis presents Intrinsic Diffusion Modeling (IDM), a generative modeling scheme that builds on the diffusion modeling framework. IDM is based on (1) identifying the extrinsic manifold describing the main degrees of freedom of the structure under analysis, (2) defining the diffusion process on a tractable intrinsic space that can be mapped to the extrinsic manifold, and (3) constructing an equivariant extrinsic-to-intrinsic model mapping points from the extrinsic manifold to scores defined in the tangent space of the intrinsic space.

We first apply IDM to molecular conformer generation, the task of determining the set of conformations that a molecule can take in 3D space. In this setting, we develop \textit{torsional diffusion}, a generative model that, intuitively, learns to model the whole distribution of torsion angles of small molecules and can generate conformations by iteratively refining its position over this torsional manifold. On a standard benchmark of drug-like molecules, \textit{torsional diffusion} generates superior conformer ensembles compared to machine learning and cheminformatics methods in terms of both RMSD and chemical properties, and is orders of magnitude faster than previous diffusion-based models.

We then move to the more complex and data-scarce problem of molecular docking, where we identify the main degrees of freedom of a pose as the position of the ligand relative to the protein, its orientation in the pocket, and the torsion angles describing its conformation. We map the resulting pose manifold to the product space of the degrees of freedom (translational, rotational, and torsional) involved in docking and develop an efficient diffusion process on this space. Empirically, \textsc{DiffDock} obtains a 38\% top-1 success rate (RMSD<2A) on PDBBind, significantly outperforming the previous state-of-the-art of traditional docking (23\%) and deep learning (20\%) methods. Moreover, \textsc{DiffDock} has fast inference times and provides confidence estimates with high selective accuracy.

\section{Overview of Thesis}

In Chapter \ref{chapter:general}, we first provide a general introduction to diffusion generative models, this exposition is mainly based on the formalisation provided by Song et al. \cite{song2021score}. Then, we present \textit{subspace diffusion generative models} where, in the setting of image generation, we show that restricting the diffusion via projections onto subspaces can provide improved runtime and image quality. This section summarises the manuscript:

\begin{displayquote}
\textbf{Subspace Diffusion Generative Models.} Bowen Jing*, Gabriele Corso*, Renato Berlinghieri, and Tommi Jaakkola.  17th European Conference on Computer Vision (ECCV 2022). \cite{jing2022subspace}
\end{displayquote}

At the end of Chapter \ref{chapter:general}, we outline the main ideas and components behind the IDM framework at an abstract level. 

In Chapter \ref{chapter:torsional}, we present \textit{torsional diffusion}, the instantiation of our framework for molecular conformer generation. This chapter is based on the manuscript:

\begin{displayquote}
\textbf{Torsional Diffusion for Molecular Conformer Generation.} Bowen Jing*, Gabriele Corso*, Jeffrey Chang, Regina Barzilay, and Tommi Jaakkola.  Advances in Neural Information Processing Systems 36 (NeurIPS 2022). \cite{jing2022torsional}
\end{displayquote}

Chapter \ref{chapter:diffdock} details how we applied and extended the framework for molecular docking to produce \textsc{DiffDock}. This chapter is based on the manuscript:

\begin{displayquote}
\textbf{DiffDock: Diffusion Steps, Twists, and Turns for Molecular Docking.} Gabriele Corso*, Hannes Stärk*, Bowen Jing*, Regina Barzilay, and Tommi Jaakkola.  11th International Conference on Learning Representations (ICLR 2023). \cite{corso2022diffdock}
\end{displayquote}

Finally, in Chapter \ref{chapter:conclusion}, we conclude by summarising the thesis and discussing the avenues for future research that this work opens.

\chapter{Intrinsic Diffusion Models} \label{chapter:general}

\section{Deep Generative Models}

Over the past decade, deep learning methods \cite{goodfellow2016deep} have achieved impressive results in the supervised learning tasks of classification and regression. In image classification, neural networks have been very successful \cite{krizhevsky2017imagenet} in learning to predict $p(y|x)$ the probability of image $x$ being of some label $y$ by minimizing some loss such as a multi-class cross-entropy. This approach of predicting probability values for every class is, however, not feasible in continuous spaces, where regression methods aim to learn an estimator $\hat{y}(x)$ that minimizes some loss function, often the mean squared error $\mathbb{E}_{p(\cdot | x)}[(\hat{y}(x)-y)^2]$. AlphaFold2 \cite{jumper2021highly} is one successful example of this approach learning the protein structure $y$ given its sequence $x$.

However, most problems on real-world continuous domains are not deterministic, therefore an accurate solution to them requires modeling the whole (conditional) probability distribution $p(\cdot|x)$. This is the goal of the field of generative modeling. Although one would often ideally want to obtain an analytical and tractable form of $p(\cdot|x)$, this is unfeasible for most complex real-world distributions, therefore, the goal of generative models is typically that of learning how to sample $y\sim p(\cdot|x)$ and, often, evaluate the likelihood of a given point $p(y|x)$. 

The intersection of the fields of deep learning and generative modeling has been a particularly flourishing one over the past decade. Developments in deep neural networks have provided very powerful function approximators, however, the question of how to use them to learn probability distributions is non-trivial and requires significant ingenuity. To answer this question a wide range of techniques have been proposed including autoregressive models, variational autoencoders \cite{kingma2013auto}, continuous normalizing flows \cite{dinh2016density}, generative adversarial networks \cite{goodfellow2020generative} and diffusion models \cite{sohl2015deep, song2021score}.

\section{Diffusion Generative Models}

Inspired by statistical physics, diffusion generative models\footnote{Also known as score-based generative models, denoising diffusion models or just diffusion models.} are a class of generative models based on the idea that, adding noise to the data distribution, one defines a gradual mapping between the data distribution and an approximate prior distribution that can be easily sampled \cite{sohl2015deep}. A neural network is then trained to reverse small steps of this noise addition process allowing to sample from the data distribution starting from a sample of the prior.

In this thesis, we will mainly follow the stochastic differential equation (SDE) formalization of diffusion models introduced by Song et al. \cite{song2021score}. In this formalization, the data distribution is considered to be the starting distribution $p_0(\rvx)$ of a forward diffusion process described, in Euclidean space, by the Ito SDE:
\begin{equation}
    d\rvx = \rvf(\rvx, t)dt + g(t) d\rvw, \; t \in (0,T)
\end{equation}
where $\rvw$ is the Wiener process and $\rvf(\rvx, t)$ and $g(t)$ are chosen functions referred to as \textit{drift} and \textit{diffusion} coefficients. As $t$ grows, the distribution approaches a Gaussian, therefore, for large enough $T$, we can approximate a sample from the prior $p_T(\rvx)$ by sampling from a Gaussian distribution. A theorem from Anderson \cite{anderson1982reverse} guarantees that the reverse of a diffusion process is also a well-defined diffusion process given by the following reverse-time SDE:
\begin{equation}
    d\rvx = [\rvf(\rvx, t)-g(t)^2 \nabla_{\rvx} \log p_t(\rvx)]dt + g(t) d\rvw
\end{equation}
Therefore, if we know $\nabla_{\rvx} \log p_t(\rvx)$ for all $t\in(0,T)$ we can sample from $p_0(\rvx)$ by sampling from $p_T(\rvx)$ and running the reverse-time SDE.

Moreover, Song et al. \cite{song2021score} also showed that the score $\nabla_{\rvx} \log p_t(\rvx)$ can be used to define the probability flow ODE, a deterministic process whose trajectories have the same marginal probability densities as the SDE:
\begin{equation}
    d\rvx = [\rvf(\rvx, t)-\frac{1}{2} g(t)^2 \nabla_{\rvx} \log p_t(\rvx)]dt
\end{equation}

To obtain estimates of $\nabla_{\rvx} \log p_t(\rvx)$ we train a score model $s_{\theta}(\rvx,t)$ via denoising score matching \cite{song2019generative}:
\begin{equation}
   \mathbf{\theta}^* = \argmin_\mathbf{\theta}
   \mathbb{E}_{t}\Big\{\lambda(t) \mathbb{E}_{\rvx(0)}\mathbb{E}_{\rvx(t) \mid \rvx(0) }
   \big[\norm{\rvs_\mathbf{\theta}(\rvx(t), t) - \nabla_{\rvx(t)}\log p_{0t}(\rvx(t) \mid \rvx(0))}_2^2 \big]\Big\}
\end{equation}

Finally, De Bortoli et al. \cite{de2022riemannian} showed that the framework presented above holds with few modifications on (non-Euclidean) compact Riemannian manifolds, as long as one is able to sample the heat kernel, compute its score and sample from the stationary distribution of these manifolds. Critically, the score is defined in the tangent space of the manifolds. 

\section{Subspace Diffusion Generative Models}

In the dominant formulation of diffusion generative models, the forward diffusion occurs in the full ambient space of the data distribution, destroying its structure but retaining its high dimensionality. It does not seem parsimonious to represent increasingly noisy latent variables---which approach zero mutual information with the original data---in a space with such high dimensionality. The practical implications of this high latent dimensionality are twofold: 

\emph{High-dimensional extrapolation}. The network must learn the score function over the entire support of the high-dimensional latent variable, even in areas very far (relative to the scale of the data) from the data manifold. Due to the curse of dimensionality, much of this support may never be visited during training, and the accuracy of the score model in these regions is called into question by the uncertain extrapolation abilities of neural networks \cite{xu2020neural}. Learning to match a lower-dimensional score function may lead to refined training coverage and further improved performance.

\emph{Computational cost}. Hundreds or even thousands of evaluations of the high-dimensional score model are required to generate an image, making inference with score-based models much slower than with GANs or VAEs \cite{ho2020denoising,song2021score}. A number of recent works aim to address this challenge by reducing the number of steps required for inference \cite{song2020denoising,salimans2021progressive,jolicoeur2021gotta,nichol2021improved,dhariwal2021diffusion,kong2021fast,watson2021learning,san2021noise,lam2021bilateral,bao2022analytic}. However, these methods generally trade-off inference runtime with sample quality. Moreover, the dimensionality of the score function---and thereby the computational cost of a single score evaluation---is an independent and equally important factor to the overall runtime, but this factor has received less attention in existing works.

\begin{figure}[t]
    \centering
    \includegraphics[width=0.9\textwidth]{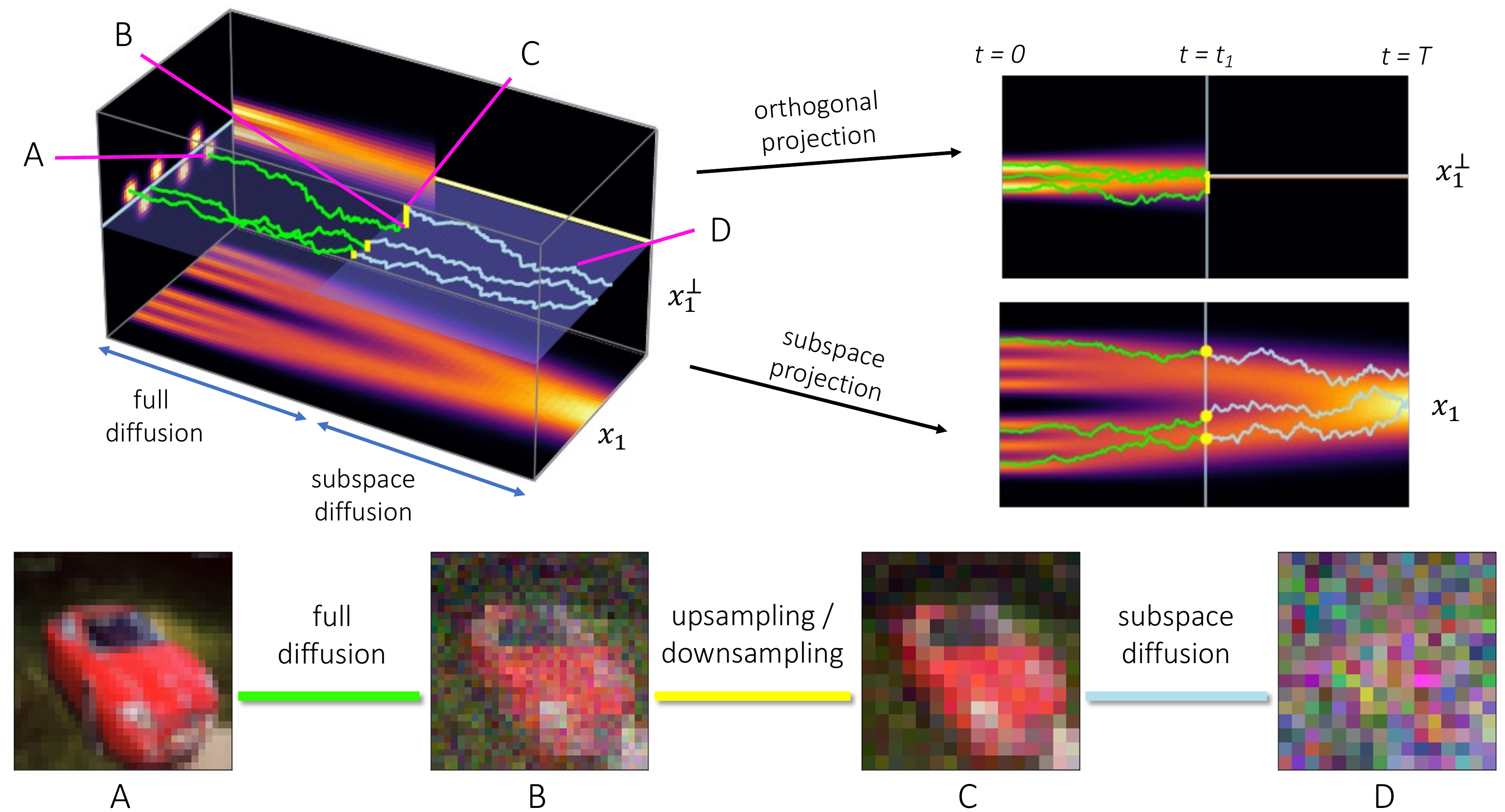}
    \caption{Visual schematic of subspace diffusion with one projection step. \emph{Top left}: The starting data distribution $\mathbf{x}_0(0)$ lies near a subspace (light blue line). As the data evolves, the distribution of the orthogonal component $\mathbf{x}_1^\perp(t)$ approaches a Gaussian faster than the subspace component $\mathbf{x}_1(t)$. At time $t_1$ we project onto the subspace and restrict the remaining diffusion to the subspace. To generate data, we use the full and subspace score models to reverse the full and subspace diffusion steps, and sample $\mathbf{x}_1^\perp(t_1)$ from a Gaussian to reverse the projection step. \emph{Top right}: The diffusion of the subspace component $\mathbf{x}_1(t)$ is unaffected by the projection step and restriction to the subspace; while the orthogonal component is diffused until $t_1$ and discarded afterward. \emph{Bottom}: CIFAR-10 images corresponding to points along the trajectory, where the subspaces correspond to lower-resolution images and projection is equivalent to downsampling.}
    \label{fig:method}
\end{figure}

\textit{Subspace diffusion models} aim to address these challenges. In some real-world domains such as images, target data lie near a linear subspace, such that under isotropic forward diffusion, the components of the data orthogonal to the subspace become Gaussian significantly before the components in the subspace. We propose to use a full-dimensional network to model the score only at lower noise levels, when all components are sufficiently non-Gaussian. At higher noise levels, we use smaller networks to model in the subspace only those components of the score which remain non-Gaussian. As this reduces both the number and domain of queries to the full-dimensional network, subspace diffusion addresses both of our motivating concerns. Moreover, in contrast to many prior works, subspace diffusion remains fully compatible with the underlying continuous diffusion framework \cite{song2021score}, and therefore preserves all the capabilities available to continuous score-based models, such as likelihood evaluation, probability flow sampling, and controllable generation. 

While subspace diffusion can be applied to arbitrary settings, we focus on generative modeling of natural images. Because the global structure of images is dominated by low-frequency visual components---i.e., adjacent pixels values are highly correlated---images lie close to subspaces corresponding to lower-resolution versions of the same image. 

Experimentally, we train and evaluate lower-dimensional subspace models in conjunction with state-of-the-art pretrained full-dimensional models from \cite{song2021score}. We improve over those models in sample quality and runtime, achieving an FID of 2.17 and a IS of 9.99 on CIFAR-10 generation with more than 30\% inference time reduction. 

These results are a first hint at the value of restricting the diffusion processes in lower dimensional spaces where the main degrees of freedom lie. This is very promising since, according to the manifold hypothesis, in many real-world domains the high-dimensional data points lie near low-dimensional latent spaces. However, limiting ourselves to Euclidean subspaces, as in \textit{subspace diffusion}, does not allow, for most problems, to notably reduce the dimensionality of the space over which to operate. This thesis proposes an alternative approach to model significantly more complex manifolds in an efficient and effective manner.

\section{Intrinsic Diffusion Models} \label{sec:intrinsic}

As discussed in the previous section, we hypothesize that restricting the diffusion process to a submanifold, that, approximately, contains all the datapoints of interest can offer significant improvements in terms of both accuracy and inference time. How to define a diffusion process on such a manifold and construct a score model that is able to generalize to different chemical systems is not straightforward. 

The generalization component is particularly important because in this thesis we will operate in inductive settings, where the distributions that we want to sample during inference might not be seen during training. For example, for conformer generation, we expect our method to run on any molecule, regardless of whether it was part of our training set.

In this section, we present the blueprint of Intrinsic Diffusion Modeling (IDM), the approach that we will show to be very effective in the tasks of molecular conformer generation and molecular docking in the rest of the thesis. IDM is composed of four main components:
\begin{enumerate}
    \item \textit{flexibility}: identification of the extrinsic manifold,
    \item \textit{mapping}: definition of the intrinsic manifold and its mapping to the extrinsic,
    \item \textit{diffusion}: specification of a diffusion process on the intrinsic manifold,
    \item \textit{score model}: construction of an extrinsic-to-intrinsic score model.
\end{enumerate}

Below we present each component in its abstract form, the reader will likely more clearly understand them by further reading the two examples of their concrete instantiations in \textit{torsional diffusion} and \textsc{DiffDock}.

\subsection{Flexibility}

Firstly, one needs to identify a low-dimensional manifold that describes most of the entropy in the distribution under analysis, we will call this manifold the extrinsic space. For the domains that we will analyze in this thesis the definition of this manifold comes from domain knowledge, trying to discover these manifolds from data directly is a very interesting avenue for future work.

Importantly, to run inference in inductive settings, one needs to have a way to identify the chosen manifold, e.g. by sampling one of its points, for any query at inference time. Moreover, if the data does not lie exactly on the manifold, but only approximately, one also needs to define a way of projecting datapoints to the manifold. Then, to avoid distributional shift at inference time, we preprocess the training data by sampling the manifold and projecting the datapoint onto it. Training is then run with these projected datapoints.

\subsection{Mapping}

De Bortoli et al. \cite{de2022riemannian} defines diffusion models for arbitrary submanifolds in terms of projecting a diffusion in ambient space onto the submanifold. However, the corresponding kernel $p(\rvx_t|\rvx_0)$ is not available in closed form and has to be sampled numerically with geodesic random walks. This makes the training process very slow or imprecise. Instead, we take a different approach defining a bijection between the extrinsic manifold and simpler intrinsic space over which we run the diffusion.

As the name suggests, we will use the definition of some intrinsic coordinates to define the intrinsic space. Critically these different coordinates must be disentangled from each other, forming, therefore, a bijection with the extrinsic manifold and guaranteeing an equivalence between distributions on the intrinsic and extrinsic manifolds. 

\subsection{Diffusion}

One then needs to derive the fundamental components of the diffusion process on the chosen intrinsic space. In particular, to train the diffusion model and run inference, we have to be able to sample the heat kernel of the diffusion, compute its score and sample from the stationary distribution. 

Luckily, for most well-studied spaces that typically compose an intrinsic coordinate space, the Brownian motion, modeled as a Geodesic Random walk, has a known closed-form solution for computing its kernel and score and simple procedures to transform samples from common distributions to sample from its stationary distribution. This allows us to avoid having to simulate geodesic random walks as described in De Bortoli et al. \cite{de2022riemannian} for a general manifold. 

\subsection{Score model}

Finally, we need to construct a score model $s_{\theta}(\xx, t)$ that for each point $\xx$ and diffusion time $t$ predicts the score of the diffused data distribution at that point on the intrinsic manifold. 

Naively, we may construct a model that works exclusively on the intrinsic manifold by taking as input the intrinsic coordinates of the current point and predicting its score. This, however, would not be able to generalize well across systems because: (1) the definition of intrinsic coordinates often requires arbitrary choices such as the order of the coordinates or their origin but the data distribution is not invariant to such choices (e.g. definition of torsion angle around a bond); (2) laws of physical interactions can be more easily described in terms of extrinsic coordinates rather than intrinsic ones (e.g. electrostatic interactions between atoms far in the molecular graph). These limitations are also one reason why previous attempts to learn distributions of structures via intrinsic coordinates have failed to generalize to multiple chemical systems \cite{noe2019boltzmann}. 

For this reason, we propose to, instead, operate in an extrinsic-to-intrinsic framework, where the score model takes in a point described in extrinsic coordinated (e.g. a 3D molecular graph) and predicts the score in terms of its intrinsic coordinated (e.g. change in torsion angles). By taking as input the object described in its extrinsic coordinates we avoid the model being influenced by arbitrary choices of origin for intrinsic coordinates and can more easily reason about physical interactions. Moreover, although the model predicts the score (which translates into an update) on the intrinsic manifold this is can be directly applied to the point in the extrinsic manifold (e.g. rotate one of the torsion angles) without ever needing to instantiate the intrinsic space.

\chapter{Torsional Diffusion} \label{chapter:torsional}

Many properties of a molecule are determined by the set of low-energy structures, called \emph{conformers}, that it adopts in 3D space. Conformer generation is therefore a fundamental problem in computational chemistry \cite{hawkins2017conformation} and an area of increasing attention in machine learning. Traditional approaches to conformer generation consist of metadynamics-based methods, which are accurate but slow \cite{pracht2020automated}; and cheminformatics-based methods, which are fast but less accurate \cite{hawkins2010conformer, riniker2015better}. Thus, there is growing interest in developing deep generative models to combine high accuracy with fast sampling. 

Diffusion or score-based generative models \cite{ho2020denoising, song2021score} have been applied to conformer generation under several different formulations. These have so far considered diffusion processes in \emph{Euclidean} space, in which Gaussian noise is injected independently into every data coordinate---either pairwise distances in a distance matrix  \cite{shi2021learning, luo2021predicting} or atomic coordinates in 3D \cite{xu2021geodiff}. However, these models require a large number of denoising steps and have so far failed to outperform the best cheminformatics methods.

We instead propose \emph{torsional diffusion}, in which the diffusion process over conformers acts only on the torsion angles and leaves the other degrees of freedom fixed. This is possible and effective because the flexibility of a molecule, and thus the difficulty of conformer generation, lies largely in torsional degrees of freedom \cite{axelrod2020geom}; in particular, bond lengths and angles can already be determined quickly and accurately by standard cheminformatics methods. Leveraging this insight significantly reduces the dimensionality of the sample space; drug-like molecules\footnote{As measured from the standard dataset GEOM-DRUGS \cite{axelrod2020geom}} have, on average, $n=44$ atoms, corresponding to a $3n$-dimensional Euclidean space, but only $m=7.9$ torsion angles of rotatable bonds.

Empirically, we obtain state-of-the-art results on the GEOM-DRUGS dataset \cite{axelrod2020geom} and are the first method to consistently outperform the established commercial software OMEGA \cite{hawkins2017conformation}. We do so using two orders of magnitude \emph{fewer} denoising steps than GeoDiff \cite{xu2021geodiff}, the best Euclidean diffusion approach.

Unlike prior work, our model provides exact likelihoods of generated conformers, enabling training with the ground-truth \emph{energy} function rather than samples alone. This connects with the literature on \emph{Boltzmann generators}---generative models which aim to sample the Boltzmann distribution of physical systems without expensive molecular dynamics or MCMC simulations \cite{noe2019boltzmann, kohler2021smooth}. Thus, as a variation on the torsional diffusion framework, we develop \emph{torsional Boltzmann generators} that can approximately sample the conditional Boltzmann distribution for unseen molecules. This starkly contrasts with existing Boltzmann generators, which are specific for the chemical system on which they are trained.

This chapter is mostly based on the paper:
\begin{displayquote}
\textbf{Torsional Diffusion for Molecular Conformer Generation.} Bowen Jing*, Gabriele Corso*, Jeffrey Chang, Regina Barzilay, and Tommi Jaakkola.  Advances in Neural Information Processing Systems 35 (NeurIPS 2022).
\end{displayquote}

\section{Background} \label{sec:background}

\paragraph{Molecular conformer generation.} The \emph{conformers} of a molecule are the set of its energetically favorable 3D structures, corresponding to local minima of the potential energy surface. The gold standards for conformer generation are metadynamics-based methods such as CREST \cite{pracht2020automated}, which explore the potential energy surface while filling in local minima \cite{hawkins2017conformation}. However, these require an average of 90 core-hours per drug-like molecule \cite{axelrod2020geom} and are not considered suitable for high-throughput applications. Cheminformatics methods instead leverage approximations from chemical heuristics, rules, and databases for significantly faster generation \cite{lagorce2009dg,cole2018knowledge,miteva2010frog2,bolton2011pubchem3d,li2007caesar}; while these can readily model highly constrained degrees of freedom, they fail to capture the full energy landscape. The most well-regarded of such methods include the commercial software OMEGA \cite{hawkins2010conformer} and the open-source RDKit ETKDG \cite{landrum2013rdkit,riniker2015better}.

A number of machine learning methods for conformer generation has been developed \cite{xu2020learning,xu2021end,shi2021learning,luo2021predicting}, the most recent and advanced of which are GeoMol \cite{ganea2021geomol} and GeoDiff \cite{xu2021geodiff}. GeoDiff is a Euclidean diffusion model that treats conformers as point clouds $\mathbf{x} \in \mathbb{R}^{3n}$ and learns an $SE(3)$ equivariant score. On the other hand, GeoMol employs a graph neural network that, in a single forward pass, predicts neighboring atomic coordinates and torsion angles from a stochastic seed. 

\paragraph{Boltzmann generators.} An important problem in physics and chemistry is that of generating independent samples from a Boltzmann distribution $p(\mathbf{x}) \propto e^{-E(\mathbf{x})/kT}$ with known but unnormalized density.\footnote{This is related to but distinct from conformer generation, as conformers are the local minima of the Boltzmann distribution rather than independent samples.} Generative models with exact likelihoods, such as normalizing flows, can be trained to match such densities \cite{noe2019boltzmann} and thus provide independent samples from an approximation of the target distribution. Such \emph{Boltzmann generators} have shown high fidelity on small organic molecules \cite{kohler2021smooth} and utility on systems as large as proteins \cite{noe2019boltzmann}. However, a {separate model} has to be trained for every molecule, as the normalizing flows operate on intrinsic coordinates whose definitions are specific to that molecule. This limits the utility of existing Boltzmann generators for molecular screening applications.

\begin{figure}[t]
    \centering
    \includegraphics[width=\textwidth]{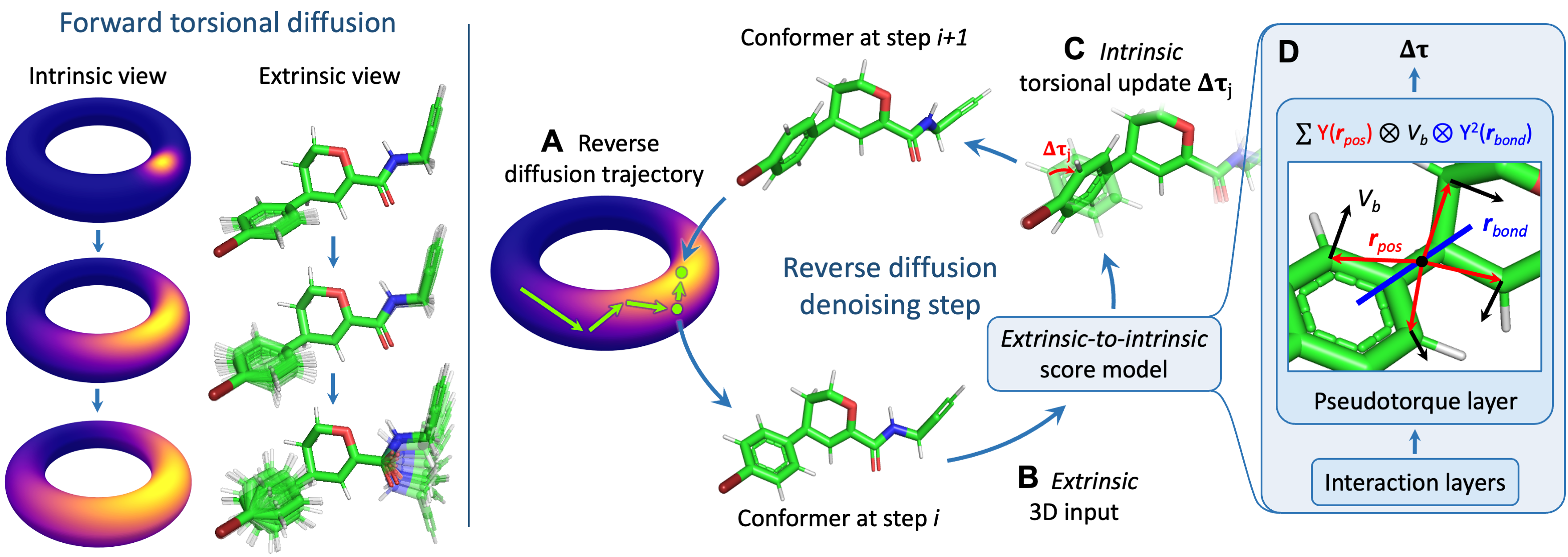}
    \caption{\textbf{Overview of torsional diffusion.} \emph{Left}: Extrinsic and intrinsic views of torsional diffusion (only 2 dimensions/bonds shown). \emph{Right}: In a step of reverse diffusion (\textbf{A}), the current conformer is provided as a 3D structure (\textbf{B}) to the score model, which predicts intrinsic torsional updates (\textbf{C}). The final layer of the score model is constructed to resemble a torque computation around each bond (\textbf{D}). $Y$ refers to the spherical harmonics and $V_b$ the learned atomic embeddings.}
    \label{fig:overview}
\end{figure}

\section{Method} \label{sec:torsional_diffusion}

Consider a molecule as a graph $G = (\mathcal{V}, \mathcal{E})$ with atoms $v \in \mathcal{V}$ and bonds $e \in \mathcal{E}$, and denote the space of its possible conformers $\mathcal{C}_G$. A conformer $C \in \mathcal{C}_G$ is typically defined in terms of its \emph{extrinsic} (or Cartesian) coordinates---that is, as a point cloud in 3D space, defined up to global roto-translation: $\mathcal{C}_G \cong \mathbb{R}^{3n} / SE(3)$. However, we can also described in terms of its \emph{intrinsic} (or internal) coordinates: local structures $L$ consisting of bond lengths, bond angles, and cycle conformations; and torsion angles $\boldsymbol{\tau}$ consisting of dihedral angles around freely rotatable bonds. We consider a bond \emph{freely rotatable} if severing the bond creates two connected components of $G$, each of which has at least two atoms. Thus, torsion angles in cycles (or rings), which cannot be rotated independently, are considered part of the local structure $L$. 

Our method, illustrated in Figure~\ref{fig:overview}, uses the intrinsic diffusion modeling framework to define a diffusion process over the space of structures defined by some local structure. Below we detail each of the four components outlined in Section \ref{sec:intrinsic}.

\subsection{Flexibility}

Conformer generation consists of learning probability distributions $p_G(C)$. However, the set of possible stable local structures $L$ for a particular molecule is very constrained and can be accurately predicted by fast cheminformatics methods, such as RDKit ETKDG \cite{riniker2015better}. Thus, we use RDKit to provide approximate samples from $p_G(L)$, and develop a diffusion model to learn distributions $p_G(C \mid L)$. We have therefore defined the extrinsic space as the submanifold defined by conditioning $C$ on a given local structure $L$. Since we will use RDKit to obtain samples from the local structure, to identify a point on the manifold at inference time, we will simply embed the given molecule. From this conformer, any point on our extrinsic manifold will be reachable with some change in torsion angles.

\paragraph{Conformer matching.} In focusing on $p_G(C\mid L)$, we have assumed that we can sample local structures $L \sim p_G(L)$ with RDKit. While this assumption is very good in terms of RMSD, the RDKit marginal $\hat p_G(L)$ is only an approximation of the ground truth $p_G(L)$. Thus, if we train on the denoising score-matching loss with ground truth conformers---i.e., conditioned on ground truth local structures---there will be a distributional shift at test time, where only approximate local structures from $\hat p_G(L)$ are available. We found that this shift significantly hurts performance.

We thus introduce a preprocessing procedure called \textit{conformer matching}. In brief, for the \emph{training} split only, we substitute each ground truth conformer $C$ with a synthetic conformer $\hat{C}$ with local structures $\hat{L} \sim \hat p_G(L)$ and made as similar as possible to $C$. That is, we use RDKit to generate $\hat{L}$ and change torsion angles $\hat{\boldsymbol{\tau}}$ to minimize $\rmsd(C, \hat{C})$. Naively, we could sample $\hat{L} \sim \hat p_G(L)$ independently for each conformer, but this eliminates any possible dependence between $L$ and $\boldsymbol{\tau}$ that could serve as training signal. Instead, we view the distributional shift as a domain adaptation problem that can be solved by optimally aligning $p_G(L)$ and $\hat p_G(L)$. See Appendix \ref{app:matching} for details.

\subsection{Mapping}

The extrinsic submanifold we have identified by conditioning $C$ on a given local structure $L$ is, however, very complex to deal with in Euclidean space. We, therefore, exploit the fact that a conformer can be univocly defined in terms of its internal coordinates $L$ and $\boldsymbol{\tau}$\footnote{This is true because we are only interested in conformers up to SE(3) transformations.}. In particular, almost surely\footnote{Unless we have all the atoms on onne side of a rotatable bond lying all exactly on the line defined by the bond.}, there is a bijection between the torsion angles $\boldsymbol{\tau}$ and the extrinsic manifold $C \mid L$. 

Since each torsion angle coordinate lies in $[0, 2\pi)$, the $m$ torsion angles of a conformer define a hypertorus $\mathbb{T}^m$. This is the intrinsic manifold over which we train the diffusion model to sample from $p_G(\boldsymbol{\tau} \mid L)$.

\subsection{Diffusion} \label{sec:toroidal_diff}

To learn a generative model over the intrinsic manifold $\mathbb{T}^m$, we apply the continuous score-based framework of Song et al. \cite{song2021score}, which holds with minor modifications on compact Riemannian manifolds  \cite{de2022riemannian}. For the forward diffusion we use rescaled Brownian motion given by $\mathbf{f}(\mathbf{x}, t) = 0, g(t) = \sqrt{\frac{d}{dt} \sigma^2(t)}$ where $\sigma(t)$ is the noise scale. Specifically, we use an exponential diffusion $\sigma(t) = \sigma^{1-t}_\text{min}\sigma^t_\text{max}$ as in Song et al. \cite{song2019generative}, with $\sigma_\text{min}=0.01\pi$, $\sigma_\text{max} = \pi, t \in (0, 1)$. Due to the compactness of the manifold, however, the prior $p_T(\mathbf{x})$ is no longer a Gaussian, but a \emph{uniform} distribution over $M$. 

Training the score model with denoising score matching requires a procedure to sample from the perturbation kernel $p_{t\mid 0}(\mathbf{x}' \mid \mathbf{x})$ of the forward diffusion and compute its score. We view the torus $\mathbb{T}^m \cong [0, 2\pi)^m$ as the quotient space $\mathbb{R}^m/2\pi\mathbb{Z}^m$ with equivalence relations $(\tau_1, \ldots \tau_m) \sim (\tau_1+2\pi, \ldots, \tau_m) \ldots \sim (\tau_1, \ldots \tau_m+2\pi)$. Hence, the perturbation kernel for rescaled Brownian motion on $\mathbb{T}^m$ is the \emph{wrapped normal distribution} on $\mathbb{R}^m$; that is, for any $\boldsymbol{\tau}, \boldsymbol{\tau}' \in [0, 2\pi)^m$, we have
\begin{equation} \label{eq:torus_score}
    p_{t\mid 0}(\boldsymbol{\tau}' \mid \boldsymbol{\tau}) \propto \sum_{\mathbf{d} \in \mathbb{Z}^m} \exp\left(-\frac{||\boldsymbol{\tau} - \boldsymbol{\tau}' + 2\pi\mathbf{d}||^2}{2\sigma^2(t)}\right)
\end{equation}
where $\sigma(t)$ is the noise scale of the perturbation kernel $p_{t\mid 0}$. We thus sample from the perturbation kernel by sampling from the corresponding unwrapped isotropic normal and taking elementwise $\mod 2\pi$. The scores of the kernel are pre-computed using a numerical approximation. During training, we sample times $t$ at uniform and minimize the denoising score matching loss
\begin{equation} \label{eq:dsm}
    J_\text{DSM}(\theta) = \mathbb{E}_t\left[\lambda(t)\mathbb{E}_{\boldsymbol{\tau}_0\sim p_0,\boldsymbol{\tau}_t\sim p_{t\mid 0}(\cdot \mid \boldsymbol{\tau}_0)}\left[||\mathbf{s}(\boldsymbol{\tau}_t, t) - \nabla_{\boldsymbol{\tau}_t} \log p_{t\mid 0}(\boldsymbol{\tau}_t \mid \boldsymbol{\tau}_0)||^2\right]\right]
\end{equation}
where the weight factors $
    \lambda(t) = 1/\mathbb{E}_{\boldsymbol{\tau} \sim p_{t\mid 0}(\cdot \mid 0)}\left[||\nabla_{\boldsymbol{\tau}} \log p_{t\mid 0}(\boldsymbol{\tau} \mid \mathbf{0})||^2\right]$
are also precomputed. As the tangent space $T_{\boldsymbol{\tau}}\mathbb{T}^m$ is just $\mathbb{R}^m$, all the operations in the loss computation are the familiar ones.

For inference, we first sample from a uniform prior over the torus. We then discretize and solve the reverse diffusion with a geodesic random walk; however, since the exponential map on the torus (viewed as a quotient space) is just $\exp_{\boldsymbol{\tau}}(\boldsymbol{\delta}) = \boldsymbol{\tau} + \boldsymbol{\delta} \mod 2\pi$, the geodesic random walk is equivalent to the wrapping of the random walk on $\mathbb{R}^m$.

\paragraph{Low-temperature sampling. } The score-matching loss used to train the score model minimizes an upper bound on the KL divergence between the model and the data distribution. Although when perfectly learned this leads to the two distributions being exactly equal, in the realistic case of limited data and model capacity the model will tend to learn an overdispersed distribution. Low-temperature sampling of some distribution $p(\xx)$ with temperature $\lambda^{-1}<1$ consists of sampling the distribution $p_{\lambda}(\xx) \propto p(\xx)^{\lambda}$. This mitigates the overdispersion problem by concentrating more on high-likelihood modes and effectively trading sample diversity for quality \cite{ingraham2022illuminating}.

Exact low-temperature sampling is intractable for most generative models, however, various approximation schemes exist. We use an adaptation of Hybrid Langevin-Reverse Time SDE proposed by Ingraham et al. \cite{ingraham2022illuminating}:
\begin{equation*}
    d \boldsymbol{\tau} = -\new{\bigg(\lambda_t + \frac{\lambda \; \psi}{2} \bigg)}\; \ss_{\theta, G}(C, t) \; g^2(t) \; dt + 
    \new{\sqrt{1+\psi}} \; g(t) \; d\mathbf{w} \quad
    \text{with } \lambda_t = \frac{\sigma_d + \sigma_t}{\sigma_d + \sigma_t/\lambda}
\end{equation*}
where $\lambda$ (the inverse temperature), $\psi$ and $\sigma_d$ are parameters that can be tuned. Setting the \new{blue} components to 1 recovers the standard reverse time SDE. 

\subsection{Score model}

\subsubsection{Extrinsic-to-intrinsic model} \label{sec:score_framework}

\begin{figure}
    \begin{center}
    \includegraphics[width=0.35\textwidth]{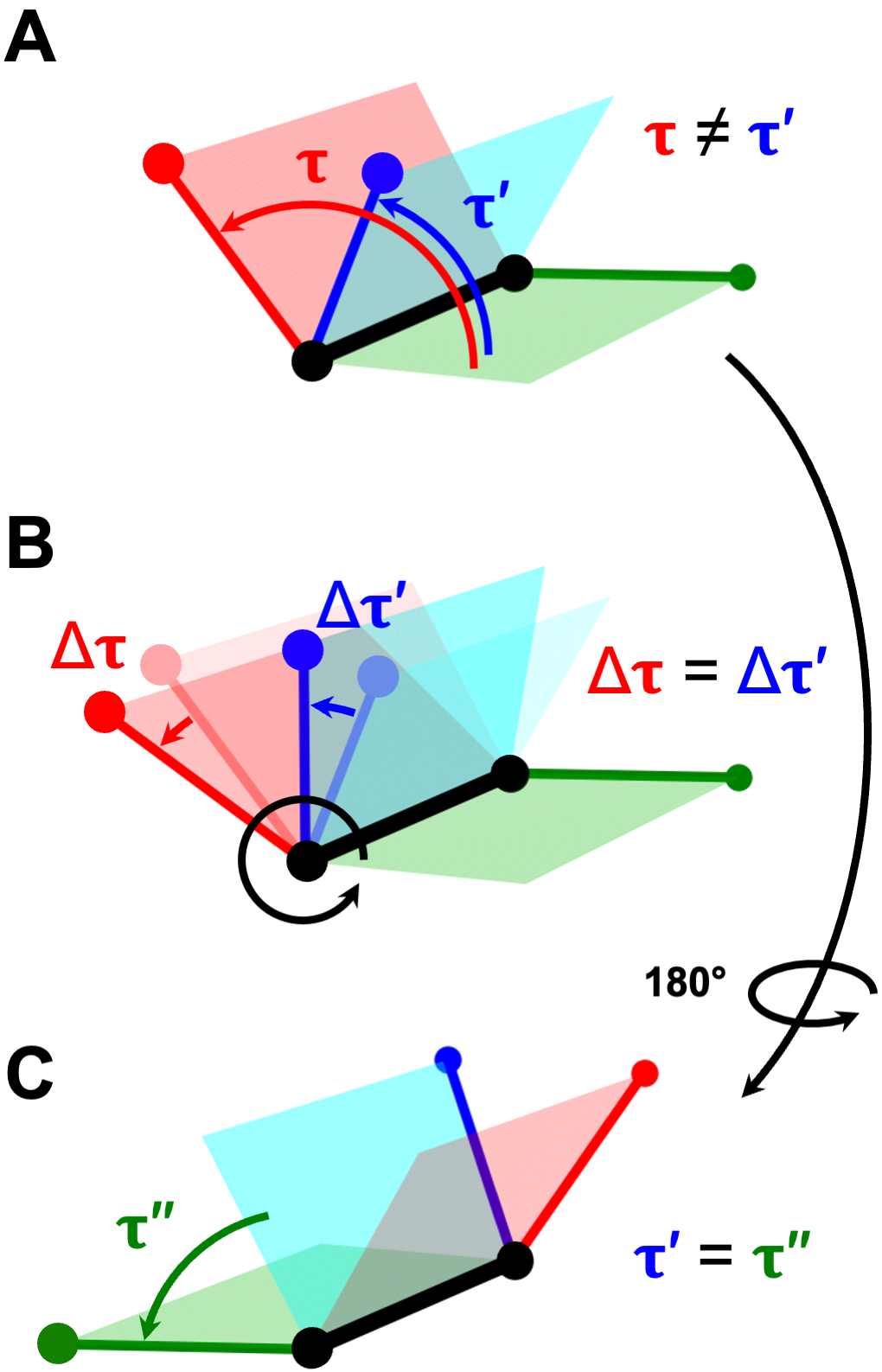}
    \caption{\textbf{A}: The torsion $\tau$ around a bond depends on a choice of neighbors. \textbf{B}: The \emph{change} $\Delta\tau$ caused by a relative rotation is the same for all choices. \textbf{C}: The sign of $\Delta\tau$ is unambiguous because given the same neighbors, $\tau$ does not depend on bond direction.}
    \label{fig:torsion}
    \end{center}
\end{figure}

While we have defined the diffusion process over intrinsic coordinates, learning a score model $\mathbf{s}(\boldsymbol{\tau}, t)$ directly over intrinsic coordinates is potentially problematic for several reasons. First, the dimensionality $m$ of the torsional space depends on the molecular graph $G$. Second, the mapping from torsional space to physically distinct conformers depends on $G$ and local structures $L$, but it is unclear how to best provide these to a model over $\mathbb{T}^m$. Third, there is no canonical choice of independent intrinsic coordinates $(L, \boldsymbol{\tau})$; in particular, the torsion angle at a rotatable bond can be defined as any of the dihedral angles at that bond, depending on an arbitrary choice of reference neighbors (Figure \ref{fig:torsion} and Appendix~\ref{app:def}). Thus, even with fixed $G$ and $L$, the mapping from $\mathbb{T}^m$ to conformers is ill-defined. This posed a significant challenge to prior works using intrinsic coordinates \cite{ganea2021geomol}.

To circumvent these difficulties, we instead consider a conformer $C\in \mathcal{C}_G$ in terms of its extrinsic coordinates. Then, we construct the score model $\mathbf{s}_G(C, t)$ as a function over $\mathcal{C}_G$ rather than $\mathbb{T}^m$. The outputs remain in the tangent space of $\mathbb{T}^m$, which is just $\mathbb{R}^m$. Such a score model is simply an $SE(3)$-\emph{invariant} model over point clouds in 3D space $\mathbf{s}_G: \mathbb{R}^{3n} \times [0, T] \mapsto \mathbb{R}^{m}$ conditioned on $G$. Thus, we have reduced the problem of learning a score on the torus, conditioned on the molecular graph and local structure, to the much more familiar problem of predicting $SE(3)$-invariant scalar quantities---one for each bond---from a 3D conformer.

It may appear that we still need to choose a definition of each torsion angle $\tau_i$ so that we can sample from  $p_{t\mid 0}(\cdot \mid \boldsymbol{\tau})$ during training and solve the reverse SDE over $\boldsymbol{\tau}$ during inference. However, we leverage the following insight: given \emph{fixed local structures}, the action on $C$ of changing a single torsion angle $\tau_i$ by some $\Delta \tau_i$ can be applied without choosing a definition (Figure \ref{fig:torsion}). In other words, we do not need to define a bijection between the extrinsic and intrinsic spaces but only map how actions in the intrisic space transform a point in the extrisic space. Geometrically, this action is a (signed) relative rotation of the atoms on opposite sides of the bond and can be applied directly to the atomic coordinates in 3D. The geometric intuition can be stated as follows (proven in Appendix~\ref{app:proof_update}).

\begin{proposition} \label{prop:torsion}
Let $(b_i, c_i)$ be a rotatable bond, let $\mathbf{x}_{\mathcal{V}(b_i)}$ be the positions of atoms on the $b_i$ side of the molecule, and let $R(\boldsymbol{\theta}, x_{c_i}) \in SE(3)$ be the rotation by Euler vector $\boldsymbol{\theta}$ about $x_{c_i}$. Then for $C, C' \in \mathcal{C}_G$, if $\tau_i$ is any definition of the torsion angle around bond $(b_i, c_i)$,
\begin{equation} 
    \begin{aligned}
        \tau_i(C') &= \tau_i(C) + \theta\\
        \tau_j(C') &= \tau_j(C) \quad \forall j\neq i
    \end{aligned}
    \qquad \text{if} \qquad
    \exists \mathbf{x} \in C, \mathbf{x'} \in C'\ldotp \quad
    \begin{aligned}
    \mathbf{x}'_{\mathcal{V}(b_i)} &=  \mathbf{x}_{\mathcal{V}(b_i)} \\
    \mathbf{x}'_{\mathcal{V}(c_i)} &=  R\left(\theta \, \mathbf{\hat r}_{b_ic_i}, x_{c_i} \right)\mathbf{x}_{\mathcal{V}(c_i)}
    \end{aligned}
\end{equation}
where $\mathbf{\hat r}_{b_ic_i} = (x_{c_i} - x_{b_i})/||x_{c_i}-x_{b_i}||$.
\end{proposition}
To apply a torsion update $\Delta\boldsymbol{\tau} = (\Delta\tau_1,\ldots\Delta\tau_m)$ involving all bonds, we apply $\Delta\tau_i$ sequentially in any order. Then, since training and inference only make use of torsion updates $\Delta\boldsymbol{\tau}$, we work solely in terms of 3D point clouds and updates applied to them. To draw local structures $L$ from RDKit, we draw full 3D conformers $C \in \mathcal{C}_G$ and then randomize all torsion angles to sample uniformly over $\mathbb{T}^m$. To solve the reverse SDE, we repeatedly predict torsion updates directly from, and apply them directly to, the 3D point cloud. Therefore, since our method never requires a choice of reference neighbors for any $\tau_i$, it is manifestly invariant to such a choice. These procedures are detailed in Section \ref{app:summary_procedures}.

\subsubsection{Parity equivariance} \label{sec:parity}

The torsional score framework presented thus far requires an $SE(3)$-invariant model. However, an additional symmetry requirement arises from the fact that the underlying physical energy is invariant, or extremely nearly so, under \emph{parity inversion} \cite{quack2002important}. Thus our learned density should respect $p(C) = p(-C)$ where $-C = \{-\mathbf{x} \mid \mathbf{x} \in C\}$. In terms of the conditional distribution over torsion angles, we require $p(\boldsymbol{\tau}(C) \mid L(C)) = p(\boldsymbol{\tau}(-C) \mid L(-C))$. Then (proof in Appendix \ref{app:proof_parity}),

\begin{proposition} \label{prop:parity}
    If $p(\boldsymbol{\tau}(C) \mid L(C)) = p(\boldsymbol{\tau}(-C) \mid L(-C))$, then for all diffusion times $t$,
    \begin{equation}
        \nabla_{\boldsymbol{\tau}} \log p_t(\boldsymbol{\tau}(C) \mid L(C)) = -\nabla_{\boldsymbol{\tau}} \log p_t(\boldsymbol{\tau}(-C) \mid L(-C)) 
    \end{equation}
\end{proposition}

Because the score model seeks to learn $\mathbf{s}_G(C, t) = \nabla_{\boldsymbol{\tau}} \log p_t(\boldsymbol{\tau}(C) \mid L(C))$, we must have $\mathbf{s}_G(C, t) = -\mathbf{s}_G(-C, t)$. Thus, the score model must be \emph{invariant} under $SE(3)$ but \emph{equivariant} (change sign) under parity inversion of the input point cloud--- i.e. it must output a set of \emph{pseudoscalars} in $\mathbb{R}^m$.

\subsubsection{Score network architecture} \label{sec:score_model}

Based on the previous discussion, the desiderata for the score model are:
\begin{center}
    \emph{Predict a pseudoscalar $\delta\tau_i := \partial \log p / \partial \tau_i \in\mathbb{R}$ that is $SE(3)$-invariant and parity equivariant for every rotatable bond in a 3D point cloud representation of a conformer.\\}
\end{center}

While there exist several GNN architectures which are $SE(3)$-equivariant \cite{jing2020learning, satorras2021n}, their $SE(3)$-invariant outputs are also parity invariant and, therefore, cannot satisfy the desired symmetry. Instead, we leverage the ability of equivariant networks based on tensor products \cite{thomas2018tensor, e3nn} to produce pseudoscalar outputs.

Our architecture, detailed in Appendix \ref{app:architecture}, consists of an embedding layer, a series of atomic convolution layers, and a final bond convolution layer. The first two closely follow the architecture of Tensor Field Networks \cite{thomas2018tensor}, and produce learned feature vectors for each atom. The final bond convolution layer constructs tensor product filters spatially centered on every rotatable bond and aggregates messages from neighboring atom features. We extract the pseudoscalar outputs of this filter to produce a single real-valued pseudoscalar prediction $\delta\tau_i$ for each rotatable bond. 

Naively, the bond convolution layer could be constructed the same way as the atomic convolution layers, i.e., with spherical harmonic filters. However, to supply information about the orientation of the bond about which the torsion occurs, we construct a filter from the product of the spherical harmonics with a representation of the bond (Figure \ref{fig:overview}D). Because the convolution conceptually resembles computing the torque, we call this final layer the \emph{pseudotorque} layer.

\subsection{Training and inference procedures} \label{app:summary_procedures}

Algorithms \ref{alg:training} and \ref{alg:inference} summarize, respectively, the training and inference procedures (without low-temperature sampling) used for torsional diffusion. In practice, during training, we limit $K_G$ to 30 i.e. we only consider the first 30 conformers found by CREST (typically those with the largest Boltzmann weight). Moreover, molecules are batched and an Adam optimizer with a learning rate scheduler is used for optimization. For inference, to fairly compare with other methods from the literature, we follow \cite{ganea2021geomol} and set $K$ to be twice the number of conformers returned by CREST.

\begin{algorithm}[h]
\caption{Training procedure}\label{alg:training}
\KwIn{molecules $[G_0, ..., G_N]$ each with true conformers $[C_{G,1}, ... C_{G,K_G}]$, learning rate $\alpha$}
\KwOut{trained score model $\mathbf{s}_\theta$}
conformer matching process for each $G$ to get $[\hat{C}_{G,1}, ... \hat{C}_{G,K_G}]$\;
\For{epoch $\leftarrow 1$ \KwTo $\text{epoch}_{\max}$}{
    \For{$G$ \textbf{in} $[G_0, ..., G_N]$}{
        sample $t \in [0, 1]$ and $\hat{C} \in [\hat{C}_{G,1}, ... \hat{C}_{G,K_G}]$\;
        sample $\Delta \boldsymbol{\tau}$ from wrapped normal $p_{t\mid 0}(\cdot \mid \mathbf{0})$ with $\sigma = \sigma_{\min} ^{1-t} \, \sigma_{\max} ^t$\; 
        apply $\Delta \boldsymbol{\tau}$ to $\hat{C}$\;
        predict $\delta \boldsymbol{\tau} = \mathbf{s}_{\theta, G}(\hat{C}, t)$\;
        update $\theta \gets \theta - \alpha \nabla_{\theta} \lVert \delta \boldsymbol{\tau} - \nabla_{\Delta\boldsymbol{\tau}} p_{t\mid 0}(\Delta\boldsymbol{\tau} \mid \mathbf{0})  \rVert^2$\;
    }
}
\end{algorithm}

\begin{algorithm}[h]
\caption{Inference procedure}\label{alg:inference}
\KwIn{molecular graph $G$, number conformers $K$, number steps $N$}
\KwOut{predicted conformers  $[C_1, ... C_K]$}
generate local structures by obtaining conformers $[C_1, ... C_K]$ from RDKit\;
\For{$C$ \textbf{in} $[C_1, ... C_K]$}{
    sample $\Delta \boldsymbol{\tau} \sim U[0, 2\pi]^m$ and apply to $C$ to randomize torsion angles\;
    \For{n $\leftarrow N$ \KwTo $1$}{
        let $t = n/N, \; g(t) = \sigma_{\min} ^{1-t} \, \sigma_{\max} ^t\sqrt{2\ln(\sigma_{\max}/\sigma_{\min})}$\;
        predict $\delta \boldsymbol{\tau} = \mathbf{s}_{\theta, G}(\hat{C}, t)$\;
        draw $\mathbf{z}$ from wrapped normal with $\sigma^2 = 1/N$\;
        set $\Delta \boldsymbol{\tau} = (g^2(t)/N)\;\delta\boldsymbol{\tau} + g(t)\;\mathbf{z}$\;
        apply $\Delta\boldsymbol{\tau}$ to $C$\;
    }
}
\end{algorithm}

\section{Experiments} 

We evaluate torsional diffusion by comparing the generated and ground-truth conformers in terms of ensemble RMSD (Section \ref{sec:ensemble_quality}) and properties (Section \ref{sec:ensemble_prop}). Code to run and replicate the presented results and links to the datasets discussed are available at \url{https://github.com/gcorso/torsional-diffusion}.

\subsection{Experimental setup} \label{sec:exp_setup}

\paragraph{Dataset.} We evaluate on the GEOM dataset \cite{axelrod2020geom}, which provides gold-standard conformer ensembles generated with metadynamics in CREST \cite{pracht2020automated}. We focus on GEOM-DRUGS---the largest and most pharmaceutically relevant part of the dataset---consisting of 304k drug-like molecules (average 44 atoms). To test the capacity to extrapolate to the largest molecules, we also collect from GEOM-MoleculeNet all species with more than 100 atoms into a dataset we call GEOM-XL and use it to evaluate models trained on DRUGS. Finally, we train and evaluate models on GEOM-QM9, a more established dataset but with significantly smaller molecules (average 11 atoms). Results for GEOM-QM9 are in Appendix \ref{app:torsional_results}.

\paragraph{Evaluation.} We use the train/val/test splits from \cite{ganea2021geomol} and use the same metrics to compare the generated and ground truth conformer ensembles: Average Minimum RMSD (AMR) and Coverage. These metrics are reported both for Recall (R)---which measures how well the generated ensemble covers the ground-truth ensemble---and Precision (P)---which measures the accuracy of the generated conformers. See Appendix \ref{app:torsional_exp_details} for exact definitions and further details. Following the literature, we generate $2K$ conformers for a molecule with $K$ ground truth conformers.

\paragraph{Baselines.} We compare with the strongest existing methods from Section \ref{sec:background}. Among cheminformatics methods, we evaluate RDKit ETKDG \cite{riniker2015better}, the most established open-source package, and OMEGA \cite{hawkins2010conformer, hawkins2012conformer}, a commercial software in continuous development. Among machine learning methods, we evaluate GeoMol \cite{ganea2021geomol} and GeoDiff \cite{xu2021geodiff}, which have outperformed all previous models on the evaluation metrics. Note that GeoDiff originally used a small subset of the DRUGS dataset, so we retrained it using the splits from \cite{ganea2021geomol}.

\begin{table}[t]
\caption{Quality of generated conformer ensembles for the GEOM-DRUGS test set in terms of Coverage (\%) and Average Minimum RMSD (\AA). We compute Coverage with a threshold of $\delta=0.75$~\AA\ to better distinguish top methods. Note that this is different from most prior works, which used $\delta=1.25$ \AA.}\label{tab:quality}
    \vspace{5pt}
\centering
\begin{tabular}{l|cccc|cccc} \toprule
                & \multicolumn{4}{c|}{Recall} & \multicolumn{4}{c}{Precision}  \\
                  & \multicolumn{2}{c}{Coverage $\uparrow$} & \multicolumn{2}{c|}{AMR $\downarrow$} & \multicolumn{2}{c}{Coverage $\uparrow$} & \multicolumn{2}{c}{AMR $\downarrow$} \\
Method & Mean & Med & Mean & Med & Mean & Med & Mean & Med \\ \midrule
RDKit ETKDG & 38.4 & 28.6 & 1.058 & 1.002 & 40.9 & 30.8 & 0.995 & 0.895 \\
OMEGA & 53.4 & 54.6 & 0.841 & 0.762 & 40.5 & 33.3 & 0.946 & 0.854 \\

GeoMol & 44.6 & 41.4 & 0.875 & 0.834 & 43.0 & 36.4 & 0.928 & 0.841 \\
GeoDiff & 42.1 & 37.8 & 0.835 & 0.809 & 24.9 & 14.5 & 1.136 & 1.090 \\ \midrule
Torsional Diffusion & 72.7 & \textbf{80.0} & 0.582 & 0.565 & 55.2 & 56.9 & 0.778 & 0.729     \\ 
TD w/ low temp. & \textbf{73.3} & 77.7 & \textbf{0.570} & \textbf{0.551} & \textbf{66.4} & \textbf{73.8} & \textbf{0.671} & \textbf{0.613}     \\ \bottomrule
\end{tabular}
\end{table}

\subsection{Ensemble RMSD} \label{sec:ensemble_quality}

Torsional diffusion significantly outperforms all previous methods on GEOM-DRUGS (Table \ref{tab:quality} and Figure \ref{fig:coverage}), reducing by 32\% the average minimum recall RMSD and by 28\% the precision RMSD relative to the previous state-of-the-art method. Torsional diffusion is also the first ML method to consistently generate better ensembles than OMEGA. As OMEGA is a well-established product used in industry, this represents an essential step towards establishing the utility of conformer generation with machine learning. 

\begin{figure}[t]
    \caption{Mean coverage for recall (\emph{left}) and precision (\emph{right}) when varying the threshold value $\delta$ on GEOM-DRUGS.}\label{fig:coverage}
    \centering
    \vspace{5pt}
    \includegraphics[width=0.495\textwidth]{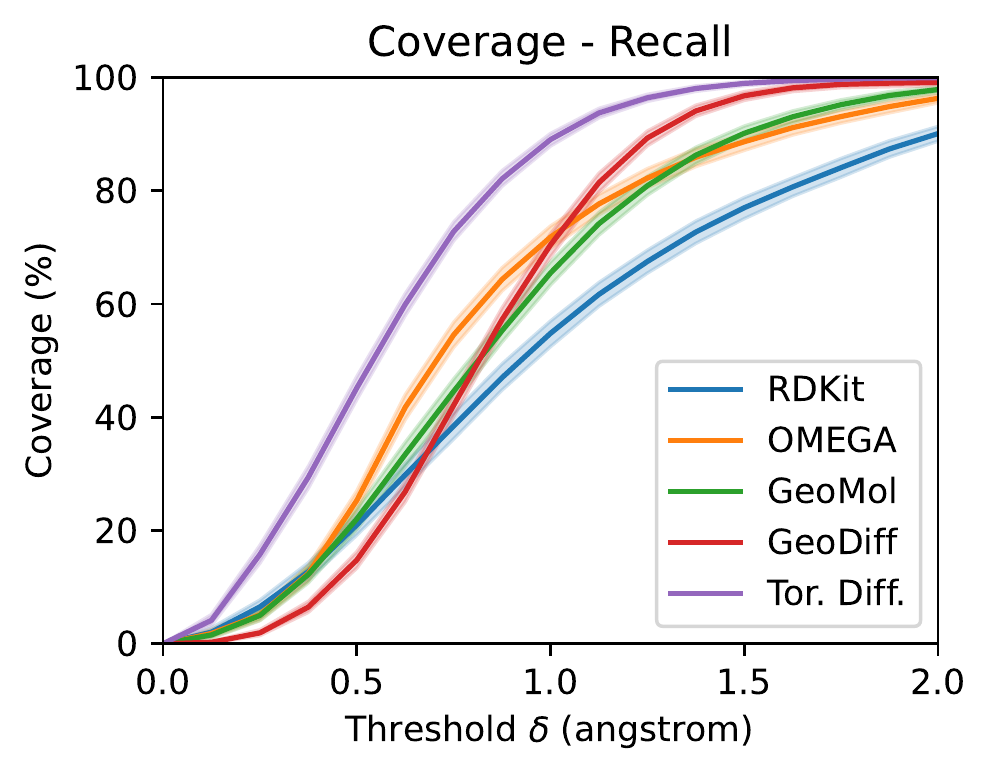}
    \includegraphics[width=0.495\textwidth]{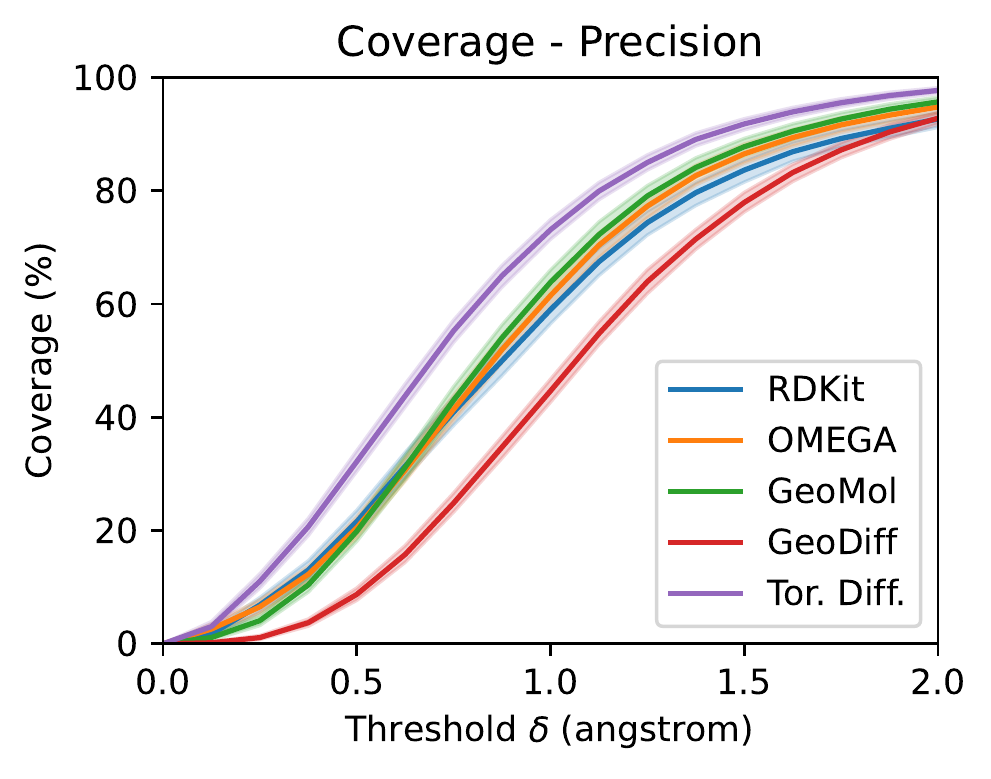}
\end{figure}

Torsional diffusion offers specific advantages over both GeoDiff and GeoMol, the most advanced prior machine learning methods. GeoDiff, a Euclidean diffusion model, requires 5000 denoising steps to obtain the results shown, whereas our model---thanks to the reduced degrees of freedom---requires only 20 steps. In fact, our model outperforms GeoDiff with as few as 5 denoising steps. As seen in Table \ref{tab:runtime}, this translates to enormous runtime improvements. 

Compared to torsional diffusion, GeoMol similarly makes use of intrinsic coordinates. However, since GeoMol can only access the molecular graph, it is less suited for reasoning about relationships that emerge only in a spatial embedding, especially between regions of the molecule that are distant on the graph. Our extrinsic-to-intrinsic score framework---which gives direct access to spatial relationships---addresses precisely this issue. The empirical advantages are most evident for the large molecules in GEOM-XL, on which GeoMol fails to improve consistently over RDKit. On the other hand, because GeoMol requires only a single-forward pass, it retains the advantage of faster runtime compared to diffusion-based methods.

\begin{table}[t]
    \caption{Median AMR and runtime (core-secs per conformer) of machine learning methods, evaluated on CPU for comparison with RDKit.} \label{tab:runtime}
    \centering
    \vspace{5pt}
    \begin{tabular}{lcccc}\toprule  
    Method & Steps &  AMR-R & AMR-P & Runtime \\\midrule
    RDKit & -  & 1.002 & 0.895 &    \textbf{0.10}       \\
    GeoMol & - & 0.834 & 0.841 & 0.18 \\
    GeoDiff & 5000 & 0.809 & 1.090 & 305           \\ \midrule
    \multirow{3}{*}{\makecell{Torsional\\Diffusion}} 
    & 5       & 0.685 & 0.963 & 1.76          \\
    & 10      & 0.580 & 0.791 & 2.82          \\
    & 20      & \textbf{0.565} & \textbf{0.729} & 4.90          \\ \bottomrule
    \end{tabular}
\end{table}

\paragraph{Performance vs size.} Figure \ref{fig:perf_rotatable_bonds} shows the performance of different models as a function of the number of rotatable bonds. Molecules with more rotatable bonds are more flexible and are generally larger; it is therefore expected that the RMSD error will increase with the number of bonds. With very few rotatable bonds, the error of torsional diffusion depends mostly on the quality of the local structures it was given, and therefore it has a similar error as RDKit. However, as the number of torsion angles increases, torsional diffusion deteriorates more slowly than other methods. 

The trend continues with the very large molecules in GEOM-XL (average 136 atoms and 32 rotatable bonds). These not only are larger and more flexible, but---for machine learning models trained on GEOM-DRUGS---are also out of distribution. As shown in Table \ref{tab:results_xl}, on GEOM-XL GeoMol only performs marginally better than RDKit, while torsional diffusion reduces RDKit AMR by 30\% on recall and 12\% on precision. These results can very likely be improved by training and tuning the torsional diffusion model on larger molecules.

\begin{figure}[h!]
    \centering
    \includegraphics[width=0.49\textwidth]{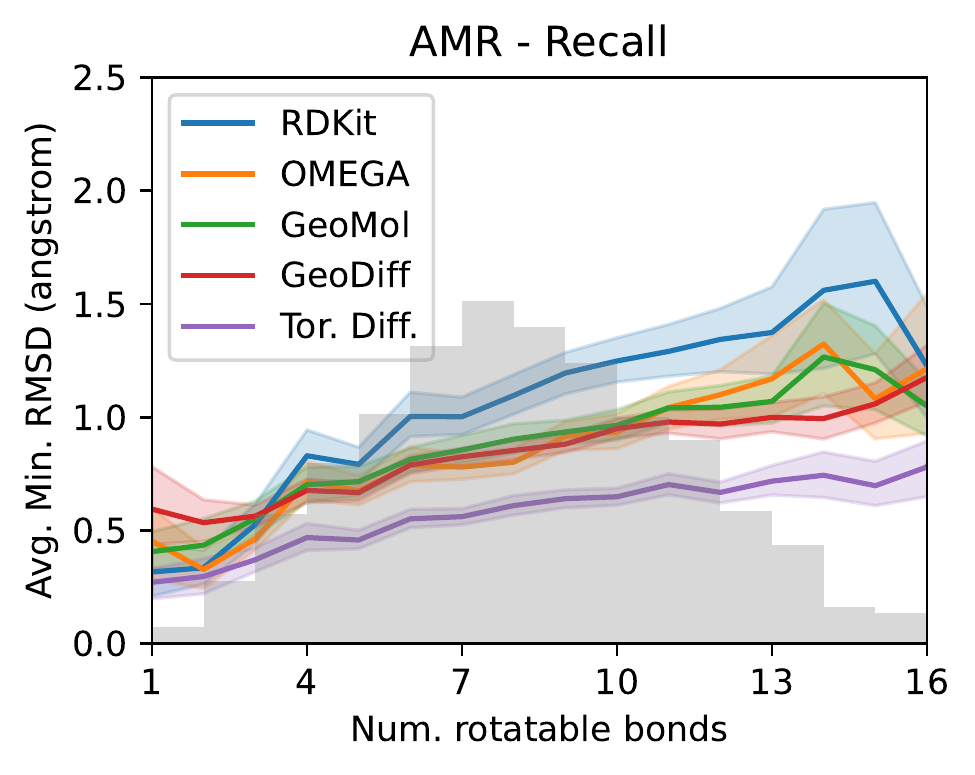}
    \includegraphics[width=0.49\textwidth]{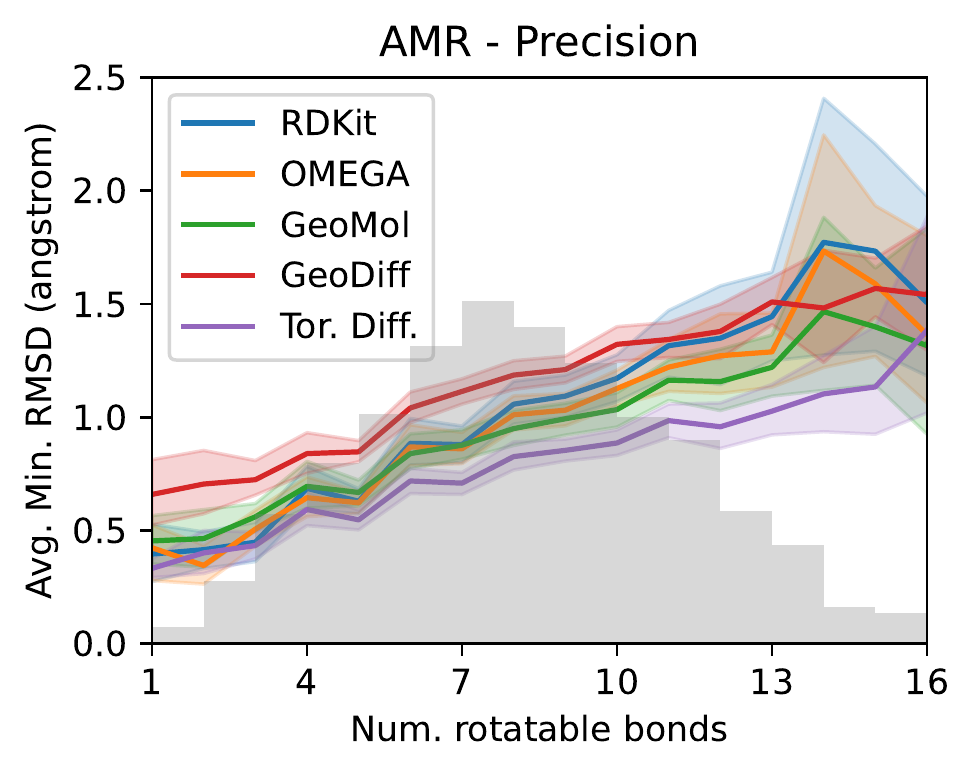}
    \caption{Average minimum RMSD (AMR) for recall (\emph{left}) and precision (\emph{right}) of the different conformer generation methods for molecules with different number of rotatable bonds in GEOM-DRUGS. The background shows the distribution of the number of rotatable bonds.}
    \label{fig:perf_rotatable_bonds}
\end{figure}

\begin{table}[h!]
\caption{Performance of various methods on the GEOM-XL dataset. } \label{tab:results_xl}
\vspace{5pt}
\centering
\begin{tabular}{lcccc} \toprule
                  & \multicolumn{2}{c}{AMR-R $\downarrow$} & \multicolumn{2}{c}{AMR-P $\downarrow$} \\
Model  & Mean      & Med         & Mean        & Med         \\ \midrule
RDKit               & 2.92          & 2.62          & 3.35          & 3.15          \\
GeoMol              & 2.47          & 2.39          & 3.30          & 3.15          \\
Torsional Diffusion & \textbf{2.05} & \textbf{1.86} & \textbf{2.94} & \textbf{2.78}             \\ \bottomrule
\end{tabular}
\end{table}

\subsection{Ensemble properties} \label{sec:ensemble_prop}

While RMSD gives a \emph{geometric} way to evaluate ensemble quality, we also consider the \emph{chemical} similarity between generated and ground truth ensembles. For a random 100-molecule subset of DRUGS, we generate $\min(2K, 32)$ conformers per molecule, relax the conformers with GFN2-xTB \cite{bannwarth2019gfn2},\footnote{Results without relaxation (which are less chemically meaningful) are in Appendix \ref{app:torsional_results}.} and compare the Boltzmann-weighted properties of the generated and ground truth ensembles. Specifically, the following properties are computed with xTB \cite{bannwarth2019gfn2}: energy $E$, dipole moment $\mu$, HOMO-LUMO gap $\Delta \epsilon$, and the minimum energy $E_{\min}$. The median errors for torsional diffusion and the baselines are shown in Table \ref{tab:boltzmann}. Our method produces the most chemically accurate ensembles, especially in terms of energy. In particular, we significantly improve over GeoMol and GeoDiff in finding the lowest-energy conformers that are only (on median) 0.13 kcal/mol higher in energy than the global minimum.

\begin{table}[t]
    \caption{Median absolute error of generated v.s. ground truth ensemble properties. $E, \Delta\epsilon, E_{\min}$ in kcal/mol, $\mu$ in debye.}\label{tab:properties}
    \centering
    \vspace{5pt}
    \begin{tabular}{lcccc}
    \toprule  
    Method & $E$ & $\mu$ & $\Delta \epsilon$ & $E_{\min}$ \\\midrule
    RDKit & 0.81 & 0.52 & 0.75 & 1.16 \\
    OMEGA & 0.68 & 0.66 & 0.68 & 0.69 \\
    GeoMol & 0.42 & \textbf{0.34} & 0.59 & 0.40 \\
    GeoDiff & 0.31 & {0.35} & 0.89 & 0.39 \\
    Tor. Diff. & \textbf{0.22} & {0.35} & \textbf{0.54} & \textbf{0.13} \\ \bottomrule
    \end{tabular}
\end{table}

\section{Torsional Boltzmann Generators}

Diffusion models offer a way of extracting the exact likelihood under the model of the generated datapoints. We exploit this property to train our diffusion model using an energy function rather than samples alone. In Section \ref{sec:likelihood}, we present a way of converting likelihoods on the intrinsic space to likelihoods on the extrinsic one. Then, in Section \ref{sec:energy}, we use the likelihoods to derive a novel training scheme based on importance sampling. Finally, in Section \ref{sec:boltzmann}, we present experimental evidence that, on unseen molecules, our method is more efficient at sampling the conditional Boltzmann distribution than annealed importance sampling (AIS).

\subsection{Likelihood} \label{sec:likelihood}

By using the probability flow ODE, we can compute the likelihood of any sample $\boldsymbol{\tau}$ as follows \cite{song2021score,de2022riemannian}:
\begin{equation} \label{eq:likelihood}
    \log p_0(\boldsymbol{\tau}_0) = \log p_T(\boldsymbol{\tau}_T) - \frac{1}{2} \int_0^T g^2(t) \; \nabla_{\boldsymbol{\tau}} \cdot \mathbf{s}_G(\boldsymbol{\tau}_t, t) \; dt
\end{equation}
In \cite{song2021score}, the divergence term is approximated via Hutchinson's method \cite{hutchinson1989stochastic}, which gives an unbiased estimate of $\log p_0(\boldsymbol{\tau})$. However, this gives a \emph{biased} estimate of $ p_0(\boldsymbol{\tau})$, which is unsuitable for our applications. Thus, we compute the divergence term directly, which is feasible here (unlike in Euclidean diffusion) due to the reduced dimensionality of the torsional space.

The above likelihood is in \textit{torsional} space $p_G(\boldsymbol{\tau} \mid L), \boldsymbol{\tau} \in \mathbb{T}^m$, but to enable compatibility with the Boltzmann measure $e^{-E(\mathbf{x})/kT}$, it is desirable to interconvert this with a likelihood in \textit{Euclidean} space $p(\mathbf{x} \mid L), \mathbf{x} \in \mathbb{R}^{3n}$. A factor is necessary to convert between the volume element in torsional space and in Euclidean space (full derivation in Appendix~\ref{app:proof_likelihood}):

\begin{proposition} \label{prop:euclidean}
Let $\mathbf{x} \in C(\boldsymbol{\tau}, L)$ be a centered\footnote{Additional formalism is needed for translations, but it is independent of the conformer and can be ignored.}
conformer in Euclidean space. Then,
\begin{equation}
    p_G(\mathbf{x} \mid L) = \frac{p_G(\boldsymbol{\tau} \mid L)}{ 8 \pi^2 \sqrt{\det g}}
    \quad \mathrm{where} \ \
    g_{\alpha\beta} = 
    \sum_{k=1}^{n} 
    J^{(k)}_{\alpha} \cdot J^{(k)}_{\beta}
\end{equation}
where the indices $\alpha,\beta$ are integers between 1 and $m+3$. For $1 \leq \alpha \leq m$, $J^{(k)}_\alpha$ is defined as
    \begin{align}
    \label{eqn:basisvec}
        J^{(k)}_{i} &= \tilde J^{(k)}_{i} - \frac 1 n \sum_{\ell=1}^{n} \tilde J^{(\ell)}_{i}
        \quad
        \mathrm{with} \ \
        \tilde J^{(\ell)}_{i} =
        \begin{cases}
            0 & \ell \in \mathcal{V}(b_i), \\
            \frac{\mathbf{x}_{b_i} - \mathbf{x}_{c_i}} {||\mathbf{x}_{b_i} - \mathbf{x}_{c_i}||}
            \times
            \left( \mathbf{x}_\ell - \mathbf{x}_{c_i} \right),
            & \ell \in \mathcal{V}(c_i),
        \end{cases}
    \end{align}
    and for $\alpha \in \{m+1, m+2, m+3\}$ as
    \begin{align}
    \label{eq:omegajacobian}
        J^{(k)}_{m+1} &= \mathbf{x}_k \times \hat{x},
        \qquad
        J^{(k)}_{m+2} = \mathbf{x}_k \times \hat{y},
        \qquad
        J^{(k)}_{m+3} = \mathbf{x}_k \times \hat{z},
        \qquad
    \end{align}
    where $(b_i, c_i)$ is the freely rotatable bond for torsion angle $i$, $\mathcal{V}(b_i)$ is the set of all nodes on the same side of the bond as $b_i$, and $\hat x, \hat y, \hat z$ are the unit vectors in the respective directions.
\end{proposition}

\subsection{Energy-based training} \label{sec:energy}

By computing likelihoods, we can train torsional diffusion models to match the Boltzmann distribution over torsion angles using the energy function. At a high level, we minimize the usual score matching loss, but with simulated samples from the Boltzmann distribution rather than data samples. The procedure therefore consists of two stages: resampling and score matching, which are tightly coupled during training (Algorithm 1). In the \emph{resampling} stage, we use the model as an importance sampler for the Boltzmann distribution, where Proposition \ref{prop:euclidean} is used to compute the (unnormalized) torsional Boltzmann density $\tilde{p}_G(\boldsymbol{\tau}\mid L)$. In the \emph{score-matching} stage, the importance weights are used to approximate the denoising score-matching loss with expectations taken over $\tilde{p}_G(\boldsymbol{\tau}\mid L)$. As the model learns the score, it improves as an importance sampler. 

{
\begin{small}
\vspace{10pt}
\begin{algorithm}[H]
\caption{Energy-based training epoch}\label{alg:energy}
\KwIn{Boltzmann density $\tilde{p}$, training pairs $\{(G_i, L_i)\}_i$, torsional diffusion model $q$}
\For{\textbf{each} $(G_i, L_i)$}{
    Sample $\boldsymbol{\tau}_1, \ldots \boldsymbol{\tau}_K \sim q_{G_i}(\boldsymbol{\tau}\mid L_i)$\;
    \For{$k\leftarrow 1$ \KwTo $K$}{
        $\tilde{w}_k = \tilde{p}_{G_i}(\boldsymbol{\tau}_k\mid L_i)/q_{G_i}(\boldsymbol{\tau}_k\mid L_i)$\;
    }
    Approximate $J_\text{DSM}$ for $p_0 \propto \tilde{p}$ using $\{(\tilde{w}_i, \boldsymbol{\tau}_i)\}_i$\;
    Minimize $J_\text{DSM}$\;
}
\end{algorithm}
\end{small}
\vspace{10pt}
}

This training procedure differs substantially from that of existing Boltzmann generators, which are trained as flows with a loss that directly depends on the model density. In contrast, we \emph{train} the model as a score-based model, but \emph{use} it as a flow---both during training and inference---to generate samples. The model density is needed only to reweight the samples to approximate the target density. Since in principle the model used for resampling does not need to be the same as the model being trained,\footnote{For example, if the resampler were perfect, the procedure would reduce to normal denoising score matching.} we can use very few steps (a shallow flow) during resampling to accelerate training, and then increase the number of steps (a deeper flow) for better approximations during inference---an option unavailable to existing Boltzmann generators.

\subsection{Torsional Boltzmann generator} \label{sec:boltzmann}

\begin{table}
    \caption{Effective sample size (out of 32) given by importance sampling weights over the torsional Boltzmann density.}\label{tab:boltzmann}
    \vspace{5pt}
    \centering
    \begin{tabular}{lcccc} \toprule
    & & \multicolumn{3}{c}{Temp. (K)} \\ \cmidrule(lr){3-5}
    Method & Steps & 1000 & 500 & 300 \\ \midrule
    Uniform & -- & 1.71 & 1.21 & 1.02 \\ \midrule
    \multirow{3}{*}{AIS} 
    & 5 & 2.20 & 1.36 & 1.18 \\
    & 20 & 3.12 & 1.76 & 1.30  \\ 
    & 100 & 6.72 & 3.12 & 2.06 \\ \midrule
    \multirow{2}{*}{\makecell[l]{Torsional\\BG}} 
    & 5 & 7.28 & 3.60 & 3.04 \\
    & 20 & \textbf{11.42} & \textbf{6.42} & \textbf{4.68} \\ \bottomrule
    \end{tabular}
\end{table}

We evaluate how well a torsional Boltzmann generator trained with MMFF \cite{halgren1996merck} energies can sample the corresponding Boltzmann density over torsion angles. We train and test on GEOM-DRUGS molecules with 3--7 rotatable bonds and use the local structures of the first ground-truth conformers. For the baselines, we implement annealed importance samplers (AIS) \cite{neal2001annealed} with Metropolis-Hastings steps over the torsional space and tune the variance of the transition kernels.

Table \ref{tab:boltzmann} shows the quality of the samplers in terms of the \emph{effective sample size} (ESS) given by the weights of 32 samples for each test molecule, which measures the $\alpha$-divergence (with $\alpha=2$) between the model and Boltzmann distributions \cite{midgley2021bootstrap}. Our method significantly outperforms the AIS baseline, and improves with increased step size despite being trained with only a 5-step resampler. Note that, since these evaluations are done on \emph{unseen} molecules, they are beyond the capabilities of existing Boltzmann generators.

Figure \ref{fig:hist_bg} shows the distributions of ESSs at 500K for the torsional Boltzmann generator and the AIS baseline. While AIS fails to generate more than one effective sample for most molecules (tall leftmost column), torsional Boltzmann generators are much more efficient, with more than five effective samples for a significant fraction of molecules.

\begin{figure}[h!]
    \centering
    \includegraphics[width=0.49\textwidth]{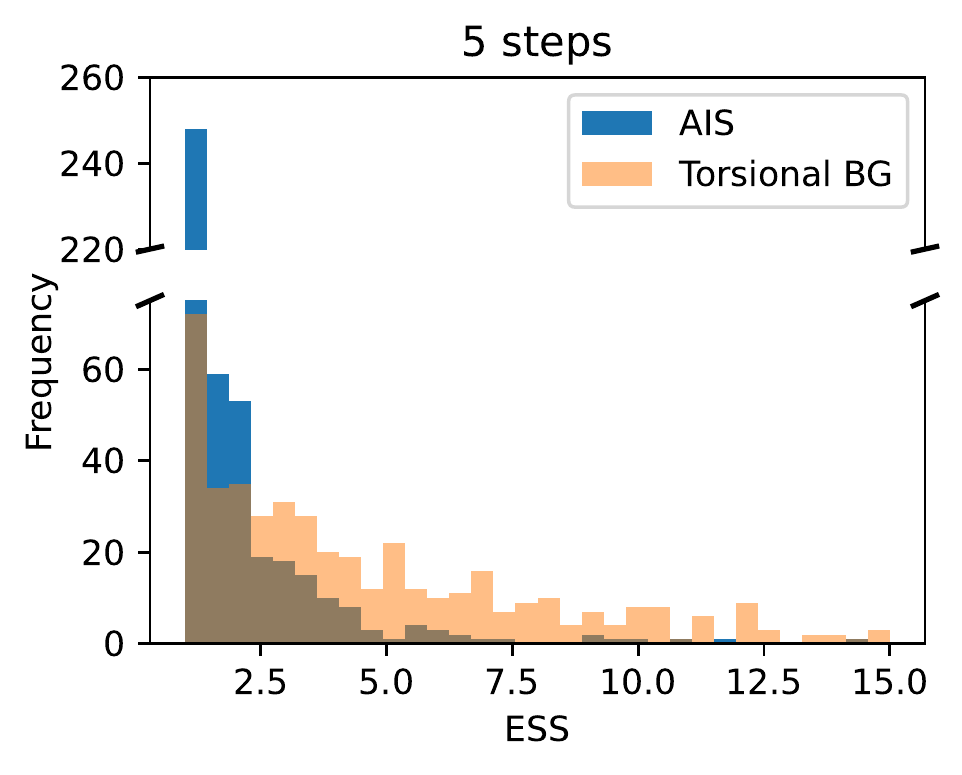}
    \includegraphics[width=0.49\textwidth]{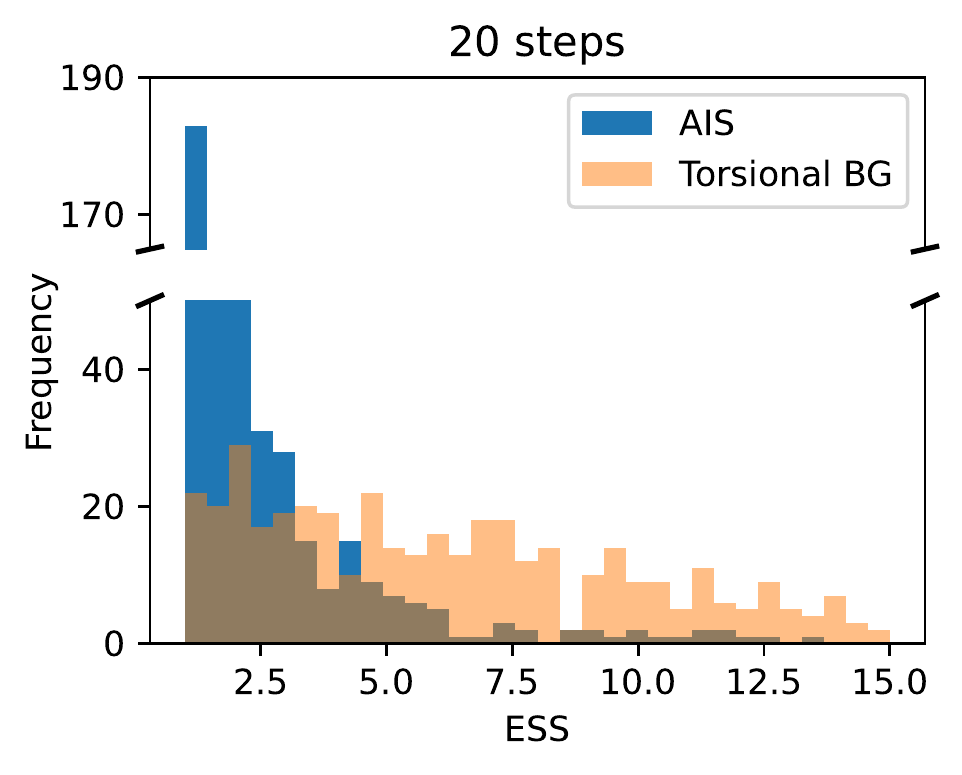}
    \caption{Histogram of the ESSs of the torsional Boltzmann generator and AIS baseline at 500K.}
    \label{fig:hist_bg}
\end{figure}

\chapter{DiffDock} \label{chapter:diffdock}

The biological functions of proteins can be modulated by small molecule ligands (such as drugs) binding to them. Thus, a crucial task in computational drug design is \emph{molecular docking}---predicting the position, orientation, and conformation of a ligand when bound to a target protein---from which the effect of the ligand (if any) might be inferred. Traditional approaches for docking \cite{trott2010autodock,halgren2004glide} rely on scoring-functions that estimate the correctness of a proposed structure or pose, and an optimization algorithm that searches for the global maximum of the scoring function. However, since the search space is vast and the landscape of the scoring functions rugged, these methods tend to be too slow and inaccurate, especially for high-throughput workflows.

Recent works \cite{equibind, Lu2022TankBind} have developed deep learning models to predict the binding pose in one shot, treating docking as a regression problem. While these methods are much faster than traditional search-based methods, they have yet to demonstrate significant improvements in accuracy. We argue that this may be because the regression-based paradigm  
corresponds imperfectly with the objectives of molecular docking, which is reflected in the fact that standard accuracy metrics resemble the \emph{likelihood} of the data under the predictive model rather than a regression loss. 
We thus frame molecular docking as a \textit{generative modeling problem}---given a ligand and target protein structure, we learn a distribution over ligand poses.

Following the intrinsic diffusion models framework, we therefore develop \textsc{DiffDock}, a diffusion generative model (DGM) over the space of ligand poses for molecular docking. We define a diffusion process over the degrees of freedom involved in docking: the position of the ligand relative to the protein (locating the binding pocket), its orientation in the pocket, and the torsion angles describing its conformation. \textsc{DiffDock} samples poses by running the learned (reverse) diffusion process, which iteratively transforms an uninformed, noisy prior distribution over ligand poses into the learned model distribution (Figure~\ref{fig:diffdock_overview}). Intuitively, this process can be viewed as the progressive refinement of random poses via updates of their translations, rotations, and torsion angles.

While DGMs have been applied to other problems in molecular machine learning \cite{xu2021geodiff,jing2022torsional,hoogeboom2022equivariant}, existing approaches are ill-suited for molecular docking, where the space of ligand poses is an $(m+6)$-dimensional submanifold $\mathcal{M} \subset \mathbb{R}^{3n}$, where $n$ and $m$ are, respectively, the number of atoms and torsion angles. To develop \textsc{DiffDock}, we recognize that the docking degrees of freedom define $\mathcal{M}$ as the space of poses accessible via a set of allowed \emph{ligand pose transformations}. We use this idea to map elements in $\mathcal{M}$ to the product space of the groups corresponding to those transformations, where a DGM can be developed and trained efficiently.

As applications of docking models often require only a fixed number of predictions and a confidence score over these, we train a \emph{confidence model} to provide confidence estimates for the poses sampled from the DGM and to pick out the most likely sample. This two-step process can be viewed as an intermediate approach between brute-force search and one-shot prediction: we retain the ability to consider and compare multiple poses without incurring the difficulties of high-dimensional search.

Empirically, on the standard blind docking benchmark PDBBind, \textsc{DiffDock} achieves 38\% of top-1 predictions with ligand root mean square distance (RMSD) below 2\AA{}, nearly doubling the performance of the previous state-of-the-art deep learning model (20\%). \textsc{DiffDock} significantly outperforms even state-of-the-art search-based methods (23\%), while still being 3 to 12 times faster on GPU. Moreover, it provides an accurate confidence score of its predictions, obtaining 83\% RMSD$<$2\AA{} on its most confident third of the previously unseen complexes.

We further evaluate the methods on structures generated by ESMFold \cite{Lin2022ESM2}. Our results confirm previous analyses \cite{wong2022benchmarking} that showed that existing methods are not capable of docking against these approximate apo-structures (RMSD$<$2\AA{} equal or below 10\%). Instead, without further training, \textsc{DiffDock} places 22\% of its top-1 predictions within 2\AA{} opening the way for the revolution brought by accurate protein folding methods in the modeling of protein-ligand interactions.

This chapter is based on the paper:
\begin{displayquote}
\textbf{DiffDock: Diffusion Steps, Twists, and Turns for Molecular Docking.} Gabriele Corso*, Hannes Stärk*, Bowen Jing*, Regina Barzilay, and Tommi Jaakkola.  11th International Conference on Learning Representations (ICLR 2023).
\end{displayquote}

\begin{figure}[t]
    \centering
    \includegraphics[width=\textwidth]{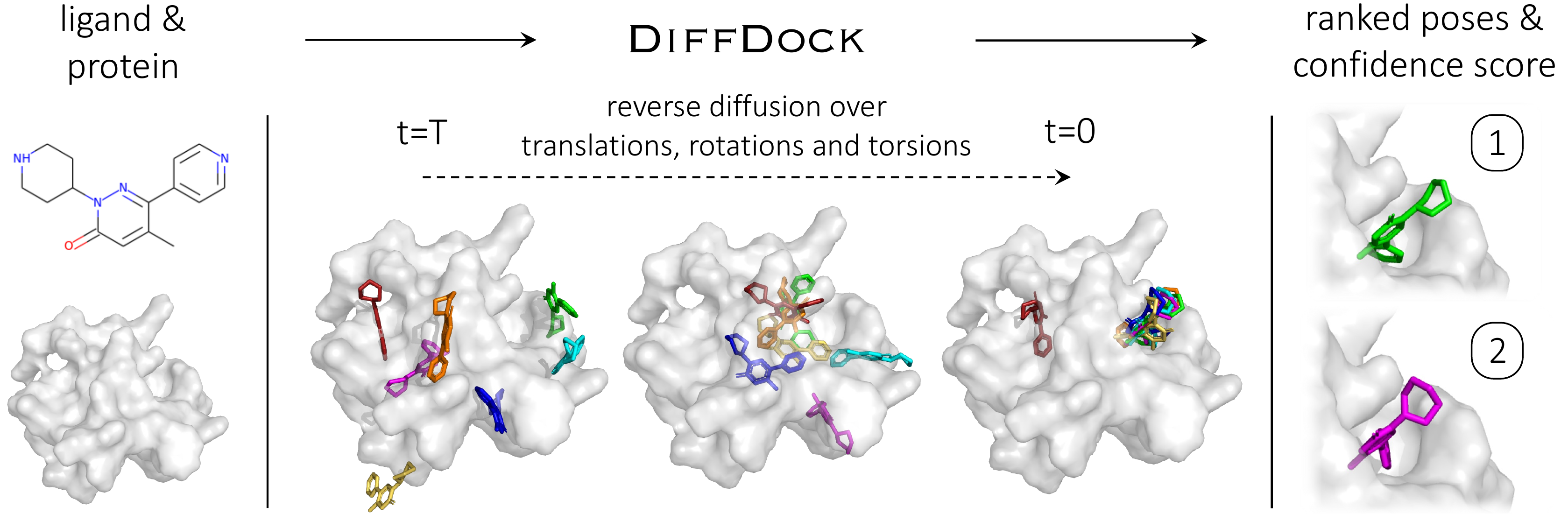}
    \caption{Overview of \textsc{DiffDock}. \emph{Left}: The model takes as input the separate ligand and protein structures. \emph{Center}: Randomly sampled initial poses are denoised via a reverse diffusion over translational, rotational, and torsional degrees of freedom. \emph{Right:}. The sampled poses are ranked by the confidence model to produce a final prediction and confidence score.}
    \label{fig:diffdock_overview}
\end{figure}

\section{Background and Related Work} \label{sec:diffdock_background}

\paragraph{Molecular docking.} The molecular docking task is usually divided between known-pocket and blind docking. Known-pocket docking algorithms receive as input the position on the protein where the molecule will bind (the \emph{binding pocket}) and only have to find the correct orientation and conformation. Blind docking instead does not assume any prior knowledge about the binding pocket; in this work, we will focus on this general setting. Due to the relative rigidity of the protein, docking methods typically assume the knowledge of the bound protein structure \cite{pagadala2017software}, this assumption is however not always realistic therefore we evaluate methods both with and without access to exact bound structure. Methods are normally evaluated by the percentage of hits, or approximately correct predictions, commonly considered to be those where the ligand RMSD error is below 2\AA{} \cite{Alhossary2015QuickVina2, Hassan2017QVinaW, mcnutt2021gnina}.

\paragraph{Search-based docking methods.} Traditional docking methods \cite{trott2010autodock,halgren2004glide,thomsen2006moldock} consist of a parameterized physics-based scoring function and a search algorithm. The scoring-function takes in 3D structures and returns an estimate of the quality/likelihood of the given pose, while the search stochastically modifies the ligand pose (position, orientation, and torsion angles) with the goal of finding the global optimum of the scoring function. Recently, machine learning has been applied to parameterize the scoring-function \cite{mcnutt2021gnina, mendez2021deepdock}. These search-based methods have offered relative improvements when docking to a known pocket but are typically very computationally expensive to run and must still grapple with the very large search space that characterizes blind docking. 

\paragraph{Machine learning for blind docking.} Recently, EquiBind \cite{equibind} has tried to tackle the blind docking task by directly predicting pocket keypoints on both ligand and protein and aligning them. TANKBind \cite{Lu2022TankBind} improved over this by independently predicting a docking pose (in the form of an interatomic distance matrix) for each possible pocket and then ranking them. Although these one-shot or few-shot regression-based prediction methods are orders of magnitude faster, their performance has not yet reached that of traditional search-based methods.

\section{Docking as Generative Modeling}\label{sec:generative_modeling}

Although EquiBind and other ML methods have provided strong runtime improvements by avoiding an expensive optimization process over ligand poses, their performance has not yet reached that of search-based methods. As our analysis below argues, this may be caused by the models' uncertainty and the optimization of an objective function that does not correspond to how molecular docking is used and evaluated in practice. 

\paragraph{Molecular docking objective.} Molecular docking plays a critical role in drug discovery because the prediction of the 3D structure of a bound protein-ligand complex enables further computational and human expert analyses on the strength and properties of the binding interaction. Therefore, a docked prediction is only useful if its deviation from the true structure does not significantly affect the output of such analyses. Concretely, a prediction is considered acceptable when the distance between the structures (measured in terms of ligand RMSD) is below some small tolerance on the order of the length scale of atomic interactions (a few \AA{}ngstr\"om). Consequently, the standard evaluation metric used in the field has been the percentage of predictions with a ligand RMSD (to the crystal ligand pose) below some value $\epsilon$. 

However, the objective of maximizing the proportion of predictions with RMSD within some tolerance $\epsilon$ is not differentiable and cannot be used for training with stochastic gradient descent. Instead, maximizing the expected proportion of predictions with RMSD $<\epsilon$ corresponds to maximizing the likelihood of the true structure under the model's output distribution, in the limit as $\epsilon$ goes to 0. This observation motivates training a generative model to minimize an upper bound on the negative log-likelihood of the observed structures under the model's distribution. Thus, we view molecular docking as the problem of learning a distribution over ligand poses conditioned on the protein structure and develop a diffusion generative model over this space (Section~\ref{sec:diffusion_model}).

\begin{figure}[t]
    \centering
    \includegraphics[width=\textwidth]{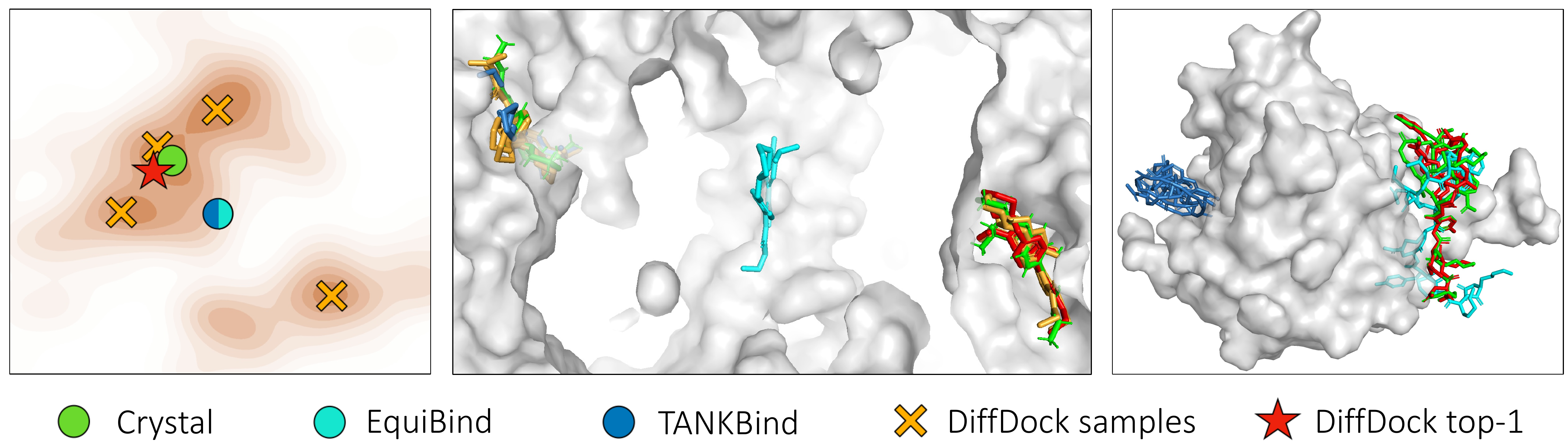}
    \caption{``\textsc{DiffDock} top-1" refers to the sample with the highest confidence. ``\textsc{DiffDock} samples" to the other diffusion model samples. \textit{Left:} Visual diagram of the advantage of generative models over regression models. Given uncertainty in the correct pose (represented by the orange distribution), regression models tend to predict the mean of the distribution, which may lie in a region of low density. \textit{Center:} when there is a global symmetry in the protein (aleatoric uncertainty), EquiBind places the molecule in the center while \textsc{DiffDock} is able to sample all the true poses. \textit{Right:} even in the absence of strong aleatoric uncertainty, the epistemic uncertainty causes EquiBind's prediction to have steric clashes and TANKBind's to have many self-intersections. }
    \label{fig:averaging}
\end{figure}

\paragraph{Confidence model.} With a trained diffusion model, it is possible to sample an arbitrary number of ligand poses from the posterior distribution according to the model. However, researchers are often interested in seeing only one or a small number of predicted poses and an associated confidence measure\footnote{For example, the pLDDT confidence score of AlphaFold2 \cite{jumper2021highly} has had a very significant impact in many applications \cite{necci2021critical,bennett2022improving}.} for downstream analysis. Thus, we train a confidence model over the poses sampled by the diffusion model and rank them based on its confidence that they are within the error tolerance. The top-ranked ligand pose and the associated confidence are then taken as \textsc{DiffDock}'s top-1 prediction and confidence score.

\paragraph{Problem with regression-based methods.} The difficulty with the development of deep learning models for molecular docking lies in the aleatoric (which is the data inherent uncertainty, e.g., the ligand might bind with multiple poses to the protein) and epistemic uncertainty (which arises from the complexity of the task compared with the limited model capacity and data available) on the pose. Therefore, given the available co-variate information (only protein structure and ligand identity), any method will exhibit uncertainty about the correct binding pose among many viable alternatives. Any regression-style method that is forced to select a single configuration that minimizes the expected square error would learn to predict the (weighted) mean of such alternatives. In contrast, a generative model with the same co-variate information would instead aim to capture the distribution over the alternatives, populating all/most of the significant modes even if similarly unable to distinguish the correct target. This behavior, illustrated in Figure~\ref{fig:averaging}, causes the regression-based models to produce significantly more physically implausible poses than our method. In particular, we observe frequent steric clashes (e.g., 26\% of EquiBind's predictions) and self-intersections in EquiBind's and TANKBind's predictions (Figures \ref{fig:self_intersections} and \ref{fig:random_examples}). We found no intersections in \textsc{DiffDock}'s predictions. 
Visualizations and quantitative evidence of these phenomena are in Appendix~\ref{appx:steric_clashes}.

\section{Method} \label{sec:diffusion_model}

\subsection{Flexibility}

A ligand pose is an assignment of atomic positions in $\mathbb{R}^{3}$, so in principle, we can regard a pose $\xx$ as an element in $\mathbb{R}^{3n}$, where $n$ is the number of atoms. However, this encompasses far more degrees of freedom than are relevant in molecular docking. In particular, bond lengths, angles, and small rings in the ligand are essentially rigid, such that the ligand flexibility lies almost entirely in the torsion angles at rotatable bonds. Traditional docking methods, as well as most ML ones, take as input a seed conformation $\cc \in \mathbb{R}^{3n}$ of the ligand in isolation and change only the relative position and the torsion degrees of freedom in the final bound conformation.\footnote{RDKit ETKDG is a popular method for predicting the seed conformation. Although the structures may not be predicted perfectly, the errors lie largely in the torsion angles, which are resampled anyways.} The space of ligand poses consistent with $\cc$ is, therefore, an $(m+6)$-dimensional submanifold $\mathcal{M}_\cc \subset \mathbb{R}^{3n}$, where $m$ is the number of rotatable bonds, and the six additional degrees of freedom come from rototranslations relative to the fixed protein. This defines the extrinsic manifold over which we will develop the intrinsic diffusion model, therefore, given as input a seed conformation $\cc$, we formulate molecular docking as learning a probability distribution $p_\cc(\xx \mid \yy)$ over the manifold $\mathcal{M}_\cc$, conditioned on a protein structure $\yy$.


\subsection{Mapping}

In order to make training a diffusion model over the manifold $\mathcal{M}_\cc$ efficient, we follow the IDM framework and define a one-to-one mapping to an intrinsic, ``nicer'', manifold where the diffusion kernel can be sampled directly. 

Any ligand pose consistent with a seed conformation can be reached by a combination of (1) ligand translations, (2) ligand rotations, and (3) changes to torsion angles. This suggests that given a continuous family of ligand pose transformations corresponding to the $m+6$ degrees of freedom, a distribution on $\mathcal{M}_\cc$ can be lifted to a distribution on the product space of the corresponding groups---which is itself a manifold. 

We associate translations of ligand position with the 3D translation group $\mathbb{T}(3)$, rigid rotations of the ligand with the 3D rotation group $SO(3)$, and changes in torsion angles at each rotatable bond with a copy of the 2D rotation group $SO(2)$. More formally, we define operations of each of these groups on a ligand pose $\cc \in \mathbb{R}^{3n}$. The translation $A_\text{tr}: \mathbb{T}(3) \times \mathbb{R}^{3n} \rightarrow \mathbb{R}^{3n}$ is defined straightforwardly as $A_\text{tr}(\mathbf{r}, \xlig)_i = \xlig_i + \mathbf{r}$ using the isomorphism $\mathbb{T}(3) \cong \mathbb{R}^3$ where $\xx_i \in \mathbb{R}^3$ is the position of the $i$th atom. Similarly, the rotation $A_\text{rot}: SO(3) \times \mathbb{R}^{3n} \rightarrow \mathbb{R}^{3n}$ is defined by $A_\text{rot}(R, \xlig)_i = R(\xlig_i - \bar \xlig) + \bar\xlig$ where $\bar \xlig = \frac{1}{n}\sum \xlig_i$, corresponding to rotations around the (unweighted) center of mass of the ligand.

Many valid definitions of a change in torsion angles are possible, as the torsion angle around any bond $(a_i, b_i)$ can be updated by rotating the $a_i$ side, the $b_i$ side, or both. However, we can specify changes of torsion angles to be \emph{disentangled} from rotations or translations. One way of doing so is to identify a central motif in the molecule, such as a ring, and change torsion angles in a way that keeps the motif fixed. However, this special treatment of the central motif introduces an arbitrary asymmetry into the problem and could be difficult for a score model to reason about. Thus, we instead define the operation of elements of $SO(2)^m$ such that it causes a minimal perturbation (in an RMSD sense) to the structure:\footnote{Since we do not define or use the composition of elements of $SO(2)^m$, strictly speaking, it is a product \emph{space} but not a \emph{group} and can be alternatively thought of as the torus $\mathbb{T}^m$ with an origin element.}
\begin{definition*}
Let $B_{k,\theta_k}(\xx) \in \mathbb{R}^{3n}$ be any valid torsion update by $\theta_k$ around the $k$th rotatable bond $(a_k, b_k)$. We define $A_\text{tor}: SO(2)^m \times \mathbb{R}^{3n} \rightarrow \mathbb{R}^{3n}$such that
$$A_\text{tor}(\boldsymbol{\theta}, \xx) = \rmsdalign(\xx, (B_{1,\theta_1} \circ \cdots B_{m,\theta_m})(\xx))$$
where $\boldsymbol{\theta} = (\theta_1, \ldots \theta_m)$ and
\begin{equation}
    \rmsdalign(\xx, \xx') = \argmin_{\xx^\dagger \in \{g\xx' \mid g\in SE(3)\}}\rmsd(\xx, \xx^\dagger)
\end{equation}
\end{definition*}
This means that we apply all the $m$ torsion updates in any order and then perform a global RMSD alignment with the unmodified pose. The definition is motivated by ensuring that the infinitesimal effect of a torsion is orthogonal to any rototranslation, i.e., it induces no \emph{linear or angular momentum}. These properties can be stated more formally as follows (proof in Appendix~\ref{app:proof_momentum}):

\begin{proposition}
Let $\bfy(t) := A_\text{tor}(t\boldsymbol{\theta}, \xlig)$ for some $\boldsymbol{\theta}$ and where $t\boldsymbol{\theta} = (t\theta_1, \ldots t\theta_m)$. Then the linear and angular momentum are zero: $\frac{d}{dt} \bar\bfy|_{t=0} = 0$ and $\sum_i (\xx - \bar\xlig) \times \frac{d}{dt}\bfy_i|_{t=0} = 0$ where $\bar\xx = \frac{1}{n}\sum_i \xx_i$.
\end{proposition}

Now consider the product space\footnote{Since we never compose elements of $\PP$, we do not need to define a group structure.} $\PP = \mathbb{T}^3 \times SO(3) \times SO(2)^m$ and define $A: \PP \times \mathbb{R}^{3n} \rightarrow \mathbb{R}^{3n}$ as
\begin{equation}\label{eq:action}
    A((\rr, R, \boldsymbol{\theta}), \xx) = A_\text{tr}(\rr, A_\text{rot}(R, A_\text{tor}(\boldsymbol{\theta}, \xx)))
\end{equation}
These definitions collectively provide the sought-after product space corresponding to the docking degrees of freedom. Indeed, for a seed ligand conformation $\cc$, we can formally define the space of ligand poses $\mathcal{M}_\cc = \{A(g, \cc) \mid g \in \PP\}$. This product space $\PP$ forms the intrinsic manifold over which we will define the diffusion process and corresponds precisely to the intuitive notion of the space of ligand poses that can be reached by rigid-body motion plus torsion angle flexibility.

To ensure that the product space $\PP$ can be used to learn a DGM over ligand poses in $\mathcal{M}_\cc$ we show that (proof in Appendix~\ref{app:proof_bijection}):
\begin{proposition}
    For a given seed conformation $\cc$, the map $A(\cdot, \cc): \PP \rightarrow \mathcal{M}_\cc$ is a bijection.
\end{proposition}
which means that the inverse $A^{-1}_\cc: \mathcal{M}_\cc \rightarrow \PP$ given by $A(g, \cc) \mapsto g$ maps ligand poses $\xx \in \mathcal{M}_\cc$ to points on the product space $\PP$. We are now ready to develop a diffusion process on $\PP$.

\subsection{Diffusion}

Following \cite{de2022riemannian} to implement a diffusion model on $\PP$, it suffices to develop a method for sampling from and computing the score of the diffusion kernel on $\PP$. Furthermore, since $\PP$ is a product manifold, the forward diffusion proceeds independently in each manifold \cite{rodola2019functional}, and the tangent space  is a direct sum: $T_g \PP = T_\rr \mathbb{T}_3 \oplus T_R SO(3) \oplus T_{\boldsymbol{\theta}} SO(2)^m \cong \mathbb{R}^3 \oplus \mathbb{R}^3 \oplus \mathbb{R}^m$ where $g = (\rr, R, \boldsymbol{\theta})$. Thus, it suffices to sample from the diffusion kernel and regress against its score in each group independently.

In all three groups, we define the forward SDE as $d\xx = \sqrt{d\sigma^2(t)/dt}\, d\ww$ where $\sigma^2 = \sigma^2_\text{tr}$, $\sigma^2_\text{rot}$, or $\sigma^2_\text{tor}$ for $\mathbb{T}(3)$, $SO(3)$, and $SO(2)^m$ respectively and where $\ww$ is the corresponding Brownian motion. Since $\mathbb{T}(3) \cong \mathbb{R}^3$, the translational case is trivial and involves sampling and computing the score of a standard Gaussian with variance $\sigma^2(t)$. The diffusion kernel on $SO(3)$ is given by the $IGSO(3)$ distribution \cite{nikolayev1970normal, leach2022denoising}, which can be sampled in the axis-angle parameterization by sampling a unit vector $\boldsymbol{\hat\omega} \in \mathfrak{so}(3)$ uniformly\footnote{$\mathfrak{so}(3)$ is the tangent space of $SO(3)$ at the identity and is the space of Euler (or rotation) vectors, which are equivalent to the axis-angle parameterization.} and random angle $\omega \in [0, \pi]$ according to
\begin{equation}
\begin{footnotesize}
    p(\omega) = \frac{1-\cos \omega}{\pi} f(\omega) \quad \text{where} \quad f(\omega) = \sum_{l=0}^\infty (2l+1)\exp(-l(l+1)\sigma^2)\frac{\sin((l+1/2)\omega)}{\sin(\omega/2)}
\end{footnotesize}
\end{equation}
Further, the score of the diffusion kernel is $\nabla \ln p_t(R' \mid R) = (\frac{d}{d\omega}\log f(\omega))\boldsymbol{\hat}{\omega} \in T_{R'}SO(3)$, where $R' = \mathbf{R}(\omega \boldsymbol{\hat}{\omega})R$ is the result of applying Euler vector $\omega\boldsymbol{\hat\omega}$ to $R$. The score computation and sampling can be accomplished efficiently by precomputing the truncated infinite series and interpolating the CDF of $p(\omega)$, respectively. Finally, the $SO(2)^m$ group is diffeomorphic to the torus $\mathbb{T}^m$, on which the diffusion kernel is a \emph{wrapped normal distribution} with variance $\sigma^2(t)$. This can be sampled directly, and the score can be precomputed as a truncated infinite series \cite{jing2022torsional}.

\subsection{Score model}

\paragraph{Extrinsic-to-intrinsic.} Following, the intrinsic diffusion models framework, although we have defined the diffusion kernel and score matching objectives on $\PP$, we nevertheless develop the training and inference procedures to operate on ligand poses in 3D coordinates directly. Providing the full 3D structure, rather than abstract elements of the product space, to the score model allows it to reason about physical interactions using $SE(3)$ equivariant models, not be dependent on arbitrary definitions of torsion angles \cite{jing2022torsional}, and better generalize to unseen complexes.

\paragraph{Dependence on seed conformation.} The training and inference procedures technically depend on the choice of seed conformation $\cc$ used to define the mapping between $\mathcal{M}_\cc$ and the product space. However, providing a definite choice of $\cc$ to the score model introduces an arbitrary inference-time parameter that may affect the final predicted distribution, which is undesirable. In other words, while $\cc$ defines the manifold of ligand poses, the precise location of $\cc$ within that manifold should not affect the predicted distribution. Thus, we develop approximate training and inference procedures that remove the dependence on the $\cc$; intuitively, these assume that updates to points in the product space $\PP$ can be applied to ligand poses in $\mathcal{M}_\cc$ directly, without referencing the origin conformer $\cc$. While these are only an approximation of the theoretically correct procedures, we find that they work well in practice. In Appendix~\ref{app:train_inf}, we present the training and inference procedures in more detail and further discussion on this point.

\paragraph{Model architecture.} We construct the score model $\ss(\xx, \yy, t)$ to take as input the current ligand pose $\xx$ and protein structure $\yy$ in 3D space. The output must be in the tangent space $T_\rr \mathbb{T}_3 \oplus T_R SO(3) \oplus T_{\boldsymbol{\theta}} SO(2)^m$. The space $T_\rr \mathbb{T}_3 \cong \mathbb{R}^3$ corresponds to translation vectors and $T_R SO(3) \cong \mathbb{R}^3$ to rotation (Euler) vectors. Critically both of these vectors are $SE(3)$-equivariant (with respect to joint rototranslations of $\xx, \yy$) as ligand pose distributions are defined relative to the protein structure, which can have arbitrary location and orientation. Finally, $T_{\boldsymbol{\theta}} SO(2)^m$ corresponds to scores on $SE(3)$-invariant quantities (torsion angles). Thus, the score model must predict two $SE(3)$-equivariant vectors for the ligand as a whole and an $SE(3)$-invariant scalar at each of the $m$ freely rotatable bonds. The score model architecture is a $SE(3)$-equivariant convolutional network over point clouds \cite{thomas2018tensor, e3nn} whose architectural components are summarized below and detailed in Appendix~\ref{app:diffdock_architecture}.

Structures are represented as heterogeneous geometric graphs formed by ligand atoms and protein residues. 
Residue nodes receive as initial features language model embeddings trained on protein sequences \cite{Lin2022ESM2}. Nodes are sparsely connected based on distance cutoffs that depend on the types of nodes being linked and on the diffusion time. Intuitively, nodes are connected with the range of elements that they might be closely interacting with; this range may span widely at the start of the diffusion but is narrow at the end. Convolutional layers simultaneously operate with different sets of weights for different connection types and generate scalar and vector representations for each node. 

The ligand atom representations after the final interaction layer are then used to produce the different outputs. To produce the two $\mathbb{R}^3$ vectors representing the translational and rotational scores, we convolve the node representations with a tensor product filter placed at the center of mass. For the torsional score, we use a pseudotorque convolution to obtain a scalar at each rotatable bond of the ligand analogously to \cite{jing2022torsional}, with the distinction that, since the score model operates on coarse-grained representations, the output is not a pseudoscalar (its parity is neither odd nor even).

\subsection{Confidence model}

\paragraph{Training and inference.} In order to collect training data for the confidence model $\dd(\xx, \yy)$, we run the trained diffusion model to obtain a set of candidate poses for every training example and generate labels by testing whether or not each pose has RMSD below 2\AA{}. The confidence model is then trained with cross-entropy loss to correctly predict the binary label for each pose. During inference, the diffusion model is run to generate $N$ poses in parallel, which are passed to the confidence model that ranks them based on its confidence that they have RMSD below 2\AA{}. 

\paragraph{Architecture.} The confidence model has a similar architecture to the score model with two main differences. Firstly, its output is a single $SE(3)$-invariant scalar produced by mean-pooling the ligand atoms' scalar representations followed by a fully connected layer. Secondly, while the score model only considers a coarse-grained representation of the protein with only its $\alpha$-carbon atoms, the confidence model has access to the full atomic structure of the protein. This multiscale setup yields improved performance and a significant speed-up w.r.t. doing the whole process at the atomic scale.

\section{Experiments} \label{sec:experiments}

\subsection{Experimental setup.} 

We evaluate our method on the complexes from PDBBind \cite{liu2017PDBBind}, a large collection of protein-ligand structures collected from PDB \cite{berman2003PDB}, which was used with time-based splits to benchmark many previous works  \cite{equibind, Volkov2022PDBBindSplits, Lu2022TankBind}. We compare \textsc{DiffDock} with state-of-the-art search-based methods SMINA \cite{koes2013smina}, QuickVina-W \cite{Hassan2017QVinaW}, GLIDE \cite{halgren2004glide}, and GNINA \cite{mcnutt2021gnina} as well as the older Autodock Vina \cite{trott2010autodock}, and the recent deep learning methods EquiBind and TANKBind presented above.  Extensive details about the experimental setup, data, baselines, and implementation are in Appendix~\ref{appx:baseline_details} and all code is available at \url{https://github.com/gcorso/DiffDock}.

As we are evaluating blind docking, the methods receive two inputs: the ligand with a predicted seed conformation (e.g., from RDKit) and the crystal structure of the protein. Since search-based methods work best when given a starting binding pocket to restrict the search space, we also test the combination of using an ML-based method, such as P2Rank \cite{krivak2018p2rank} (also used by TANKBind) or EquiBind to find an initial binding pocket, followed by a search-based method to predict the exact pose in the pocket. 

To evaluate the generated complexes, we compute the heavy-atom RMSD (permutation symmetry corrected) between the predicted and the ground-truth ligand atoms when the protein structures are aligned. All methods except for EquiBind are able to generate multiple structures and rank them. We report the metrics for the highest ranked prediction as the top-1; top-5 refers to selecting the most accurate pose out of the 5 highest ranked predictions, which is a useful metric when multiple predictions are used for downstream tasks.

\subsection{Apo-structure docking}

Although large and comprehensive, the PDBBind benchmark only evaluates the capacity that various docking methods have to bind ligands to their corresponding receptor holo-structure. This is a much simpler and less realistic scenario than what is typically encountered in real applications where docking for new ligands is done against apo or holo-structures bound to a different ligand. In particular, since the development of accurate protein folding methods \cite{jumper2021highly}, docking programs are often run on top of AI-generated protein structures. With this in mind, we develop a new benchmark, referred to as PDBBind-ESMFold, where we combine the complex prediction of PDBBind with protein structures generated by ESMFold \cite{Lin2022ESM2}. 

The main design choice when generating this benchmark relies on how to best align the PDBBind complex with the ESMFold structure to obtain the "ground-truth" docked prediction on the ESMFold structure. An unbiased global alignment of the two protein structures is not desirable because a difference in structure not affecting the pocket where the ligand binds would cause the two pockets to misalign; on the other hand, only aligning residues within a single arbitrary pocket cutoff has many undesirable cases where too many or too few residues are selected or not weighted properly. 

Instead, we align receptors' residues with the Kabsch algorithm using exponential weighting, for every receptor $\xx$ its weight is $w_{\xx} = e^{-\lambda \; d_{\xx}}$ where $\lambda$ is a smoothing factor and $d_{\xx}$ is the minimum distance of $\xx$ to a ligand atom in the original complex, this way residues closer to the ligand will have a higher weight in the alignment. For each complex, we individually select $\lambda \in [0,1]$ so that it preserves distances as best as possible, in particular, we use the L-BFGS-B \cite{byrd1995limited} from \texttt{scipy} \cite{virtanen2020scipy} to minimize:
\begin{equation*}
    \lambda^* = \min_\lambda \;  \sum_{\xx \in \mathcal{X}} \;  \sum_{\yy \in \mathcal{Y}} \; \bigg(\frac{1}{\lVert \xx_c - \yy \rVert} - \frac{1}{\lVert \xx_e(\lambda) - \yy \rVert} \bigg)^2
\end{equation*}
where $\lVert \xx_c - \yy \rVert$ and $\lVert \xx_e(\lambda) - \yy \rVert$ correspond to the distances between protein residue $\xx$ and ligand atom $\yy$ respectively in the original crystal structure from PDBBind and in the complex structure obtained aligning the ESMFold structure with smoothing parameter $\lambda$. We use inverse distances to give more importance to residues closer to the ligand (in either structure) and avoid steric clashes. We only consider protein backbones because the side-chain predictions are often less reliable and their structure typically changes upon binding.

Thus we obtain protein structures on which we run the docking methods and the associated docked ligand positions that we use to evaluate them.

\subsection{Results}

\begin{table}[t]
    \caption{\textbf{PDBBind blind docking.} All methods receive a small molecule and are tasked to find its binding location, orientation, and conformation. Shown is the percentage of predictions with RMSD $<$ 2\AA{} and the median RMSD with the standard deviation (see Appendix \ref{appx:hyperparameters}). The top half contains methods that directly find the pose; the bottom half those that use a pocket prediction method. The last two lines show our method's performance. In parenthesis we specify the number of poses sampled from the generative model. * indicates that the method runs exclusively on CPU, ``-" means not applicable; some cells are empty due to infrastructure constraints. For TANKBind, the runtimes for the top-1 and top-5 predictions are different. Further evaluation details are in Appendix~\ref{appx:baseline_details}.}
    \label{tab:results_main}
     \begin{small}
     \begin{center}

    \begin{tabular}{l=c+c|+c+c|c}
    \toprule
      & \multicolumn{2}{c}{Top-1 RMSD (\AA{})} & \multicolumn{2}{c}{Top-5 RMSD (\AA{})} & Average\\ \rule{0pt}{1.5ex}  
    
        Method & \,\%$<$2\, & \,Med.\, & \,\%$<$2\, & \,Med.\, &  Runtime (s)  \\
    \midrule
    \textsc{Autodock Vina}             & 5.5 & 10.7  &      &      &  205*    \\
    \textsc{QVinaW}             & 20.9\tiny{$\pm$2.1} & 7.7\tiny{$\pm$0.8}  &      &      &  49*    \\
    \textsc{GNINA}              & 22.9\tiny{$\pm$2.2} & 7.7\tiny{$\pm$1.1}  & 32.9\tiny{$\pm$2.5} & 4.5\tiny{$\pm$0.4}  &  127  \\
    \textsc{SMINA}              & 18.7\tiny{$\pm$2.0} & 7.1\tiny{$\pm$0.4}  & 29.3\tiny{$\pm$2.3} & 4.6\tiny{$\pm$0.5}  &  126*   \\
    \textsc{GLIDE}              & 21.8\tiny{$\pm$2.1} & 9.3\tiny{$\pm$1.3}  &      &      &  1405* \\
    \textsc{EquiBind}           & 5.5\tiny{$\pm$1.2}  & 6.2\tiny{$\pm$0.3}  &  -   &  -   &  \textbf{0.04}  \\ \midrule
    \textsc{TANKBind}           & 20.4\tiny{$\pm$2.1} & 4.0\tiny{$\pm$0.2}  & 24.5\tiny{$\pm$2.1} & 3.4\tiny{$\pm$0.1}  &  0.7/2.5  \\
    \textsc{P2Rank+SMINA}       & 20.4\tiny{$\pm$2.2} & 6.9\tiny{$\pm$0.6}  & 33.2\tiny{$\pm$2.5} & 4.4\tiny{$\pm$0.5}  & 126* \\ 
    \textsc{P2Rank+GNINA}       & 28.8\tiny{$\pm$2.4} & 5.5\tiny{$\pm$0.7}  & 38.3\tiny{$\pm$2.6} & 3.4\tiny{$\pm$0.4}  & 127  \\ 
    \textsc{EquiBind+SMINA}     & 23.2\tiny{$\pm$2.2} & 6.5\tiny{$\pm$0.5}  & 38.6\tiny{$\pm$2.5} & 3.4\tiny{$\pm$0.4}  & 126*  \\
    \textsc{EquiBind+GNINA}     & 28.8\tiny{$\pm$2.3} & 4.9\tiny{$\pm$0.7}  & 39.1\tiny{$\pm$2.5} & 3.1\tiny{$\pm$0.4}  & 127   \\ \midrule
    \textbf{\textsc{DiffDock} (10)}  & 35.0\tiny{$\pm$2.5} & 3.6\tiny{$\pm$0.4}  & 40.7\tiny{$\pm$2.6} & 2.65\tiny{$\pm$0.2}  &  10     \\
    \textbf{\textsc{DiffDock} (40)}  & \textbf{38.2\tiny{$\pm$2.5}} & \textbf{3.3\tiny{$\pm$0.3}}  & \textbf{44.7\tiny{$\pm$2.6}} & \textbf{2.40\tiny{$\pm$0.2}}  & 40      \\
    \bottomrule
    \end{tabular}
    \end{center}
    \end{small}
\end{table}

\paragraph{Docking accuracy.}  \textsc{DiffDock} significantly outperforms all previous methods (Table~\ref{tab:results_main}). In particular, \textsc{DiffDock} obtains an impressive 38.2\% top-1 success rate (i.e., percentage of predictions with RMSD $<$2\AA{}\footnote{Most commonly used evaluation metric \cite{Alhossary2015QuickVina2, Hassan2017QVinaW, mcnutt2021gnina}}) when sampling 40 poses and 35.0\% when sampling just 10. This performance vastly surpasses that of state-of-the-art commercial software such as \textsc{GLIDE} (21.8\%, $p{=}2.7{\times} 10^{-7}$) and the previous state-of-the-art deep learning method TANKBind (20.4\%, $p{=}1.0{\times} 10^{-12}$). The use of ML-based pocket prediction in combination with search-based docking methods improves over the baseline performances, but even the best of these (EquiBind+GNINA)  reaches a success rate of only 28.8\% ($p{=}0.0003$).

Figure~\ref{fig:results_main}-\textit{left} shows the proportion of RMSDs below an arbitrary threshold $\epsilon$ with \textsc{DiffDock} exceeding previous methods for almost every possible $\epsilon$.\footnote{With the exception of very small $\epsilon<$1\AA{} where \textsc{GLIDE} performs better.} Figure~\ref{fig:results_main}-\textit{right} plots how the model's performance changes with the number of generative samples. Unlike regression methods like EquiBind, \textsc{DiffDock} is able to provide multiple diverse predictions of different likely poses, as highlighted in the top-5 performances. 

\paragraph{Inference runtime.} \textsc{DiffDock} holds its superior accuracy while being (on GPU) 3 to 12 times faster than the best search-based method, GNINA (Table~\ref{tab:results_main}). This high speed is critical for applications such as high throughput virtual screening for drug candidates or reverse screening for protein targets, where one often searches over a vast number of complexes. As a diffusion model, \textsc{DiffDock} is inevitably slower than the one-shot deep learning method \textsc{EquiBind}, but as shown in Figure~\ref{fig:results_main}-\textit{right} and Appendix~\ref{app:diffdock_ablations}, it can be significantly sped up without significant loss of accuracy.

\vspace{-6pt}
\begin{figure}[htb]
\begin{center}
\includegraphics[width=.485\textwidth]{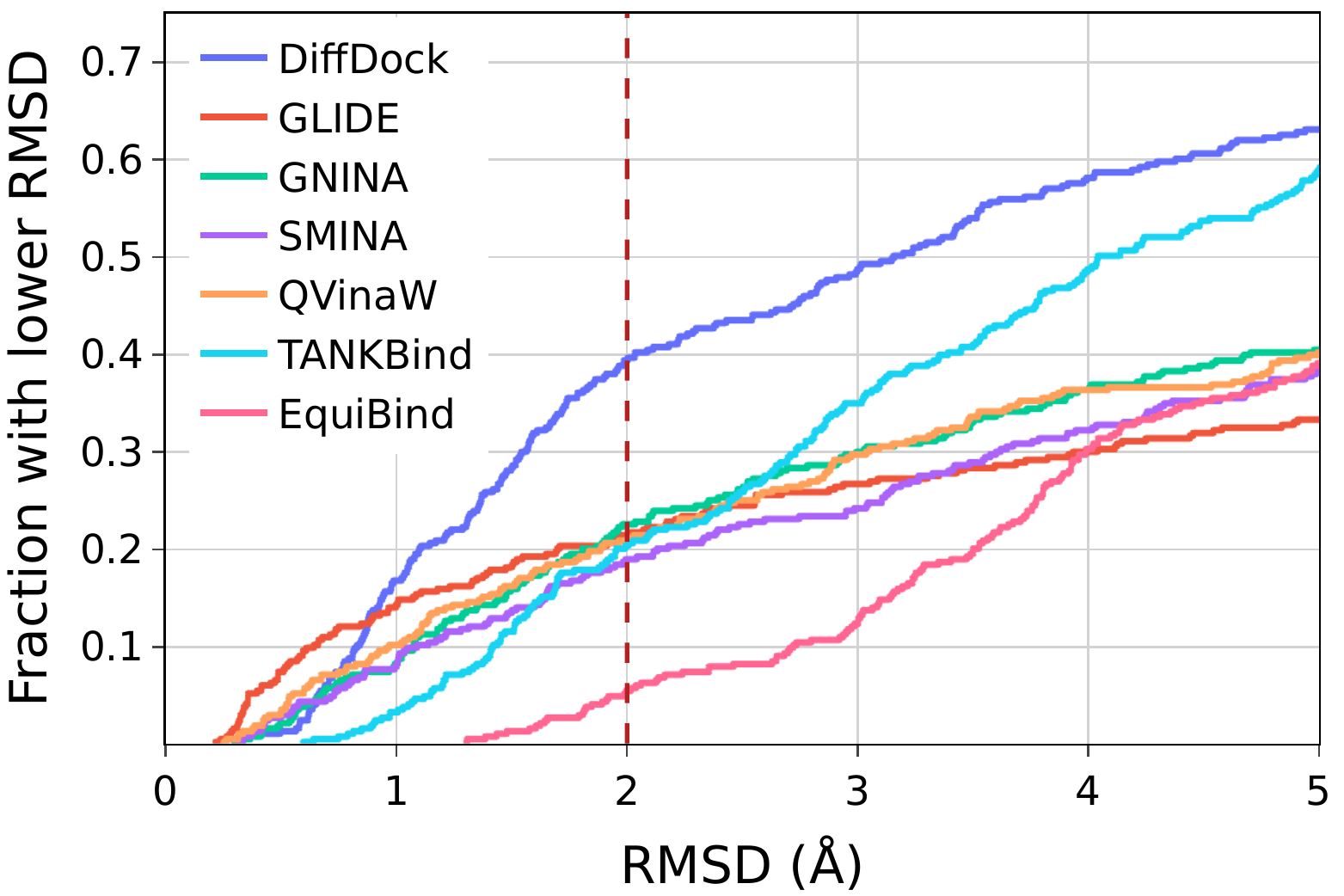}
\hspace{5pt}
\includegraphics[width=.485\textwidth]{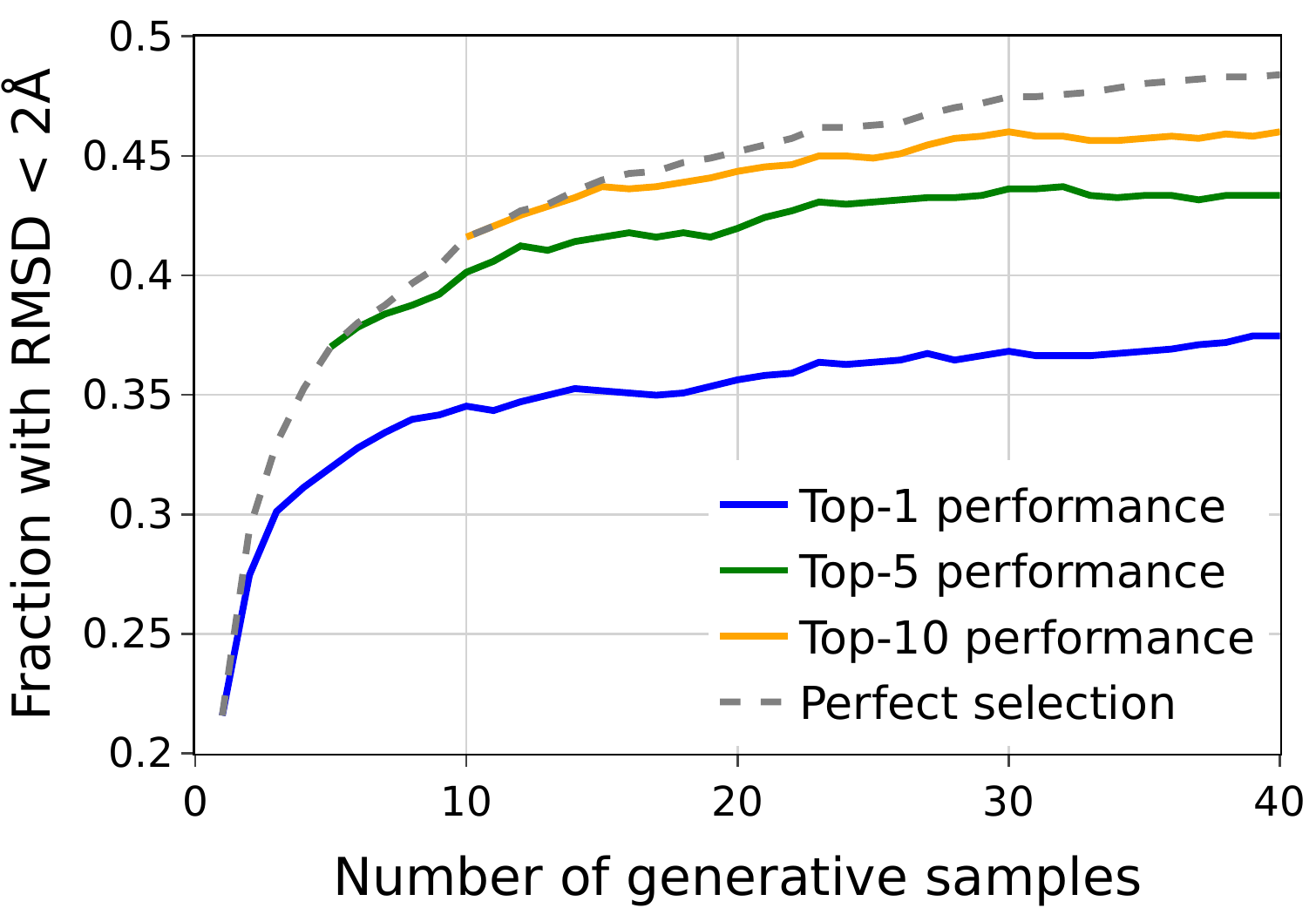}
\vspace{-6pt}
\caption{\textbf{Left:} cumulative density histogram of the methods' RMSD. \textbf{Right:} \textsc{DiffDock}'s performance as a function of the number of samples from the generative model. ``Perfect selection" refers to choosing the sample with the lowest RMSD. } 
\label{fig:results_main}
\end{center}
\end{figure}

\begin{figure}
\centering
\includegraphics[width=0.5\linewidth]{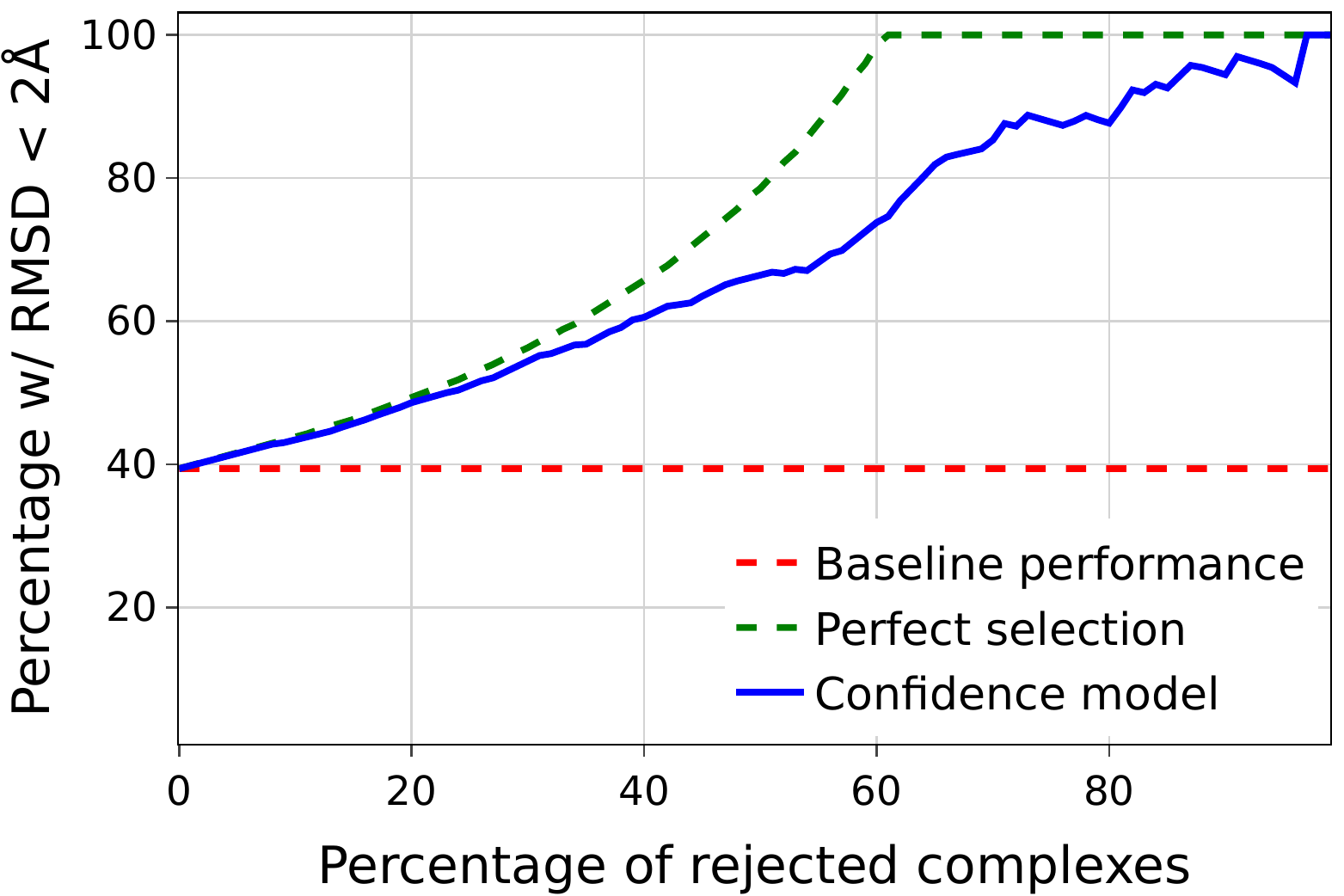}
\caption{\textbf{Selective accuracy.} Percentage of predictions with RMSD below 2\AA{} when only making predictions for the portion of the dataset where \textsc{DiffDock} is most confident.}
\label{fig:histograms_calibration}
\end{figure}

\paragraph{Selective accuracy of confidence score.} As the top-1 results show, \textsc{DiffDock}'s confidence model is very accurate in ranking the sampled poses for a given complex and picking the best one. We also investigate the \textit{selective accuracy} of the confidence model across \emph{different} complexes by evaluating how \textsc{DiffDock}'s accuracy increases if it only makes predictions when the confidence is above a certain threshold, known as \textit{selective prediction}. In Figure~\ref{fig:histograms_calibration}, we plot the success rate as we decrease the percentage of complexes for which we make predictions, i.e., increase the confidence threshold. When only making predictions for the top one-third of complexes in terms of model confidence, the success rate improves from 38\% to 83\%. Additionally, there is a high Spearman correlation of 0.68 between \textsc{DiffDock}'s confidence and the negative RMSD. Thus, the confidence score is a good indicator of the quality of \textsc{DiffDock}'s top-ranked sampled pose and provides a highly valuable confidence measure for downstream applications.

\begin{table}[t]
    \caption{\textbf{PDBBind-ESMFold blind apo-structure docking.} $^{\text{flex}}$ indicates that the side chain flexibility feature is turned on in the identified pocket. DiffDock refers to the same model described above, no further training or tuning was done on ESMFold structures. Further evaluation details are in Appendix~\ref{appx:baseline_details}. }
    \label{tab:results_main_esm}
     \begin{small}
     \begin{center}

    \begin{tabular}{l=c+c|+c+c}
    \toprule
      & \multicolumn{2}{c}{Top-1 RMSD (\AA{})} & \multicolumn{2}{c}{Top-5 RMSD (\AA{})} \\ \rule{0pt}{1.5ex}  
    
        Method & \,\%$<$2\, & \,Med.\, & \,\%$<$2\, & \,Med.\,  \\
    \midrule
    \textsc{GNINA}              & 2.0 & 22.3  & 4.0 & 14.22    \\
    \textsc{SMINA}              & 3.4 & 15.4  & 6.9 & 10.0     \\
    \textsc{EquiBind}           & 1.7 & 7.1  & - & -   \\ \midrule
    \textsc{TANKBind}           & 10.4 & 5.4  & 14.7 & 4.3    \\
    \textsc{P2Rank+SMINA}       & 4.6 & 10.0  & 10.3 & 7.0   \\ 
    \textsc{P2Rank+GNINA}       & 8.6 & 11.2  & 12.8 &  7.2   \\ 
    \textsc{EquiBind+SMINA}     & 4.3 & 8.3  & 11.7 & 5.8    \\
    \textsc{EquiBind+GNINA}     & 10.2 & 8.8  & 18.6 & 5.6     \\ \midrule
    \textsc{SMINA+SMINA}$^{\text{flex}}$       & 3.4 & 12.6  & 8.3 & 11.6   \\ 
    \textsc{GNINA+GNINA}$^{\text{flex}}$       & 1.7 & 22.1  & 5.1 &  20.0   \\ 
    \textsc{EquiBind+SMINA}$^{\text{flex}}$     & 4.3 & 7.3  & 11.7 & 5.8    \\
    \textsc{EquiBind+GNINA}$^{\text{flex}}$     & 6.6 & 9.8  & 14.6 & 6.1     \\ \midrule
    \textbf{\textsc{DiffDock} (10)}  & \textbf{21.7} & \textbf{5.0}  & \textbf{31.9} & \textbf{3.3}       \\
    \textbf{\textsc{DiffDock} (40)}  & 20.3 & 5.1  & 31.3 & \textbf{3.3}       \\
    \bottomrule
    \end{tabular}
    \end{center}
    \end{small}
\end{table}

\paragraph{Apo-structure docking.} Previous work \cite{wong2022benchmarking} highlighted that traditional search-based docking methods are not well adapted to dock molecules to apo-structures especially when these have been generated computationally. These observations are confirmed in the results in Table \ref{tab:results_main_esm} where search-based methods obtain top-1 accuracies of only 10\% or below. This is most likely due to their reliance on trying to find key-lock matches that makes them inflexible to imperfect protein structures, even when built-in options allowing side-chain flexibility are activated the results do not improve. This problem has, so far, largely prevented the computational protein folding revolution, started by AlphaFold2, to have a significant effect on the modeling of protein-ligand binding interactions \cite{wong2022benchmarking}.

Instead, the results presented in Table \ref{tab:results_main_esm} show that DiffDock is able to retain a larger proportion of its accuracy placing the top-ranked ligand below 2\AA{} away on 22\% of the complexes. This ability to better generalize to imperfect structures, even without retraining, can be attributed to a combination of (1) the robustness of the diffusion model to small perturbations in the backbone atoms, and (2) the fact that \textsc{DiffDock} does not use the exact position of side chains in the score model and is therefore forced to implicitly model their flexibility.

\chapter{Conclusion} \label{chapter:conclusion}

\section{Summary}

In this thesis, we have presented a novel approach to the fundamental class of problems around learning the 3D structure of molecules and their interactions. This approach, referred to as Intrinsic Diffusion Modeling (IDM), tackles the dynamic and uncertain nature of these structures by learning a diffusion generative model. Moreover, IDM remedies to the high dimensionality and data scarcity characterizing the problems in this class by leveraging scientific knowledge in the form of the specification of the main degrees of freedom of the systems under analysis. In order to leverage this knowledge in an efficient and generalizable way, we define a mapping of the extrinsic manifold of flexibility to a simpler intrinsic manifold, define the diffusion process on the intrinsic manifold and learn an extrinsic-to-intrinsic score model. 

We hypothesized this approach could provide significant runtime and accuracy improvement because it drastically reduces the dimensionality and increases the smoothness of the space over which we are generating while maintaining the useful inductive biases of the objects over which the model operates. In fact, we showed that instantiations of IDM tailored to the problems of molecular conformer generation and molecular docking significantly outperform existing scalable computational approaches achieving an unpreceded level of accuracy.

For molecular conformer generation, we presented, in Chapter \ref{chapter:torsional}, \emph{torsional diffusion}, which uses the IDM framework to restrict the diffusion process to the torsion angles, the most flexible degrees of freedom in molecular conformations. Torsional diffusion is the first machine learning model to significantly outperform standard cheminformatics methods and is orders of magnitude faster than previous Euclidean diffusion models. Using the exact likelihoods provided by our model, we also train the first system-agnostic Boltzmann generator.

In Chapter \ref{chapter:diffdock}, we presented \textsc{DiffDock}, an instantiation of the intrinsic diffusion modeling framework tailored to the task of molecular docking. This represents a paradigm shift from previous deep learning approaches, which use regression-based frameworks, to a generative modeling approach that is better aligned with the objective of molecular docking. The intrinsic diffusion process over the manifold describing the main degrees of freedom produces a fast and accurate generative model. 

Empirically, \textsc{DiffDock} outperforms the state-of-the-art by very large margins on PDBBind, has fast inference times, and provides confidence estimates with high selective accuracy. Moreover, unlike previous methods, it retains a large part of its accuracy even when run on apo and computationally generated protein structures, opening the way for the revolution brought by accurate protein folding methods in the modeling of protein ligand interactions.

\section{Future directions}

There are several avenues for future work that the work presented in this thesis opens up. Firstly, there is the improvement and establishment of the tools presented in chapters \ref{chapter:torsional} and \ref{chapter:diffdock}. The established benchmark for conformer generation in the machine learning community, GEOM, used to train \textit{torsional diffusion}, is composed of conformers derived with the metadynamics tool CREST with molecules simulated in a vacuum. This raises two concerns, firstly, the accuracy of CREST is not on par with more expensive computational methods like DFT or crystallography data, and secondly, chemists are typically interested in solvents very different from vacuum. Training \textit{torsional diffusion} on more accurate conformers and conditioning its generation on different solvents is an avenue for future work with high impact potential. 

When studying the interaction between a protein and a small molecule, researchers are typically not only interested in the pose with which the molecule binds to the protein, predicted by \textsc{DiffDock}, but also the affinity of such interaction. Physically this corresponds to free energy and its accurate prediction is one of the most impactful open problems in computational biophysics due to its importance in the field of drug discovery. Free energy is a thermodynamic property that depends on both the ``strength'' (enthalpy) of the interaction and its ``tightness'' (entropy), therefore, towards the goal of its accurate prediction generative methods, like \textsc{DiffDock}, providing the binding structure conformational ensemble will a key component.

Finally, for these tools to be adopted and facilitate research in chemistry and biology, it is important that they are distributed with efficient and easy-to-use libraries and programs. The development of these tools is of critical importance for the impact that these methods will have on scientific research and industry.

A second class of avenues for future work consists of the use of the Intrinsic Diffusion Modeling paradigm to tackle new problems or extend the existing methods to further degrees of flexibility. Below I list some of the problems and degrees of flexibility that I believe could be effectively tackled with IDM:
\begin{enumerate}
    \item \textbf{Molecular rings and cycles.} Since both torsional diffusion and DiffDock model the flexibility of conformers based on the torsion angles of rotatable bonds, they assume that the conformation of cycles is fixed and rely for its prediction on RDKit. While this works fine for small rings typically present in drug-like molecules, it suffers for larger and more flexible rings, especially for macrocycles. These degrees of freedom could be integrated with torsion angles in the IDM framework by, for example, employing the ring puckering coordinates \cite{cremer1975general} to model the flexibility of ring conformations as points on hyperspheres.
    \item \textbf{Protein flexibility.} \textsc{DiffDock} assumes that the structure of the protein is fixed and preserves the structure that was given as input. Although preliminary results have shown that \textsc{DiffDock} is robust to inaccuracies in the structure given as input, the fixed protein assumption prevents us to study how the protein conformation changes upon binding, a factor that can be very important for evaluating the affinity of the interaction. Modeling protein flexibility, both in general and upon binding, is therefore a very important problem where an IDM-based approach could provide significant improvements over existing methods. While some follow-up work \cite{wu2022protein} applied the torsional diffusion framework to the full protein molecular graph, this approach is problematic because of the large lever-arm effect that changing a torsion in the backbone can have on very distant parts of the protein. Instead, I believe that a promising approach is to use torsional flexibility to model sidechain flexibility (where the lever-arm effect is limited) and use some local flexibility scheme such as the backrub motion \cite{davis2006backrub} to model movements in the backbone.
    \item \textbf{Protein-protein interactions.} The IDM framework could be also applied to model protein-protein interactions, a fundamental problem in structural biology. One promising approach to this problem could involve combining the $SO(3)$ and $\mathbb{T}(3)$ components of \textsc{DiffDock} to model the rigid protein-protein docking problem and the protein flexibility components discussed above to model the conformation of each of the proteins.
\end{enumerate}

Finally, it is also a very exciting avenue of future work the extension of the IDM framework to model more general and complex problems. Some interesting avenues of research in this direction are:
\begin{enumerate}
    \item automatically discovering from data extrinsic and intrinsic manifolds that well describe degrees of freedom of a generation problem;
    \item relaxing the condition that the diffusion is done exclusively on the extrinsic manifold, but, instead, using such manifold as a soft constraint or inductive bias to make the full dimensional diffusion more efficient;
    \item improve the framework to train IDM from an energy or reward function presented in Section \ref{sec:boltzmann}, making it more efficient and effective;
    \item successively improve or jointly train the score model and the confidence or energy model presented in \textsc{DiffDock};
    \item support and design forward diffusion processes that more closely align with physical priors leading to more stable conformations even before relaxation.
\end{enumerate}

Overall the work presented in this thesis makes me very optimistic that diffusion generative models will have a profound impact in many areas of structural biology and persuaded that a more careful and effective design of the domain and process of the diffusion will be critical to achieving these results.

\appendix
\chapter{Proofs}

\setcounter{proposition}{0}

\section{Chapter 3: Torsional Diffusion} \label{app:propositions}

Reported in this section are the proofs of the propositions in Chapter \ref{chapter:torsional}. These were primarily developed by Bowen Jing and Jeffrey Chang.

\subsection{Definitions}\label{app:def}

Consider a molecular graph $G = (\mathcal{V}, \mathcal{E})$ and its space of possible conformers $\mathcal{C}_G$. A conformer is an assignment $\mathcal{V} \mapsto \mathbb{R}^3$ of each atom to a point in 3D-space, defined up to global rototranslation. For notational convenience, we suppose there is an ordering of nodes such that we can regard a mapping as a vector in $\mathbb{R}^{3n}$ where $n=|\mathcal{V}|$. Then a conformer $C \in \mathcal{C}_G$ is a set of $SE(3)$-equivalent vectors in $\mathbb{R}^{3n}$---that is, $\mathcal{C}_G \cong \mathbb{R}^{3n}/SE(3)$. This defines the space of conformers in terms of \emph{extrinsic} (or Cartesian) coordinates.

An \emph{intrinsic} (or internal) coordinate is a function over $\mathcal{C}_G$---i.e., it is an $SE(3)$-invariant function over $\mathbb{R}^{3n}$. There are four types of such coordinates typically considered:

\textbf{Bond lengths}. For $(a, b) \in \mathcal{E}$, the bond length $l_{ab} \in [0, \infty)$ is defined as $|x_{a} - x_b|$.

\textbf{Bond angles}. For $a, b, c \in \mathcal{V}$ such that $a, c \in \mathcal{N}(b)$, the bond angle $\alpha_{abc} \in [0, \pi]$ is defined by
\begin{equation}
    \cos \alpha_{abc} := \frac{(x_c - x_b)\cdot(x_a - x_b)}{|x_c - x_b||x_a - x_b|}
\end{equation}

\textbf{Chirality}. For $a \in \mathcal{V}$ with 4 neighbors $b, c, d, e \in \mathcal{N}(a)$, the chirality ${z}_{abcd}\in \{-1, 1\}$ is defined as
\begin{align}
    {z}_{abcde} := \sign\det\begin{pmatrix} 1 & 1 & 1 & 1 \\ x_b - x_a & x_c -x_a & x_d-x_a & x_e - x_a\end{pmatrix}
\end{align}
Similar quantities are defined for atoms with other numbers of neighbors. Chirality is often considered part of the specification of the molecule, rather than the conformer. 

\textbf{Torsion angles}. For $(b, c) \in \mathcal{E}$, with a choice of reference neighbors $a\in\mathcal{N}(b)\setminus\{c\}, d\in\mathcal{N}(c)\setminus\{b\}$, the torsion angle $\tau_{abcd} \in [0, 2\pi)$ is defined as the dihedral angle between planes $abc$ and $bcd$:
\begin{equation}\label{eq:torsion_def}
\begin{aligned}
    \cos \tau_{abcd} &= \frac{\mathbf{n}_{abc} \cdot \mathbf{n}_{bcd}}{|\mathbf{n}_{abc}||\mathbf{n}_{bcd}|}\\
    \sin \tau_{abcd} &= \frac{\mathbf{u}_{bc} \cdot (\mathbf{n}_{abc} \times \mathbf{n}_{bcd})}{|\mathbf{u}_{bc}||\mathbf{n}_{abc}||\mathbf{n}_{bcd}|}
\end{aligned}
\end{equation}
    
where $\mathbf{u}_{ab} = x_b - x_a$ and $\mathbf{n}_{abc}$ is the normal vector $\mathbf{u}_{ab}\times\mathbf{u}_{bc}$. Note that $\tau_{abcd} = \tau_{dcba}$---i.e., the dihedral angle is the same for four consecutively bonded atoms regardless of the direction in which they are considered.

A \textbf{complete set of intrinsic coordinates} of the molecule is a set of such functions $(f_1, f_2, \ldots)$ such that $F(C) = (f_1(C), f_2(C), \ldots)$ is a bijection. In other words, they fully specify a unique element of $\mathcal{C}_G$ without overparameterizing the space. In general there exist many possible such sets for a given molecular graph. We will not discuss further how to find such sets, as our work focuses on manipulating molecules in a way that holds fixed all $l, \alpha, z$ and only modifies (a subset of) torsion angles $\tau$.

As presently stated, the \textbf{torsion angle about a bond} $(b, c)\in\mathcal{E}$ is ill-defined, as it could be any $\tau_{abcd}$ with $a\in\mathcal{N}(b)\setminus\{c\}, d\in\mathcal{N}(c)\setminus\{b\}$. However, any complete set of intrinsic coordinates needs to only have at most one such $\tau_{abcd}$ for each bond $(b, c)$ \cite{ganea2021geomol}. Thus, we often refer to \emph{the} torsion angle about a bond $(b_i, c_i)$ as $\tau_i$ when reference neighbors $a_i, b_i$ are not explicitly stated.

\subsection{Torsion update} \label{app:proof_update}

Given a freely rotatable bond $(b_i, c_i)$, by definition removing $(b_i, c_i)$ creates two connected components $\mathcal{V}(b_i), \mathcal{V}(c_i)$. Then, consider torsion angle $\tau_j$ at a different bond $(b_j, c_j)$ with neighbor choices $a_j \in \mathcal{N}(b_j), d_j \in \mathcal{N}(c_j), a_j \neq c_j, d_j \neq b_j$. Without loss of generality, there are two cases
\begin{itemize}
    \item Case 1: $a_j, b_j, c_j, d_j \in \mathcal{V}(b_i)$
    \item Case 2: $d_j \in \mathcal{V}(c_i)$ and $a_j, b_j, c_j \in \mathcal{V}(b_i)$
\end{itemize} 
Note that in Case 2, $c_j = b_i$ and $d_j = c_i$ must hold because there is only one edge between $\mathcal{V}(b_i), \mathcal{V}(c_i)$. With these preliminaries we now restate the proposition:

\begin{proposition} 
Let $(b_i, c_i)$ be a rotatable bond, let $\mathbf{x}_{\mathcal{V}(b_i)}$ be the positions of atoms on the $b_i$ side of the molecule, and let $R(\boldsymbol{\theta}, x_{c_i}) \in SE(3)$ be the rotation by Euler vector $\boldsymbol{\theta}$ about $x_{c_i}$. Then for $C, C' \in \mathcal{C}_G$, if $\tau_i$ is any definition of the torsion angle around bond $(b_i, c_i)$,
\begin{equation} 
    \begin{aligned}
        \tau_i(C') &= \tau_i(C) + \theta\\
        \tau_j(C') &= \tau_j(C) \quad \forall j\neq i
    \end{aligned}
    \qquad \text{if} \qquad
    \exists \mathbf{x} \in C, \mathbf{x'} \in C'\ldotp \quad
    \begin{aligned}
    \mathbf{x}'_{\mathcal{V}(b_i)} &=  \mathbf{x}_{\mathcal{V}(b_i)} \\
    \mathbf{x}'_{\mathcal{V}(c_i)} &=  R\left(\theta \, \mathbf{\hat r}_{b_ic_i}, x_{c_i} \right)\mathbf{x}_{\mathcal{V}(c_i)}
    \end{aligned}
\end{equation}
where $\mathbf{\hat r}_{b_ic_i} = (x_{c_i} - x_{b_i})/||x_{c_i}-x_{b_i}||$.
\end{proposition}
\begin{small}
\begin{proof}
First we show $\tau_i(C') = \tau_i(C)+\theta$, for which it suffices to show $\tau_i(\mathbf{x}')=\tau_i(\mathbf{x})+\theta$. Because $a_i, b_i \in \mathcal{V}(b_i)$, $x_{a_i}' = x_{a_i}$ and $x_{b_i}' = x_{b_i}$. Since the rotation of $\mathbf{x}_{\mathcal{V}(c_i)}$ is centered at $x_{c_i}$, we have $x_{c_i}' = x_{c_i}$ as well. Now we consider $d_i$ and $\mathbf{u}'_{cd} = x'_{d_i} - x_{c_i}$. By the Rodrigues rotation formula,
\begin{equation}
\begin{aligned}
    \mathbf{u}'_{cd} 
    &= \mathbf{u}_{cd}\cos\theta + \frac{\mathbf{n}_{bcd}}{|\mathbf{u}_{bc}|}\sin\theta + \frac{\mathbf{u}_{bc}}{|\mathbf{u}_{bc}|}\left(\frac{\mathbf{u}_{bc}}{|\mathbf{u}_{bc}|}\cdot \mathbf{u}_{cd} \right)(1-\cos\theta)
\end{aligned}
\end{equation}
Then we have
\begin{equation}
\begin{aligned}
    \mathbf{n}'_{bcd} = \mathbf{u}_{bc} \times \mathbf{u}'_{cd}
    &=\mathbf{n}_{bcd}\cos\theta - \left(\mathbf{n}_{bcd} \times \frac{\mathbf{u}_{bc}}{|\mathbf{u}_{bc}|}\right)\sin\theta
\end{aligned}\end{equation}
To obtain $|\mathbf{n}'_{bcd}|$, note that since $\mathbf{n}_{bcd}\perp \mathbf{u}_{bc}$, 
\begin{equation}
    \bigg|\mathbf{n}_{bcd} \times \frac{\mathbf{u}_{bc}}{|\mathbf{u}_{bc}|}\bigg| = |\mathbf{n}_{bcd}|
\end{equation}
which gives $|\mathbf{n}'_{bcd}| = |\mathbf{n}_{bcd}|$.
Thus,
\begin{equation}
\begin{aligned}    
    \cos\tau'_i &= \frac{\mathbf{n}_{abc} \cdot \mathbf{n}'_{bcd}}{|\mathbf{n}_{abc}||\mathbf{n}_{bcd}|} = \frac{\mathbf{n}_{abc} \cdot \mathbf{n}_{bcd}}{|\mathbf{n}_{abc}||\mathbf{n}_{bcd}|} \cos\theta - \frac{\mathbf{n}_{abc} \cdot \left(\mathbf{n}_{bcd}\times\mathbf{u}_{bc}\right)}{|\mathbf{n}_{abc}||\mathbf{n}_{bcd}||\mathbf{u}_{bc}|} \sin\theta \\ 
    &= \cos\tau_i\cos\theta - \sin\tau_i\sin\theta = \cos (\tau_i + \theta)
\end{aligned}
\end{equation}
Similarly, 
\begin{equation}
\begin{aligned}
    \sin\tau'_i &= \frac{\mathbf{u}_{bc} \cdot (\mathbf{n}_{abc} \times \mathbf{n}'_{bcd})}{|\mathbf{u}_{bc}||\mathbf{n}_{abc}||\mathbf{n}_{bcd}|} = \frac{\mathbf{u}_{bc} \cdot (\mathbf{n}_{abc} \times \mathbf{n}_{bcd})}{|\mathbf{u}_{bc}||\mathbf{n}_{abc}||\mathbf{n}_{bcd}|}\cos\theta - \frac{\mathbf{u}_{bc} \cdot (\mathbf{n}_{abc} \times \left(\mathbf{n}_{bcd}\times\mathbf{u}_{bc}\right))}{|\mathbf{u}_{bc}|^2|\mathbf{n}_{abc}||\mathbf{n}_{bcd}|}\sin\theta \\
    &= \sin\tau_i\cos\theta + \cos\tau_i\sin\theta = \sin(\tau_i + \theta)
\end{aligned}
\end{equation}
Therefore, $\tau'_i = \tau_i+\theta$

Now we show $\tau'_j = \tau_j$ for all $j\neq i$. Consider any such $j$. For Case 1, $x'_{a_j} = x_{a_j}, x'_{b_j} = x_{b_j}, x'_{c_j} = x_{c_j}, x'_{d_j} = x_{d_j}$ so clearly $\tau'_j = \tau_j$. For Case 2, $x'_{a_j} = x_{a_j}, x'_{b_j} = x_{b_j}, x'_{c_j} = x_{c_j}$ immediately. But because $d_j = c_i$, we also have $x'_{d_j} = x_{d_j}$. Thus,  $\tau'_j = \tau_j$.
\end{proof}
\end{small}
\subsection{Parity equivariance}  \label{app:proof_parity}
\begin{proposition}
    If $p(\boldsymbol{\tau}(C) \mid L(C)) = p(\boldsymbol{\tau}(-C) \mid L(-C))$, then for all diffusion times $t$,
    \begin{equation}
        \nabla_{\boldsymbol{\tau}} \log p_t(\boldsymbol{\tau}(C) \mid L(C)) = -\nabla_{\boldsymbol{\tau}} \log p_t(\boldsymbol{\tau}(-C) \mid L(-C)) 
    \end{equation}
\end{proposition}
\begin{small}

\begin{proof}
From Equation~\ref{eq:torsion_def} we see that for any torsion $\tau_i$, we have $\tau_i(-C) = -\tau_i(C)$; therefore $\boldsymbol{\tau}_i(-C) = -\boldsymbol{\tau}_i(C)$, which we denote $\boldsymbol{\tau}_-$. Also denote $\boldsymbol{\tau} := \boldsymbol{\tau}(C), p_t(\boldsymbol{\tau}) := p_t(\boldsymbol{\tau} \mid L(C))$ and $p'_t(\boldsymbol{\tau}_-) := p_t(\boldsymbol{\tau}_- \mid L(-C))$. We claim $p_t(\boldsymbol{\tau}) = p'_t(\boldsymbol{\tau}_-)$ for all $t$. Since the perturbation kernel (\eqref{eq:torus_score}) is parity invariant,
\begin{equation}
\begin{aligned}
    p'_t(\boldsymbol{\tau}_{-}) &= \int_{\mathbb{T}^m} p'_0(\boldsymbol{\tau}'_{-}) p_{t\mid 0}(\boldsymbol{\tau}_{-} \mid \boldsymbol{\tau}'_{-})\; d\boldsymbol{\tau}'_{-} \\
    &= \int_{\mathbb{T}^m} p_0(\boldsymbol{\tau}') p_{t\mid 0}(\boldsymbol{\tau} \mid \boldsymbol{\tau}')\; d\boldsymbol{\tau}'_{-} = p_t(\boldsymbol{\tau})
\end{aligned}
\end{equation}
Next, we have
\begin{equation}
\begin{aligned}
    \nabla_{\boldsymbol{\tau}}\log p'_t(\boldsymbol{\tau}_-) &= \frac{\partial\boldsymbol{\tau}_-}{\partial\boldsymbol{\tau}}\nabla_{\boldsymbol{\tau}^{-}}\log p'_t(\boldsymbol{\tau}_-) \\
    &= -\nabla_{\boldsymbol{\tau}}\log p_t(\boldsymbol{\tau}) 
\end{aligned}
\end{equation}
which concludes the proof.
\end{proof}
\end{small}

\subsection{Likelihood conversion} \label{app:proof_likelihood}

\begin{proposition}
 \label{prop:euclidean_app}
Let $\mathbf{x} \in C(\boldsymbol{\tau}, L)$ be a centered conformer in Euclidean space. Then,
\begin{equation}
    p_G(\mathbf{x} \mid L) = \frac{p_G(\boldsymbol{\tau} \mid L)}{ 8 \pi^2 \sqrt{\det g}}
    \quad \mathrm{where} \ \
    g_{\alpha\beta} = 
    \sum_{k=1}^{n} 
    J^{(k)}_{\alpha} \cdot J^{(k)}_{\beta}
\end{equation}
where the indices $\alpha,\beta$ are integers between 1 and $m+3$. For $1 \leq \alpha \leq m$, $J^{(k)}_\alpha$ is defined as
    \begin{align}
    \label{eqn:basisvec_app}
        J^{(k)}_{i} &= \tilde J^{(k)}_{i} - \frac 1 n \sum_{\ell=1}^{n} \tilde J^{(\ell)}_{i}
        \quad
        \mathrm{with} \ \
        \tilde J^{(\ell)}_{i} =
        \begin{cases}
            0 & \ell \in \mathcal{V}(b_i), \\
            \frac{\mathbf{x}_{b_i} - \mathbf{x}_{c_i}} {||\mathbf{x}_{b_i} - \mathbf{x}_{c_i}||}
            \times
            \left( \mathbf{x}_\ell - \mathbf{x}_{c_i} \right),
            & \ell \in \mathcal{V}(c_i),
        \end{cases}
    \end{align}
    and for $\alpha \in \{m+1, m+2, m+3\}$ as
    \begin{align}
    \label{eq:omegajacobian_app}
        J^{(k)}_{m+1} &= \mathbf{x}_k \times \hat{x},
        \qquad
        J^{(k)}_{m+2} = \mathbf{x}_k \times \hat{y},
        \qquad
        J^{(k)}_{m+3} = \mathbf{x}_k \times \hat{z},
        \qquad
    \end{align}
    where $(b_i, c_i)$ is the freely rotatable bond for torsion angle $i$, $\mathcal{V}(b_i)$ is the set of all nodes on the same side of the bond as $b_i$, and $\hat x, \hat y, \hat z$ are the unit vectors in the respective directions.
\end{proposition}

\begin{proof}\begin{small}
Let $M$ be $(m+3)$-dimensional manifold embedded in $3n$-dimensional Euclidean space formed by the set of all centered conformers with fixed local structures but arbitrary torsion angles and orientation. A natural set of coordinates for $M$ is $q^\alpha = \{\tau_1, \tau_2, \ldots, \tau_m, \omega_x, \omega_y, \omega_z\}$, where $\tau_i$ is the torsion angle at bond $i$ and $\omega_x, \omega_y, \omega_z$ define the global rotation about the center of mass:
\begin{equation}
    \mathbf{x}_k = \tilde {\mathbf{x}}_k  - \frac 1 n \sum_{\ell=1}^n {\tilde{\mathbf{x}}}_\ell
    \quad
    \mathrm{where}
    \quad
    {\tilde{\mathbf{x}}}_\ell = e^{\Lambda(\omega)} \mathbf{x}^\prime_k,
    \quad
    \Lambda(\omega) =
    \begin{pmatrix}
        0 & -\omega_z & \omega_y \\
        \omega_z & 0 & -\omega_x \\
        -\omega_y & \omega_x & 0
    \end{pmatrix}.
\end{equation}
Here ${\mathbf{x}}_k^\prime$ is the position of atom $k$ as determined by the torsion angles, without centering or global rotations, and $\omega_x, \omega_y, \omega_z$ are rotation about the $x$, $y$, and $z$ axis respectively.

Consider the set of covariant basis vectors
\begin{equation}
    \mathbf{J}_\alpha = \frac{\partial \mathbf{x}}{\partial q^\alpha}.
\end{equation}
and corresponding the covariant components of the metric tensor,
\begin{equation}
    g_{\alpha \beta} = \mathbf{J}_\alpha \cdot \mathbf{J}_\beta = \frac{\partial \mathbf{x}}{\partial q^\alpha} \cdot \frac{\partial \mathbf{x}}{\partial q^\beta}.
\end{equation}
The conversion factor between torsional likelihood and Euclidean likelihood is given by
\begin{equation}
    \int \sqrt{\det \mathbf{g}} \, d^{3} \omega,
    \label{eq:conversionfactor_app}
\end{equation}
where $\sqrt{\det \mathbf{g}} \, d^{m+3}q$ is the invariant volume element on $M$ \cite{carroll2019spacetime}, and the integration over $\omega$ marginalizes over the uniform distribution over global rotations. The calculation of Eq.~\ref{eq:conversionfactor_app} proceeds as follows.

Let the position of the $k$'th atom be $x_k$, and let the three corresponding components of $\mathbf{J}_\alpha$ be $J^{(k)}_\alpha$. For $1 \leq i \leq m$, $J^{(k)}_i$is given by
\begin{equation}
   J^{(k)}_i = \frac{\partial}{\partial \tau_i}
   \left(
        \tilde {\mathbf{x}}_k
        - \frac 1 n \sum_{\ell=1}^n {\tilde{\mathbf{x}}}_\ell
   \right)
   = \tilde J^{(k)}_{i} - \frac 1 n \sum_{\ell=1}^{n} \tilde J^{(\ell)}_{i}
\end{equation}
where $\tilde J^{(k)}_{i} := \partial \tilde {\mathbf{x}}_{k}/\partial \tau_i$ is the displacement of atom $k$ upon an infinitesmal change in the torsion angle $\tau_i$, without considering the change in the center of mass. Clearly $\tilde J^{(b_i)}_{i} = \tilde J^{(c_i)}_{i} = 0$ because neither $b_i$ nor $c_i$ itself is displaced; furthermore, all atoms on the $b$ side of torsioning bond are not displaced, so $J^{(k)}_{i} = 0$ for all $k \in \mathcal{N}(b_i)$. The remaining atoms, in $\mathcal{N} (c_i)$, are rotated about the axis of the $(b_i,c_i)$ bond. The displacement per infinitesimal $\partial \tau_i$ is given by the cross product of the unit normal along the rotation axis, $(\tilde x_{c_i} - \tilde x_{b_i}) / {||\tilde x_{c_i} - \tilde x_{b_i}||}$, with the displacement from rotation axis, $\tilde x_k - \tilde x_{b_i}$. This cross product yields ${J}^{(k)}_\alpha$ in Eq.~\ref{eqn:basisvec_app}, where the tildes are dropped as relative positions do not depend on center of mass. For $\alpha \in \{m+1, m+2, m+3\}$, a similar consideration of the cross product with the rotation axis yields Eq.~\ref{eq:omegajacobian_app}. Finally, since none of the components of the metric tensor depend explicitly on $\omega$, the integration over $\omega$ in Eq.~\ref{eq:conversionfactor_app} is trivial and yields the volume over $SO(3)$ of $8 \pi^2$ \cite{chirikjian2011stochastic}, proving the proposition.
\end{small}\end{proof}

\section{Chapter 4: DiffDock} \label{app:diffdock_propositions}

Reported in this section are the proofs of the propositions in Chapter \ref{chapter:diffdock}. These were primarily developed by Bowen Jing.

\subsection{Zero momentum} \label{app:proof_momentum}

\begin{proposition}
Let $\bfy(t) := A_\text{tor}(t\boldsymbol{\theta}, \xlig)$ for some $\boldsymbol{\theta}$ and where $t\boldsymbol{\theta} = (t\theta_1, \ldots t\theta_m)$. Then the linear and angular momentum are zero: $\frac{d}{dt} \bar\bfy|_{t=0} = 0$ and $\sum_i (\xx - \bar\xlig) \times \frac{d}{dt}\bfy_i|_{t=0} = 0$ where $\bar\xx = \frac{1}{n}\sum_i \xx_i$.
\end{proposition}
\begin{proof}
    Let $\bfy(t) = R(t)(B(t, \boldsymbol{\theta}, \xx)-\bar\xx) + \bar\xx + \pp(t)$ where $B(t, \boldsymbol{\theta}, \cdot)$ = $B_{1,t\theta_1}\circ \cdots B_{m, t\theta_m}$ and $R(t), \pp(t)$ are the rotation (around $\bar\xx$) and translation associated with the optimal RMSD alignment between $B(t, \boldsymbol{\theta}, \xx)$ and $\xx$. By definition of $\rmsd$, for any $t$, $R(t)$ and $\pp(t)$ minimize
    \begin{equation}
         \lnorm\bfy(t)-\xx\rnorm = \lnorm R(t)(B(t, \boldsymbol{\theta}, \xx)-\bar\xx) + \bar\xx + \pp(t) - \xx\rnorm
    \end{equation}
    For infinitesimal $t = dt$, the RHS becomes
    \begin{equation}
    \begin{aligned}
        \text{RHS} &= \lnorm R(dt)(B(dt, \boldsymbol{\theta}, \xx)-\bar\xx) + \bar\xx + \pp(dt) - \xx\rnorm \\
        &= \lnorm \left(R'(0) \, dt + R(0)\right)\left(B'(0, \boldsymbol{\theta}, \xx)\, dt + B(0, \boldsymbol{\theta}, \xx)-\bar\xx\right) + \bar\xx + \pp'(0)\, dt + \pp(0) - \xx\rnorm\\
        &= \lnorm R'(0) \left(\xx -\bar\xx \right)\, +B'(0, \boldsymbol{\theta}, \xx)\, + \pp'(0)\, \rnorm \, dt
    \end{aligned}
    \end{equation}
    where we have used $R(0) = I$, $B(0, \boldsymbol{\theta}, \xx) = \xx$, and $\pp(0) = 0$. Thus, we see that RMSD alignment implies that the derivatives of $R(t), \pp(t)$ minimize the norm of
    \begin{equation}
        \bfy'(0) = R'(0) \left(\xx -\bar\xx \right)\, +B'(0, \boldsymbol{\theta}, \xx)\, + \pp'(0)
    \end{equation}
    This expression represents the instantaneous velocity of the points $\bfy_i$ at $t=0$. We now show that minimizing the velocity results in zero linear and angular momentum.
    
    We abbreviate $B'(t, \boldsymbol{\theta}, \xx(0))_i := \bb_i$ and $\pp' = \vv$. Further, let $\rr_i = \xx_i - \bar\xx$, such that the rotational contribution to the velocity can be written in terms of an angular velocity vector $\boldsymbol{\omega}$. With this, at $t=0$ we have
    \begin{equation}
        \bfy'_i = \bb_i + \boldsymbol{\omega} \times \rr_i + \vv
    \end{equation}
    We thus obtain the squared norm as
    \begin{equation}
    \begin{footnotesize}
        \begin{aligned}
            \sum_i \lnorm \bfy'_i \rnorm^2 &= \sum_i (\bb_i + \boldsymbol{\omega} \times \rr_i + \vv)\cdot(\bb_i + \boldsymbol{\omega} \times \rr_i+ \vv) \\
            &= \sum_i \left[\lnorm \bb_i \rnorm^2 + 2\bb_i\cdot(\boldsymbol{\omega} \times \rr_i) + 2\bb_i\cdot\vv + (\boldsymbol{\omega} \times \rr_i) \cdot (\boldsymbol{\omega} \times \rr_i) + 2 (\boldsymbol{\omega} \times \rr_i) \cdot \vv + \lnorm \vv \rnorm^2\right]\\
            &= \sum_i \lnorm \bb_i \rnorm^2 + 2\boldsymbol{\omega} \cdot \sum_i (\rr_i \times \bb_i) + 2\left(\sum_i \bb_i\right)\cdot\vv +  n\lnorm \vv \rnorm^2 + \boldsymbol{\omega}^T\mathcal{I}(\rr)\boldsymbol{\omega}
        \end{aligned}
    \end{footnotesize}
    \end{equation}
    where we have used the fact that $\sum_i \rr_i = 0$ and where $\mathcal{I}(\rr) = \left(\sum_i \rr_i \cdot \rr_i\right)I - \sum_i \rr_i \rr_i^T$ is the $3\times 3$ \emph{inertia tensor}. To minimize the squared norm (and thus the norm itself), we set gradients with respect to $\vv, \boldsymbol{\omega}$ to zero. This gives
    \begin{equation}
        \vv = -\frac{1}{n}\sum_i \bb_i \quad \text{and} \quad \boldsymbol{\omega} = -\mathcal{I}(\rr)^{-1}\left(\sum_i \rr_i \times \bb_i\right)
    \end{equation}
    Now with $\bfy'_i = \bb_i + \boldsymbol{\omega} \times \rr_i + \vv$ we evaluate the linear momentum
    \begin{align}
        \frac{1}{n}\sum_i \bfy_i' = \frac{1}{n}\left(\sum_i \bb_i + \boldsymbol{\omega}\times\sum_i\rr_i + n\vv\right) = 0
    \end{align}
    which is zero by direct substitution of $\vv$. Similarly, we evaluate the angular momentum
    \begin{equation}
    \begin{aligned}
        \sum_i \rr_i \times \yy'_i &= \sum_i \rr_i \times \bb_i + \sum_i \rr_i \times (\boldsymbol{\omega} \times \rr_i) + \sum_i \rr_i \times \vv \\
        &= \sum_i \rr_i \times \bb_i + \mathcal{I}(\rr)\boldsymbol{\omega} = 0
    \end{aligned}
    \end{equation}
    which is zero by direct substitution of $\boldsymbol{\omega}$. Thus, the linear and angular momentum are zero at $t=0$ for arbitrary $\xx$. 
\end{proof}

Note that since we did not use the particular form of $B(t\boldsymbol{\theta}, \xx)$ in the above proof, we have shown that RMSD alignment can be used to disentangle rotations and translations from the infinitesimal action of any arbitrary function.

\subsection{Map is bijection} \label{app:proof_bijection}

\begin{proposition}
    For a given seed conformation $\cc$, the map $A(\cdot, \cc): \PP \rightarrow \mathcal{M}_\cc$ is a bijection.
\end{proposition}
\begin{proof}
    Since we defined $\mathcal{M}_\cc = \{ A(g, \cc) \mid g \in \PP\}$, $A(\cdot, \cc)$ is automatically surjective. We now show that it is injective. Assume for the sake of contradiction that $A(\cdot, \cc)$ is not injective, so that there exist elements of the product space $g_1, g_2 \in \PP$ with $g_1 \neq g_2$ but with $A(g_1, \cc) = A(g_2, \cc) = \cc'$. That is,
    \begin{equation}\label{eq:injective}
        A_\text{tr}(\rr_1, A_\text{rot}(R_1, A_\text{tor}(\boldsymbol{\theta}_1, \cc))) = A_\text{tr}(\rr_2, A_\text{rot}(R_2, A_\text{tor}(\boldsymbol{\theta}_2, \cc)))
    \end{equation}
    which we abbreviate as $\cc^{(1)} = \cc^{(2)}$. Since only $A_\text{tr}$ changes the center of mass $\sum_i \cc_i /n$, we have $\sum_i \cc^{(1)}_i/n = \sum_i \cc_i/n + \rr_1$ and $\sum_i \cc^{(2)}_i/n = \sum_i \cc_i/n + \rr_2$. However, since $\cc^{(1)} = \cc^{(2)}$, this implies $\rr_1 = \rr_2$. Next, consider the torsion angles $\boldsymbol{\tau}_1 = (\tau^{(1)}_1, \ldots \tau^{(1)}_m)$ of $\cc^{(1)}$ corresponding to some choice of dihedral angles at each rotatable bond. Because $A_\text{tr}$ and $A_\text{rot}$ are rigid-body motions, only $A_\text{tor}$ changes the dihedral angles; in particular, by definition we have $\tau^{(1)}_i \cong \tau_i + \theta^{(1)}_i \mod 2\pi$ and $\tau^{(2)}_i \cong \tau_i + \theta^{(2)}_i \mod 2\pi$ for all $i=1,\ldots m$. However, because $\tau^{(1)}_i = \tau^{(2)}_i$, this means $\theta^{(1)}_i \cong \theta^{(2)}_i$ for all $i$ and therefore $\boldsymbol{\theta}_1 = \boldsymbol{\theta}_2$ (as elements of $SO(2)^m$). Now denote $\cc^\star = A_\text{tor}(\boldsymbol{\theta}_1, \cc) = A_\text{tor}(\boldsymbol{\theta}_2, \cc)$ and apply $A_\text{tr}(-\rr_1, \cdot) = A_\text{tr}(-\rr_2, \cdot)$ to both sides of Equation \ref{eq:injective}. We then have
    \begin{equation}
        A_\text{rot}(R_1, \cc^\star) = A_\text{rot}(R_2, \cc^\star)
    \end{equation}
    which further leads to
    \begin{equation}
        \cc^\star - \bar\cc^\star = R_1^{-1}R_2(\cc^\star-\bar\cc^\star)
    \end{equation}
    In general, this does not imply that $R_1 = R_2$. However, $R_1 \neq R_2$ is possible only if $\cc^\star$ is \emph{degenerate}, in the sense that all points are collinear along the shared axis of rotation of $R_1, R_2$. However, in practice, conformers never consist of a collinear set of points, so we can safely assume $R_1 = R_2$. We now have $(\rr_1, R_1, \boldsymbol{\theta}_1) = (\rr_2, R_2, \boldsymbol{\theta}_2)$, or $g_1 = g_2$, contradicting our initial assumption. We thus conclude that $A(\cdot, \cc)$ is injective, completing the proof.
\end{proof}

\chapter{Methodological Details}

\section{Chapter 3: Torsional Diffusion}

\subsection{Score network architecture} \label{app:architecture}

\paragraph{Overview} To perform the torsion score prediction under these symmetry constraints we design an architecture formed by three components: an embedding layer, a series of $K$ interaction layers and a pseudotorque layer. The pseudotorque layer produces pseudoscalar torsion scores $\delta\tau := \partial \log p/\partial\tau$ for every rotatable bond.
Following the notation from Thomas et al. \cite{thomas2018tensor} and Batzner et al. \cite{batzner2021se}, we represent the node representations as $V_{acm}^{(k, l, p)}$ a dictionary with keys the layer $k$, rotation order $l$ and parity $p$ that contains tensors with shapes $[ |\mathcal{V}|, n_l, 2l+1 ]$ corresponding to the indices of the node, channel and representation respectively.  We use the \verb|e3nn| library \cite{e3nn} to implement our architecture.

\begin{figure}[h]
    \centering
    \includegraphics[width=0.8\textwidth]{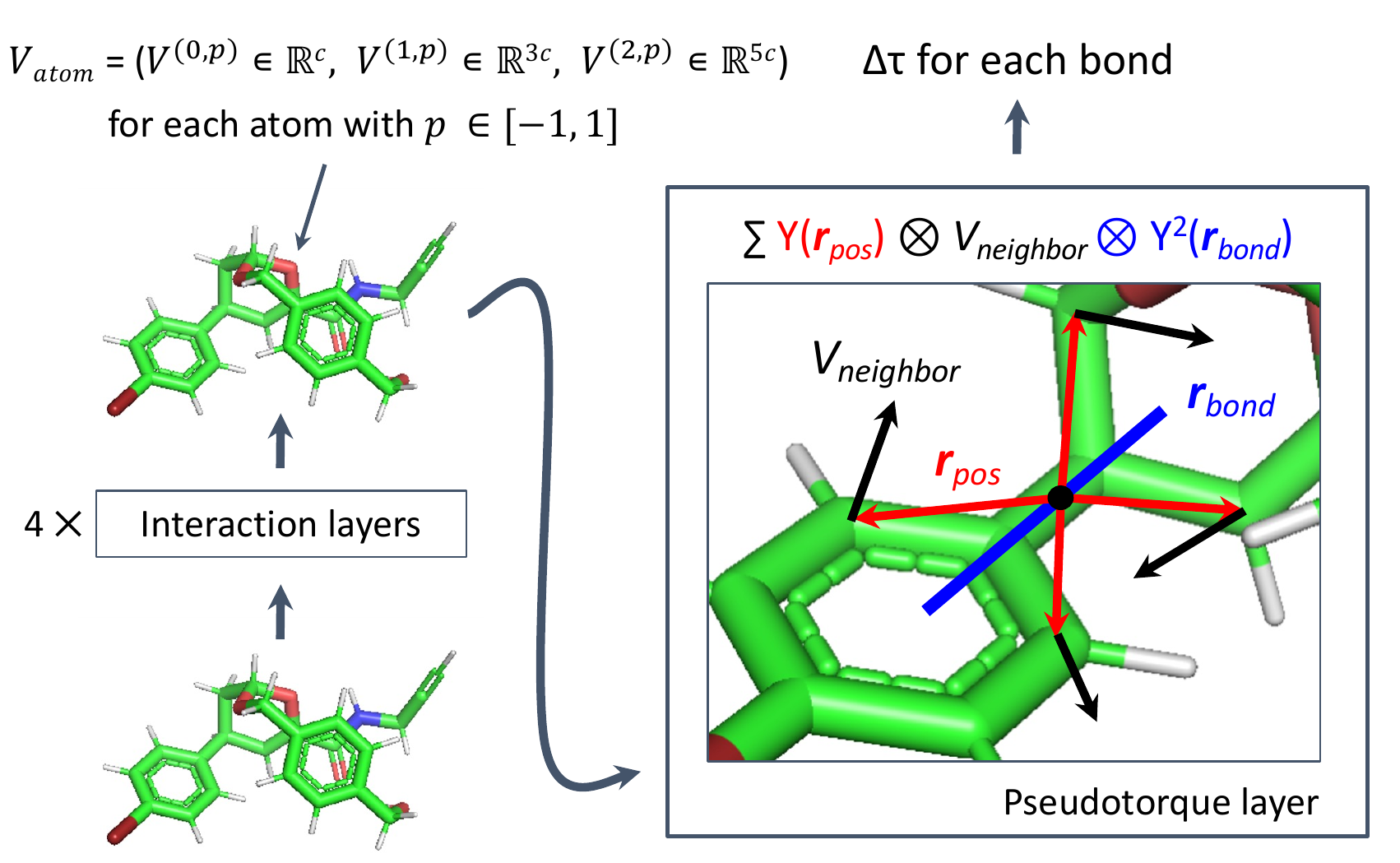}
    \caption{Overview of the architecture and visual intuition of the pseudotorque layer.}
    \label{fig:architecture}
\end{figure}

\paragraph{Embedding layer} In the embedding layer, we build a radius graph $(\mathcal{V}, \mathcal{E}_{r_{\max} })$ around each atom on top of the original molecular graph and generate initial scalar embeddings for nodes $V_a^{(0, 0, 1)}$ and edges $e_{ab}$ combining chemical properties, sinusoidal embeddings of time $\phi(t)$ \cite{vaswani2017attention} and, for the edges, a radial basis function representation of their length $\mu(r_{ab})$ \cite{schutt2017schnet}:
\begin{equation}
\begin{gathered}
\mathcal{E}_{r_{\max} } = \mathcal{E} \sqcup \{ (a,b) \mid r_{ab} < r_{\max} \}  \\
e_{ab} = \Upsilon^{(e)}(fe_{ab} || \mu(r_{ab}) || \phi(t)) \quad \forall (a,b) \in \mathcal{E}_{r_{\max}}  \\
V_a^{(0, 0, 1)} = \Upsilon^{(v)}(f_a || \phi(t))  \quad \forall a \in \mathcal{V}  
\end{gathered}
\end{equation}

where $\Upsilon^{(e)}$ and $\Upsilon^{(v)}$ are learnable two-layers MLPs, $r_{ab}$ is the Euclidean distance between atoms $a$ and $b$, $r_{\max} = 5\AA$ is the distance cutoff, $f_a$ are the chemical features of atom $a$, $f_{ab}$ are the chemical features of bond $(a,b)$ if it was part of $\mathcal{E}$ and 0 otherwise.

The node and edge chemical features $f_a$ and $f_{ab}$ are constructed as in Ganea et al. \cite{ganea2021geomol}. Briefly, the node features include atom identity, atomic number, aromaticity, degree, hybridization, implicit valence, formal charge, ring membership, and ring size, constituting a 74-dimensional vector for GEOM-DRUGS and 44-dimensional for QM9 (due to fewer atom types). The edge features are a 4 dimensional one-hot encoding of the bond type.

\paragraph{Interaction layers} The interaction layers are based on \texttt{e3nn} \cite{e3nn} convolutional layers. At each layer, for every pair of nodes in the graph, we construct messages using tensor products of the current irreducible representation of each node with the spherical harmonic representations of the normalized edge vector. These messages are themselves irreducible representations, which are weighted channel-wise by a scalar function of the current scalar representations of the two nodes and the edge and aggregated with Clebsch-Gordan coefficients.

At every layer $k$, for every node $a$, rotation order $l_o$, and output channel $c'$:
\begin{small}
\begin{equation}
\begin{gathered}
V_{ac'm_o}^{(k, l_o, p_o)} = \sum_{l_f, l_i, p_i} \sum_{m_f, m_i} C_{(l_i, m_i)(l_f, m_f)}^{(l_o,m_o)} \frac{1}{|\mathcal{N}_a|} \sum_{b \in \mathcal{N}_a} \sum_c \psi_{abc}^{(k, l_o, l_f, l_i, p_i)} \; Y_{m_f}^{(l_f)}(\hat{r}_{ab}) \; V_{bcm_i}^{(k-1,l_i,p_i)} \\
\text{with } \psi_{abc}^{(k,l_o, l_f, l_i, p_i)} = \Psi_c^{(k, l_o, l_f, l_i, p_i)}(e_{ab} || V_a^{(k-1, 0, 1)} || V_b^{(k-1, 0, 1)})
\end{gathered}\end{equation}\end{small}
where the outer sum is over values of $l_f, l_i, p_i$ such that $|l_i-l_f| \le l_o \le l_i+l_f$ and $(-1)^{l_f}p_i = p_o$, $C$ indicates the Clebsch-Gordan coefficients \cite{thomas2018tensor}, $\mathcal{N}_a = \{b \mid (a,b) \in \mathcal{E}_{\max} \}$ the neighborhood of $a$ and $Y$ the spherical harmonics. The rotational order of the nodes representations $l_o$ and $l_i$  and of the spherical harmonics of the edges ($l_f$) are restricted to be at most 2. All the learnable weights are contained in $\Psi$, a dictionary of MLPs that compute per-channel weights based on the edge embeddings and scalar features of the outgoing and incoming node.

\paragraph{Pseudotorque layer} The final part of our architecture is a pseudotorque layer that predicts a pseudoscalar score $\delta\tau$ for each rotatable bond from the per-node outputs of the interaction layers. For every rotatable bond, we construct a tensor-valued filter, centered on the bond, from the tensor product of the spherical harmonics with a $l=2$ representation of the \emph{bond axis}. Since the parity of the $l=2$ spherical harmonic is even, this representation does not require a choice of bond direction. The filter is then used to convolve with the representations of every neighbor on a radius graph, and the products which produce pseudoscalars are passed through odd-function (i.e., with $\tanh$ nonlinearity and no bias) dense layers (not shown in equation \ref{eq:pseudotorque}) to produce a single prediction.

For all rotatable bonds $g = (g_0, g_1) \in \mathcal{E}_{\text{rot}}$ and $b \in \mathcal{V}$, let $r_{gb}$ and $\hat{r}_{gb}$ be the magnitude and direction of the vector connecting the center of bond $g$ and $b$.
\begin{equation}\label{eq:pseudotorque}
\begin{gathered}
\mathcal{E}_\tau = \{ (g,b) \mid g \in \mathcal{E}_r, b \in \mathcal{V}, r_{gb} < r_{\max} \} \quad \quad e_{gb} = \Upsilon^{(\tau)}(\mu(r_{gb}))\\
T_{gbm_o}^{(l_o, p_o)} = \sum_{m_g, m_r, l_r: p_o=(-1)^{l_r}} C_{(2, m_g)(l_r, m_r)}^{(l_o,m_o)} Y_{m_f}^{(2)}(\hat{r}_{g}) \; Y_{m_r}^{(l_r)}(\hat{r}_{gb})\\
\delta\tau_{g} = \sum_{l, p_f, p_i: p_fp_i=-1} \sum_{m_o, m_i} C_{(l, m_f)(l, m_i)}^{(0,0)} \frac{1}{|\mathcal{N}_g|} \sum_{b \in \mathcal{N}_g} \sum_c \gamma_{gcb}^{(l, p_i)} \; T_{gbm_f}^{(l, p_f)} \; V_{bcm_i}^{(K,l,p_i)}  \\
\text{with } \gamma_{gcb}^{(l, p_i)} = \Gamma_c^{(l, p_i)}(e_{gb} || V_b^{(K, 0, 1)} || V_{g_0}^{(K, 0, 1)} + V_{g_1}^{(K, 0, 1)})
\end{gathered}\end{equation}

where $\Upsilon^{(\tau)}$ and $\Gamma$ are MLPs with learnable parameters and $\mathcal{N}_g = \{b \mid (g,b) \in \mathcal{E}_{\tau} \}$.


\subsection{Conformer matching} \label{app:matching}

The conformer matching procedure, summarised in Algorithm \ref{alg:matching}, proceeds as follows. For a molecule with $K$ conformers, we first generate $K$ random local structure estimates $\hat L$ from RDKit. To match with the ground truth local structures, we compute the cost of matching each true conformer $C$ with each estimate $\hat L$ (i.e. a $K\times K$ cost matrix), where the cost is the best RMSD that can be achieved by modifying the torsions of the RDKit conformer with local structure $\hat L$ to match the ground truth conformer $C$. Note that in practice, we compute an upper bound to this optimal RMSD using the fast von Mises torsion matching procedure proposed by Stark et al. \cite{equibind}. 

We then find an optimal matching of true conformers $C$ to local structure estimates $\hat{L}$ by solving the linear sum assignment problem over the approximate cost matrix \cite{crouse2016implementing}. Finally, for each matched pair, we find the true optimal $\hat{C}$ by running a differential evolution optimization procedure over the torsion angles \cite{mendez2021geometric}. The complete assignment resulting from the linear sum solution guarantees that there is no distributional shift in the local structures seen during training and inference. 

\begin{algorithm}[h]
\caption{Conformer matching}\label{alg:matching}
\KwIn{true conformers of $G$ $[C_1, ... C_K]$}
\KwOut{approximate conformers for training  $[\hat{C}_1, ... \hat{C}_K]$}
generate local structures $[\hat{L}_1, ... \hat{L}_K]$ with RDKit\;
\For{$(i,j)$ \textbf{in} $[1,K]\times [1,K]$}{
    $C_{\text{temp}}$ = von\_Mises\_matching($C_i$, $\hat{L}_j$)\;
    cost[i,j] = $\rmsd(C_i$, $C_{\text{temp}})$\;
}
assignment = linear\_sum\_assignment(cost)\;
\For{$i\leftarrow 1$ \KwTo $K$}{
    j = assignment[i]\;
    $\hat{C}_i$ = differential\_evolution($C_i$, $\hat{L}_j$, RMSD)\;
}
\end{algorithm}

Table \ref{tab:matching_results} shows the average RMSD between a ground truth conformer $C_i$ and its matched conformer $\hat{C}_i$. The average RMSD of 0.324 \AA $\,$ obtained via conformer matching provides an approximate lower bound on the achievable AMR performance for methods that do not change the local structure and take those from RDKit (further discussion in Appendix \ref{app:discuss_rdkit}).

\begin{table}[t]
 \caption{Average $\rmsd(c_i, \hat{c}_i)$ achieved by different variants of conformer matching. "Original RDKit" refers to the RMSD between a random RDKit conformer and a ground truth conformer without any optimization. In "Von Mises optimization" and "Differential evolution," the torsions of the RDKit conformer are adjusted using the respective procedures, but the pairing of RDKit and ground truth conformers is still random. In "Conformer matching," the cost-minimizing assignment prior to differential evolution provides a 15\% improvement in average RMSD. The results are shown for a random 300-molecule subset of GEOM-DRUGS.}
 \label{tab:matching_results}
 \vspace{10pt}
 \centering
 \begin{tabular}{lc} \toprule
 Matching method & RMSD (\AA) \\ \midrule
 Original RDKit           & 1.448     \\
 Von Mises optimization   & 0.728     \\
 Differential evolution   & 0.379     \\
 Conformer matching       & 0.324    \\ \bottomrule
 \end{tabular}
 \end{table}

\section{Chapter 4: DiffDock} 

\subsection{Training and Inference} \label{app:train_inf}

In this section we present the training and inference procedures of the diffusion generative model. First, however, there are a few subtleties of the generative approach to molecular docking that are worth mentioning. Unlike the standard generative modeling setting where the dataset consists of many samples drawn from the data distribution, each training example $(\xx^\star, \yy)$ of protein structure $\yy$ and ground-truth ligand pose $\xx^\star$ is the \emph{only sample} from the corresponding conditional distribution $p_{\xx^\star}(\cdot \mid \yy)$ defined over $\mathcal{M}_{\xx^\star}$. Thus, the innermost training loop iterates over distinct conditional distributions $p_{\xx^\star}(\cdot \mid \yy)$, along with a single sample from that distribution, rather than over samples from a common data distribution $p_\text{data}(\xx)$.

As discussed in Section~\ref{sec:diffusion_model}, during inference, $\cc$ is the ligand structure generated with a method such as RDKit. However, during training we require $\mathcal{M}_\cc = \mathcal{M}_{\xx^\star}$ in order to define a bijection between $\cc \in \mathcal{M}_{\xx^\star}$ and $\PP$. If we take $\cc \in \mathcal{M}_{\xx^\star}$, there will be a distribution shift between the manifolds $\mathcal{M}_\cc$ considered at training time and those considered at inference time. To circumvent this issue, at training time we predict $\cc$ with RDKit and replace $\xx^\star$ with $\argmin_{\xx^\dagger \in \mathcal{M}_\cc} \rmsd(\xx^\star, \xx^\dagger)$ using the conformer matching procedure described in Jing et al. \cite{jing2022torsional}.

The above paragraph may be rephrased more intuitively as follows: during inference, the generative model docks a ligand structure generated by RDKit, keeping its non-torsional degrees of freedom (e.g., local structures) fixed. At training time, however, if we train the score model with the local structures of the ground truth pose, this will not correspond to the local structures seen at inference time. Thus, at training time, we replace the ground truth pose by generating a ligand structure with RDKit and aligning it to the ground truth pose while keeping the local structures fixed.

With these preliminaries, we now continue to the full procedures (Algorithms~\ref{alg:diffdock_training} and \ref{alg:inference_training}). The training and inference procedures of a score-based diffusion generative model on a Riemannian manifold consist of (1) sampling and regressing against the score of the diffusion kernel during training; and (2) sampling a geodesic random walk with the score as a drift term during inference \cite{de2022riemannian}. Because we have developed the diffusion process on $\PP$ but continue to provide the score model with elements in $\mathcal{M}_\cc \subset \mathbb{R}^{3n}$, the full training and inference procedures involve repeatedly interconverting between the two spaces using the bijection given by the seed conformation $\cc$.

\begin{algorithm}[h]
\caption{Training procedure (single epoch)}\label{alg:diffdock_training}
\KwIn{Training pairs $\{(\xx^\star, \yy)\}$, RDKit predictions $\{\cc\}$}
\ForEach{$\cc, \xx^\star, \yy$}{
    Let $\xx_0 \gets \argmin_{\xx^\dagger \in \mathcal{M}_\cc} \rmsd(\xx^\star, \xx^\dagger)$\;
    Compute $(\rr_0, R_0, \boldsymbol{\theta}_0) \gets A^{-1}_\cc(\xx_0)$\;
    Sample $t \sim \uni([0,1])$\;
    Sample $\Delta\rr, \Delta R, \Delta \boldsymbol{\theta}$ from diffusion kernels $p^\text{tr}_t(\cdot \mid 0), p^\text{rot}_t(\cdot \mid 0), p^\text{tor}_t(\cdot \mid 0)$\;
    Set $\rr_t \gets \rr_0 + \Delta \rr$\;
    Set $R_t \gets (\Delta R)R_0$\;
    Set $\boldsymbol{\theta}_t \gets \boldsymbol{\theta}_0 + \Delta\boldsymbol{\theta}\mod 2\pi$\;
    Compute $\xx_t \gets A((\rr_t, R_t, \boldsymbol{\theta}_t), \cc)$\;
    Predict scores $\alpha \in \mathbb{R}^3, \beta \in \mathbb{R}^3, \gamma \in \mathbb{R}^m = \ss(\xx_t, \cc, \yy, t)$ \;
    Take optimization step on loss $\mathcal{L} = \lnorm\alpha - \nabla p^\text{tr}_t(\Delta \rr \mid 0) \rnorm^2 + \lnorm\beta - \nabla p^\text{rot}_t(\Delta R \mid 0) \rnorm^2 + \lnorm\gamma - \nabla p^\text{tor}_t(\Delta\boldsymbol{\theta} \mid 0) \rnorm^2$
}
\end{algorithm}

\begin{algorithm}[h]
\caption{Inference procedure}\label{alg:inference_training}
\KwIn{RDKit prediction $\cc$, protein structure $\yy$ (both centered at origin)}
\KwOut{Sampled ligand pose $\xx_0$}
Sample $\boldsymbol{\theta}_N \sim \uni(SO(2)^m)$, $ R_N \sim \uni(SO(3))$, $\rr_N \sim \mathcal{N}(0, \sigma_\text{tor}^2(T))$\;
Let $\xx_N = A((\rr_N, R_N, \boldsymbol{\theta}_N), \cc)$\;
\For{n $\leftarrow N$ \KwTo $1$}{
    Let $t = n/N$ and $\Delta \sigma^2_\text{tr} = \sigma^2_\text{tr}(n/N) - \sigma^2_\text{tr}((n-1)/N)$ and similarly for $\Delta \sigma^2_\text{rot}, \Delta \sigma^2_\text{tor}$\;
    Predict scores $\alpha \in \mathbb{R}^3, \beta \in \mathbb{R}^3, \gamma \in \mathbb{R}^m \gets \ss(\xx_n, \cc, \yy, t)$\;
    Sample $\mathbf{z}_\text{tr}, \mathbf{z}_\text{rot}, \mathbf{z}_\text{tor}$ from $\mathcal{N}(0, \Delta\sigma^2_\text{tr}),\mathcal{N}(0, \Delta\sigma^2_\text{rot}), \mathcal{N}(0, \Delta\sigma^2_\text{tor}) $ respectively\;
    Set $\rr_{n-1} \gets \rr_0 + \Delta\sigma^2_\text{tr} \alpha + \mathbf{z}_\text{tr}$\;
    Set $R_{n-1} \gets \mathbf{R}(\Delta\sigma^2_\text{rot} \beta + \mathbf{z}_\text{rot})R_n)$\;
    Set $\boldsymbol{\theta}_{n-1} \gets \boldsymbol{\theta}_{n} + (\Delta\sigma^2_\text{tor} \gamma + \mathbf{z}_\text{tor})\mod 2\pi$\;
    Compute $\xx_{n-1} \gets A((\rr_{n-1}, R_{n-1}, \boldsymbol{\theta}_{n-1}), \cc)$\;
}
Return $\xx_0$\;
\end{algorithm}

However, as noted in the main text, the dependence of these procedures on the exact choice of $\cc$ is potentially problematic, as it suggests that at inference time, the model distribution may be different depending on the orientation and torsion angles of $\cc$. Simply removing the dependence of the score model on $\cc$ is not sufficient since the update steps themselves still occur on $\PP$ and require a choice of $\cc$ to be mapped to $\mathcal{M}_\cc$. However, notice that the update steps---in both training and inference---consist of (1) sampling the diffusion kernels at the origin; (2) applying these updates to the point on $\PP$; and (3) transferring the point on $\PP$ to $\mathcal{M}_\cc$ via $A(\cdot, \cc)$. Might it instead be possible to apply the updates to 3D ligand poses $\xx \in \mathcal{M}_\cc$ \emph{directly}?

It turns out that the notion of applying these steps to ligand poses ``directly'' corresponds to the formal notion of \emph{group action}. The operations $A_\text{tr}, A_\text{rot}, A_\text{tor}$ that we have already defined are formally group actions if they satisfy $A_{(\cdot)}(g_1g_2, \xx) = A(g_1, A(g_2, \xx))$. While true for $A_\text{tr}, A_\text{rot}$, this is not generally true for $A_\text{tor}$ if we take $SO(2)^m$ to be the direct product group; however, the approximation is increasingly good as the magnitude of the torsion angle updates decreases. If we then define $\PP$ to be the direct product group of its constituent groups, $A$ is a group action of $\PP$ on $\mathcal{M}_\cc$, as the operations of $A_\text{tr}, A_\text{rot}, A_\text{tor}$ commute and are (under the approximation) individually group actions.

The implication of $A$ being a group action can be seen as follows. Let $\delta = g_bg_a^{-1}$ be the update which brings $g_a \in \PP$ to $g_b \in \PP$ via left multiplication, and let $\xx_a, \xx_b$ be the corresponding ligand poses $A(g_a,\cc), A(g_b,\cc)$. Then
\begin{equation}
    \xx_b = A(g_bg_a^{-1}g_a, \cc) = A(\delta, \xx_a)
\end{equation}
which means that the updates $\delta$ can be applied directly to $\xx_a$ using the operation $A$. The training and inference procedures then become Algorithm \ref{alg:approx-training} and \ref{alg:approx-inference} below. The initial conformer $\cc$ is no longer used, except in the initial steps to define the manifold---to find the closest point to $\xx^\star$ in training, and to sample $\xx_N$ from the prior over $\mathcal{M}_\cc$ in inference.

Conceptually speaking, this procedure corresponds to ``forgetting'' the location of the origin element on $\mathcal{M}_\cc$, which is permissible because a change of the origin to some equivalent seed $\cc' \in \mathcal{M}_\cc$ merely translates---via right multiplication by $A^{-1}_\cc(\cc')$---the original and diffused data distributions on $\PP$, but does not cause any changes on $\mathcal{M}_\cc$ itself. The training and inference routines involve updates---formally left multiplications---to group elements, but as left multiplication on the group corresponds to group actions on $\mathcal{M}_\cc$, the updates can act on $\mathcal{M}_\cc$ directly, without referencing the origin $\cc$.

We find that the approximation of $A$ as a group action works quite well in practice and use Algorithms \ref{alg:approx-training} and \ref{alg:approx-inference} for all training and experiments discussed in the paper. Of course, disentangling the torsion updates from rotations in a way that makes $A_\text{tor}$ exactly a group action would justify the procedure further, and we regard this as a possible direction for future work.

\begin{algorithm}[h]
\caption{Approximate training procedure (single epoch)}\label{alg:approx-training}
\KwIn{Training pairs $\{(\xx^\star, \yy)\}$, RDKit predictions $\{\cc\}$}
\ForEach{$\cc, \xx^\star, \yy$}{
    Let $\xx_0 \gets \argmin_{\xx^\dagger \in \mathcal{M}_\cc} \rmsd(\xx^\star, \xx^\dagger)$\;
    Sample $t \sim \uni([0,1])$\;
    Sample $\Delta\rr, \Delta R, \Delta \boldsymbol{\theta}$ from diffusion kernels $p^\text{tr}_t(\cdot \mid 0), p^\text{rot}_t(\cdot \mid 0), p^\text{tor}_t(\cdot \mid 0)$\;
    Compute $\xx_t \gets A((\Delta\rr, \Delta R, \Delta \boldsymbol{\theta}), \xx_0)$\;
    Predict scores $\alpha \in \mathbb{R}^3, \beta \in \mathbb{R}^3, \gamma \in \mathbb{R}^m = \ss(\xx_t, \yy, t)$ \;
    Take optimization step on loss $\mathcal{L} = \lnorm\alpha - \nabla p^\text{tr}_t(\Delta \rr \mid 0) \rnorm^2 + \lnorm\beta - \nabla p^\text{rot}_t(\Delta R \mid 0) \rnorm^2 + \lnorm\gamma - \nabla p^\text{tor}_t(\Delta\boldsymbol{\theta} \mid 0) \rnorm^2$
}
\end{algorithm}

\begin{algorithm}[h]
\caption{Approximate inference procedure}\label{alg:approx-inference}
\KwIn{RDKit prediction $\cc$, protein structure $\yy$ (both centered at origin)}
\KwOut{Sampled ligand pose $\xx_0$}
Sample $\boldsymbol{\theta}_N \sim \uni(SO(2)^m)$, $ R_N \sim \uni(SO(3))$, $\rr_N \sim \mathcal{N}(0, \sigma_\text{tor}^2(T))$\;
Let $\xx_N = A((\rr_N, R_N, \boldsymbol{\theta}_N), \cc)$\;
\For{n $\leftarrow N$ \KwTo $1$}{
    Let $t = n/N$ and $\Delta \sigma^2_\text{tr} = \sigma^2_\text{tr}(n/N) - \sigma^2_\text{tr}((n-1)/N)$ and similarly for $\Delta \sigma^2_\text{rot}, \Delta \sigma^2_\text{tor}$\;
    Predict scores $\alpha \in \mathbb{R}^3, \beta \in \mathbb{R}^3, \gamma \in \mathbb{R}^m \gets \ss(\xx_n, \yy, t)$\;
    Sample $\mathbf{z}_\text{tr}, \mathbf{z}_\text{rot}, \mathbf{z}_\text{tor}$ from $\mathcal{N}(0, \Delta\sigma^2_\text{tr}),\mathcal{N}(0, \Delta\sigma^2_\text{rot}), \mathcal{N}(0, \Delta\sigma^2_\text{tor}) $ respectively\;
    Set $\Delta\rr \gets \rr_0 + \Delta\sigma^2_\text{tr} \alpha + \mathbf{z}_\text{tr}$\;
    Set $\Delta R \gets  \mathbf{R}(\Delta\sigma^2_\text{rot} \beta + \mathbf{z}_\text{rot})$\;
    Set $\Delta\boldsymbol{\theta} \gets \Delta\sigma^2_\text{tor} \gamma + \mathbf{z}_\text{tor}$\;
    Compute $\xx_{n-1} \gets A((\Delta\rr, \Delta R, \Delta\boldsymbol{\theta}), \xx_n)$\;
}
Return $\xx_0$\;
\end{algorithm}

\subsection{Architecture Details} \label{app:diffdock_architecture}

We use convolutional networks based on tensor products of irreducible representations (irreps) of $SO(3)$ \cite{thomas2018tensor} as architecture for both the score and confidence models. In particular, these are implemented using the \verb|e3nn| library \cite{e3nn}. Below, $\otimes_w$ refers to the spherical tensor product of irreps with path weights $w$, and $\oplus$ refers to normal vector addition (with possibly padded inputs). Features have multiple channels for each irrep.
Both the architectures can be decomposed into three main parts: embedding layer, interaction layers, and output layer. We outline each of them below.

\subsubsection{Embedding layer}

\paragraph{Geometric heterogeneous graph.} Structures are represented as heterogeneous geometric graphs with nodes representing ligand (heavy) atoms, receptor residues (located in the position of the $\alpha$-carbon atom), and receptor (heavy) atoms (only for the confidence model). Because of the high number of nodes involved, it is necessary for the graph to be sparsely connected for runtime and memory constraints. Moreover, sparsity can act as a useful inductive bias for the model, however, it is critical for the model to find the right pose that nodes that might have a strong interaction in the final pose to be connected during the diffusion process. Therefore, to build the radius graph, we connect nodes using cutoffs that are dependent on the types of nodes they are connecting:
\begin{enumerate}
    \item Ligand atoms-ligand atoms, receptor atoms-receptor atoms, and ligand atoms-receptor atoms interactions all use a cutoff of 5\AA{}, standard practice for atomic interactions. For the ligand atoms-ligand atoms interactions we also preserve the covalent bonds as separate edges with some initial embedding representing the bond type (single, double, triple and aromatic). For receptor atoms-receptor atoms interactions, we limit at 8 the maximum number of neighbors of each atom. Note that the ligand atoms-receptor atoms only appear in the confidence model where the final structure is already set.
    \item Receptor residues-receptor residues use a cutoff of 15 \AA{} with 24 as the maximum number of neighbors for each residue.
    \item Receptor residues-ligand atoms use a cutoff of $20+3*\sigma_{tr}$ \AA{} where $\sigma_{tr}$ represents the current standard deviation of the diffusion translational noise present in each dimension (zero for the confidence model). Intuitively this guarantees that with high probability, any of the ligands and receptors that will be interacting in the final pose the diffusion model converges to are connected in the message passing at every step.
    \item Finally, receptor residues are connected to the receptor atoms that form the corresponding amino-acid.
\end{enumerate}

\paragraph{Node and edge featurization.} For the receptor residues, we use the residue type as a feature as well as a language model embedding obtained from ESM2 \cite{Lin2022ESM2}. The ligand atoms have the following features: atomic number; chirality; degree; formal charge; implicit valence; the number of connected hydrogens; the number of radical electrons; hybridization type; whether or not it is in an aromatic ring; in how many rings it is; and finally, 6 features for whether or not it is in a ring of size 3, 4, 5, 6, 7, or 8. These are concatenated with sinusoidal embeddings of the diffusion time \cite{vaswani2017attention} and, in the case of edges, radial basis embeddings of edge length \cite{schutt2017schnet}. These scalar features of each node and edge are then transformed with learnable two-layer MLPs (different for each node and edge type) into a set of scalar features that are used as initial representations by the interaction layers.

\paragraph{Notation} Let $(\mathcal{V}, \mathcal{E})$ represent the heterogeneous graph, with $\mathcal{V} = (\mathcal{V}_\ell, \mathcal{V}_r)$ respectively ligand atoms and receptor residues (receptor atoms $\mathcal{V}_a$, present in the confidence model, are for simplicity not included here), and similarly $\mathcal{E}=(\mathcal{E}_{\ell\ell},\mathcal{E}_{\ell r},\mathcal{E}_{r\ell},\mathcal{E}_{rr})$. Let $\mathbf{h}_a$ be the node embeddings (initially only scalar channels) of node $a$, $e_{ab}$ the edge embeddings of $(a,b)$, and $\mu(r_{ab})$ radial basis embeddings of the edge length. Let $\sigma_{tr}^2$, $\sigma_{rot}^2$, and $\sigma_{tor}^2$ represent the variance of the diffusion kernel in each of the three components: translational, rotational and torsional.

\subsubsection{Interaction layers}

At each layer, for every pair of nodes in the graph, we construct messages using tensor products of the current node features with the spherical harmonic representations of the edge vector. The weights of this tensor product are computed based on the edge embeddings and the \emph{scalar} features---denoted $\mathbf{h}^0_a$---of the outgoing and incoming nodes. The messages are then aggregated at each node and used to update the current node features. For every node $a$ of type $t_a$:
\begin{equation}
\begin{gathered}
\mathbf{h}_a \leftarrow \mathbf{h}_a \underset{t\in \{ \ell, r\}}{\oplus} \textsc{BN}^{(t_a,t)} \Bigg( \frac{1}{|\mathcal{N}_a^{(t)}|}\sum_{b \in \mathcal{N}_a^{(t)}} Y(\hat r_{ab}) \; \otimes_{\psi_{ab}} \; \mathbf{h}_b \Bigg) \\
\text{with} \; \psi_{ab} = \Psi^{(t_a,t)}(e_{ab}, \mathbf{h}^0_a, \mathbf{h}^0_b)
\end{gathered}
\end{equation}
Here, $t$ indicates an arbitrary node type, $\mathcal{N}_a^{(t)} = \{b \mid (a,b) \in \mathcal{E}_{t_a t}\}$ the neighbors of $a$ of type $t$, $Y$ are the spherical harmonics up to $\ell=2$, and $\textsc{BN}$ the (equivariant) batch normalisation. The orders of the output are restricted to a maximum of $\ell=1$. All learnable weights are contained in $\Psi$, a dictionary of MLPs, which uses different sets of weights for different edge types (as an ordered pair so four types for the score model and nine for the confidence) and different rotational orders. 

\subsubsection{Output layer} 

The ligand atom representations after the final interaction layer are used in the output layer to produce the required outputs. This is where the score and confidence architecture differ significantly. On one hand, the score model's output is in the tangent space $T_\rr \mathbb{T}_3 \oplus T_R SO(3) \oplus T_{\boldsymbol{\theta}} SO(2)^m$. This corresponds to having two $SE(3)$-equivariant output vectors representing the translational and rotational score predictions and $m$ $SE(3)$-invariant output scalars representing the torsional score. For each of these, we design final tensor-product convolutions inspired by classical mechanics. On the other hand, the confidence model outputs a single $SE(3)$-invariant scalar representing the confidence score. Below we detail how each of these outputs is generated.

\paragraph{Translational and rotational scores.} The translational and rotational score intuitively represent, respectively, the linear acceleration of the center of mass of the ligand and the angular acceleration of the rest of the molecule around the center. Considering the ligand as a rigid object and given a set of forces and masses at each ligand, a tensor product convolution between the atoms and the center of mass would be capable of computing the desired quantities. Therefore, for each of the two outputs, we perform a convolution of each of the ligand atoms with the (unweighted) center of mass $c$. \begin{equation}
\begin{gathered}
\mathbf{v} \leftarrow \frac{1}{|\mathcal{V}_\ell|}\sum_{a \in \mathcal{V}_\ell} Y(\hat r_{ca}) \; \otimes_{\psi_{ca}} \; \mathbf{h}_a \\
\text{with} \; \psi_{ca} = \Psi(\mu(r_{ca}), \mathbf{h}^0_a)
\end{gathered}
\end{equation}
We restrict the output of $\mathbf{v}$ to a single odd and a single even vectors (for each of the two scores). Since we are using coarse-grained representations of the protein, the score will neither be even nor odd; therefore, we sum the even and odd vector representations of $\mathbf{v}$. Finally, the magnitude (but not direction) of these vectors is adjusted with an MLP taking as input the current magnitude and the sinusoidal embeddings of the diffusion time. Finally, we (revert the normalization) by multiplying the outputs by $1/\sigma_{tr}$ for the translational score and by the expected magnitude of a score in $SO(3)$ with diffusion parameter $\sigma_{rot}$ (precomputed numerically). 

\paragraph{Torsional score. } To predict the $m$ $SE(3)$-invariant scalar describing the torsional score, we use a pseudotorque layer similar to that of Jing et al. \cite{jing2022torsional}. This predicts a scalar score $\delta\tau$ for each rotatable bond from the per-node outputs of the atomic convolution layers. For rotatable bond $g = (g_0, g_1)$ and $b \in \mathcal{V_\ell}$, let $r_{gb}$ and $\hat{r}_{gb}$ be the magnitude and direction of the vector connecting the center of bond $g$ and $b$. We construct a convolutional filter $T_g$ for each bond $g$ from the tensor product of the spherical harmonics with a $\ell=2$ representation of the {bond axis} $\hat{r}_g$:\footnote{Since the parity of the $\ell=2$ spherical harmonic is even, this representation is indifferent to the choice of bond direction.}
\begin{equation}
T_g(\hat{r}) := Y^2(\hat{r}_{g}) \otimes Y(\hat{r})
\end{equation}
$\otimes$ is the full (i.e., unweighted) tensor product as described in Geiger et al. \cite{geiger2022e3nn}, and the second term contains the spherical harmonics up to $\ell=2$ (as usual). This filter (which contains orders up to $\ell=3$) is then used to convolve with the representations of every neighbor on a radius graph:
\begin{equation}
\begin{gathered}
\mathcal{E}_\tau = \{ (g,b) \mid g \text{ a rotatable bond}, b \in \mathcal{V}_\ell\} \\ 
\quad e_{gb} = \Upsilon^{(\tau)}(\mu(r_{gb})) \quad \forall (g,b)\in\mathcal{E}_\tau\\
\mathbf{h}_g = \frac{1}{|\mathcal{N}_g|} \sum_{b \in \mathcal{N}_g} T_{g}(\hat r_{gb}) \otimes_{\gamma_{gb}} \mathbf{h}_b \\
\text{with} \; \gamma_{gb} = \Gamma(e_{gb}, \mathbf{h}_{b}^0, \mathbf{h}_{g_0}^{0} + \mathbf{h}_{g_1}^0)
\end{gathered}\end{equation}
Here, $\mathcal{N}_g = \{b \mid (g,b) \in \mathcal{E}_{\tau} \}$ and $\Upsilon^{(\tau)}$ and $\Gamma$ are MLPs with learnable parameters. Since unlike Jing et al. \cite{jing2022torsional}, we use coarse-grained representations the parity also here is neither even nor odd, the irreps in the output are restricted to arrays both even $\mathbf{h}'_g$ and odd $\mathbf{h}''_g$ scalars. Finally, we produce a single scalar prediction for each bond:
\begin{equation}
\delta\tau_g = \Pi(\mathbf{h}'_g+\mathbf{h}''_g)
\end{equation}
where $\Pi$ is a two-layer MLP with $\tanh$ nonlinearity and no biases. This is also ``denormalized" by multiplying by the expected magnitude of a score in $SO(2)$ with diffusion parameter $\sigma_{tor}$.

\paragraph{Confidence output.} The single $SE(3)$-invariant scalar representing the confidence score output is instead obtained by concatenating the even and odd final scalar representation of each ligand atom, averaging these feature vectors among the different atoms, and finally applying a three layers MLP (with batch normalization).

\chapter{Further Discussion}

\section{Chapter 3: Torsional Diffusion}  \label{app:discuss}

\subsection{RDKit local structures} \label{app:discuss_rdkit}

In this section, we provide empirical justification for the claim that cheminformatics methods like RDKit already provide accurate local structures. It is well known in chemistry that bond lengths and angles take on a very narrow range of values due to strong energetic constraints. However, it is not trivial to empirically evaluate the claim due to the difficulty in defining a distance measure between a pair of local structures. In this section, we will employ two sets of observations: marginal error distributions and matched conformer RMSD.

\begin{figure}[h!]
    \centering
    \includegraphics[width=0.49\textwidth]{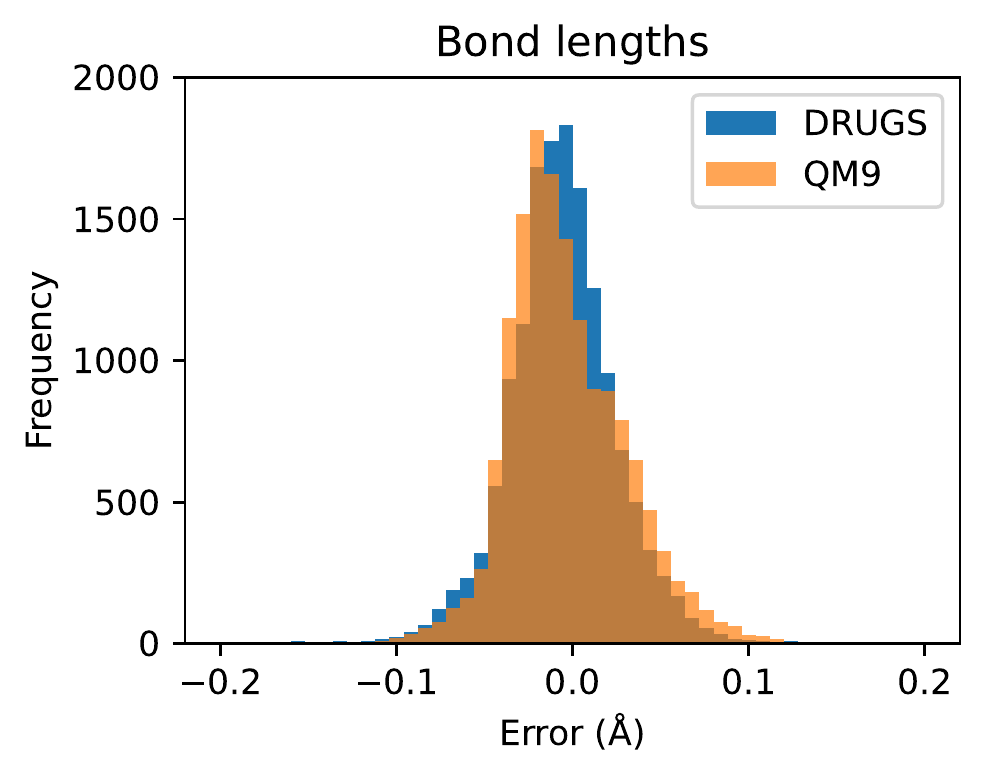}
    \includegraphics[width=0.49\textwidth]{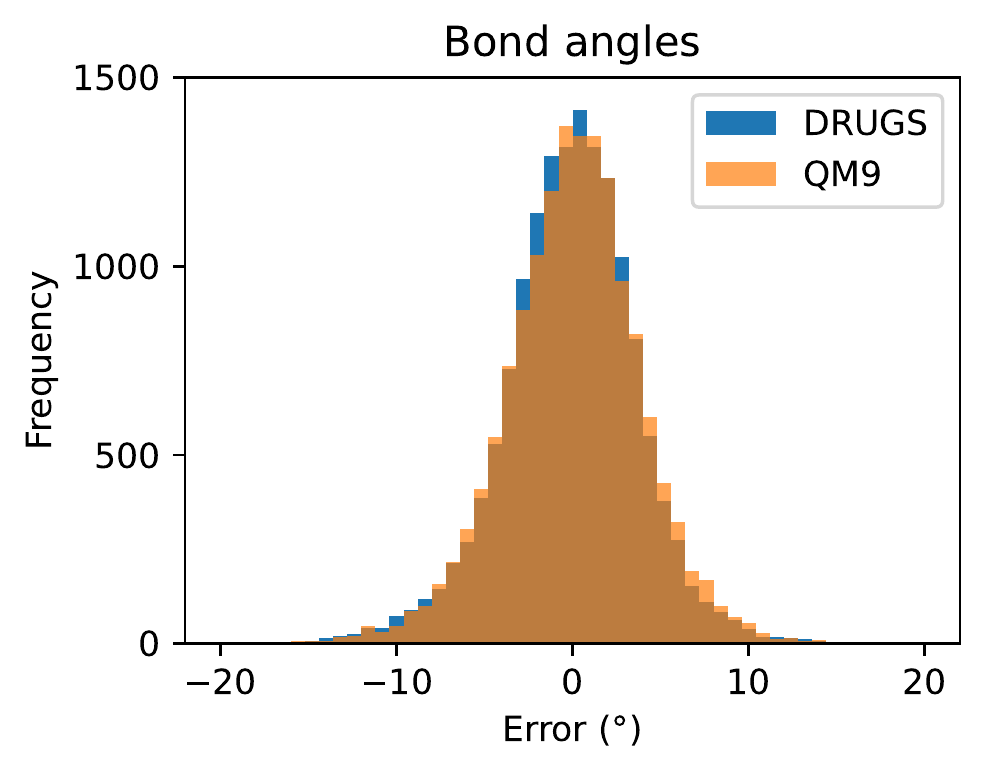}
    \caption{Histogram of the errors in 15000 predicted bond lengths and angles from randomly sampled molecules in GEOM-DRUGS and GEOM-QM9.}
    \label{fig:marginal_local_errors}
\end{figure}

\paragraph{Marginal error distributions} We examine the distribution of errors of the bond lengths and angles in a random RDKit conformer relative to the corresponding lengths and angles in a random CREST conformer (Figure \ref{fig:marginal_local_errors}). The distributions are narrow and uni-modal distributions around zero, with a RMSE of 0.03 \AA \, for bond lengths and 4.1\textdegree\, for bond angles on GEOM-DRUGS. Comparing DRUGS and QM9, the error distribution does not depend on the size of the molecule. Although it is difficult to determine how these variations will compound or compensate for each other in the global conformer structure, the analysis demonstrates that bond lengths and angles have little flexibility (i.e., no strong variability among conformers) and are accurately predicted by RDKit.

\paragraph{Matched conformer RMSD} We can more rigorously analyze the quality of a local structure $\hat{L}$ with respect to a given reference conformer $C$ by computing the minimum RMSD that can be obtained by combining $\hat{L}$ with optimal torsion angles. That is, we consider the RMSD distance of $C$ to the closest point on the manifold of possible conformers with local structure $\hat{L}$: $\rmsd_{\min}(C, \hat L) := \min_\tau \rmsd(C, \hat{C})$ where $\hat{C} = (\hat{L}, \tau)$. 

Conveniently, $\hat{C}$ is precisely the output of the differential evolution in Appendix \ref{app:matching}. Thus, the average RSMD reported in the last row of Table \ref{tab:matching_results} is the expected $\rmsd_{\min}$ of an optimal assignment of RDKit local structures to ground-truth conformers. This distance---0.324 \AA \, on GEOM-DRUGS---is significantly smaller than the error of the current state-of-the-art conformer generation methods. Further, it is only slightly larger than the average $\rmsd_{\min}$ of 0.284 \AA \, resulting from matching a ground truth conformer to the local structure of another randomly chosen ground truth conformer, which provides a measure of the variability among ground truth local structures. These observations support the claim that the accuracy of existing approaches on drug-like molecules can be significantly improved via better conditional sampling of torsion angles.

\subsection{Limitations of torsional diffusion}



As demonstrated in Section \ref{sec:experiments}, torsional diffusion significantly improves the accuracy and reduces the denoising runtime for conformer generation. However, torsional diffusion also has a number of limitations that we will discuss in this section.

\paragraph{Conformer generation} The first clear limitation is that the error that torsional diffusion can achieve is lower bounded by the \textit{quality of the local structure} from the selected cheminformatics method. As discussed in Appendix~\ref{app:discuss_rdkit}, this corresponds to the mean RMSD obtained after conformer matching, which is 0.324 \AA \, with RDKit local structures on DRUGS. Moreover, due to the the local structure \textit{distributional shift}, conformer matching (or another method bridging the shift) is required to generate the training set. However, the resulting conformers are not the minima of the (unconditional or even conditional) potential energy function. Thus, the learning task becomes less physically interpretable and potentially more difficult; empirically we observe this clearly in the training and validation score-matching losses. We leave to future work the exploration of \textit{relaxations} of the rigid local structures assumption in a way that would still leverage the predominance of torsional flexibility in molecular structures, while at the same time allowing some flexibility in the independent components.

\paragraph{Rings} The largest source of flexibility in molecular conformations that is not directly accounted for by torsional diffusion is the variability in \textit{ring conformations}. Since the torsion angles at bonds inside cycles cannot be independently varied, our framework treats them as part of the local structure. Therefore, torsional diffusion relies on the local structure sampler $p_G(L)$ to accurately model cycle conformations. Although this is true for a large number of relatively small rings (especially aromatic ones) present in many drug-like molecules, it is less true for puckered rings, fused rings, and larger cycles. In particular, torsional diffusion does not address the longstanding difficulty that existing cheminformatics methods have with macrocycles---rings with 12 or more atoms that have found several applications in drug discovery \cite{driggers2008exploration}. We hope, however, that the idea of restricting diffusion processes to the main sources of flexibility will motivate future work to define diffusion processes over cycles conformations combined with free torsion angles.

\paragraph{Boltzmann generation} With Boltzmann generators we are typically interested in sampling the Boltzmann distribution over the entire (Euclidean) conformational space $p_G(C)$. However, the procedure detailed in Section \ref{sec:energy} generates (importance-weighted) samples from the Boltzmann distribution \textit{conditioned} on a given local structure $p_G(C \mid L)$. To importance sample from the full Boltzmann distribution $p_G(C)$, one would need a model $p_G(L)$ over local structures that also provides exact likelihoods. This is not the case with RDKit or, to the best of our knowledge, other existing models, and therefore an interesting avenue for future work.

\chapter{Experimental Details}

\section{Chapter 3: Torsional Diffusion} \label{app:torsional_exp_details}

\subsection{Dataset details}

\paragraph{Splits} We follow the data processing and splits from Ganea et al. \cite{ganea2021geomol}. The splits are random with train/validation/test of 243473/30433/1000 for GEOM-DRUGS and 106586/13323/1000 for GEOM-QM9. GEOM-XL consists of only a test split (since we do not train on it), which consists of all 102 molecules in the MoleculeNet dataset with at least 100 atoms. For all splits, the molecules whose CREST conformers all have a canonical SMILES different from the SMILES of the molecule (meaning a reacted conformer), or that cannot be handled by RDKit, are filtered out. 

\paragraph{Dataset statistics} As can be seen in Figure \ref{fig:data_stats}, the datasets differ significantly in molecule size as measured by number of atoms or rotatable bonds. Particularly significant is the domain shift between DRUGS and XL, which we leverage in our experiments by testing how well models trained on DRUGS generalize to XL.

\begin{figure}[h!]
    \centering
    \includegraphics[width=0.49\textwidth]{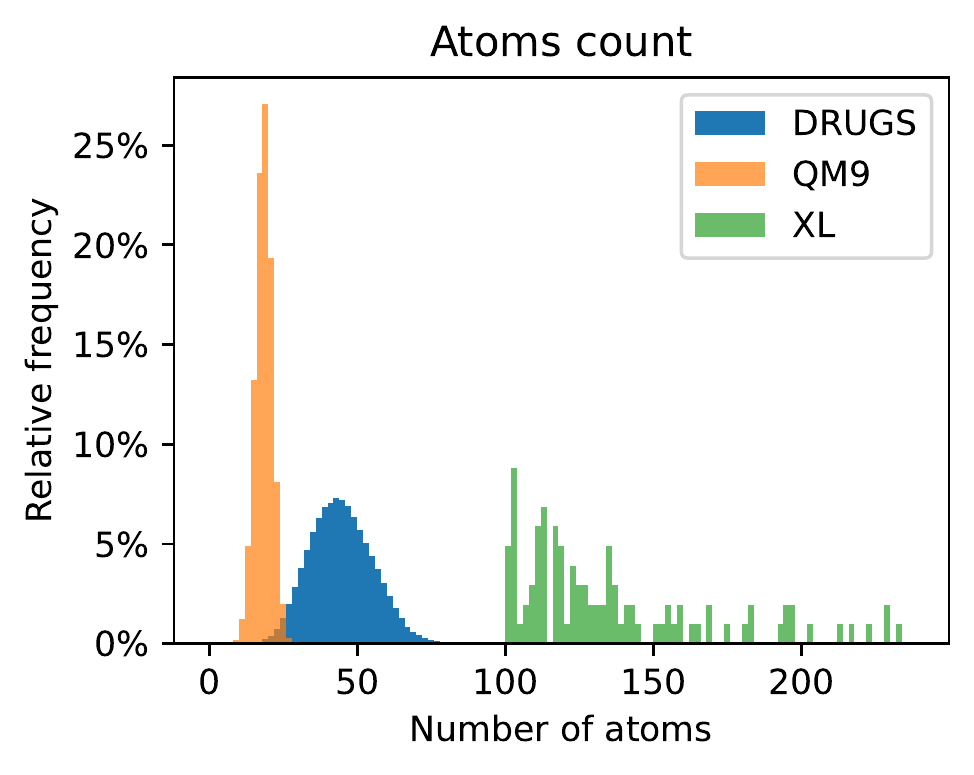}
    \includegraphics[width=0.49\textwidth]{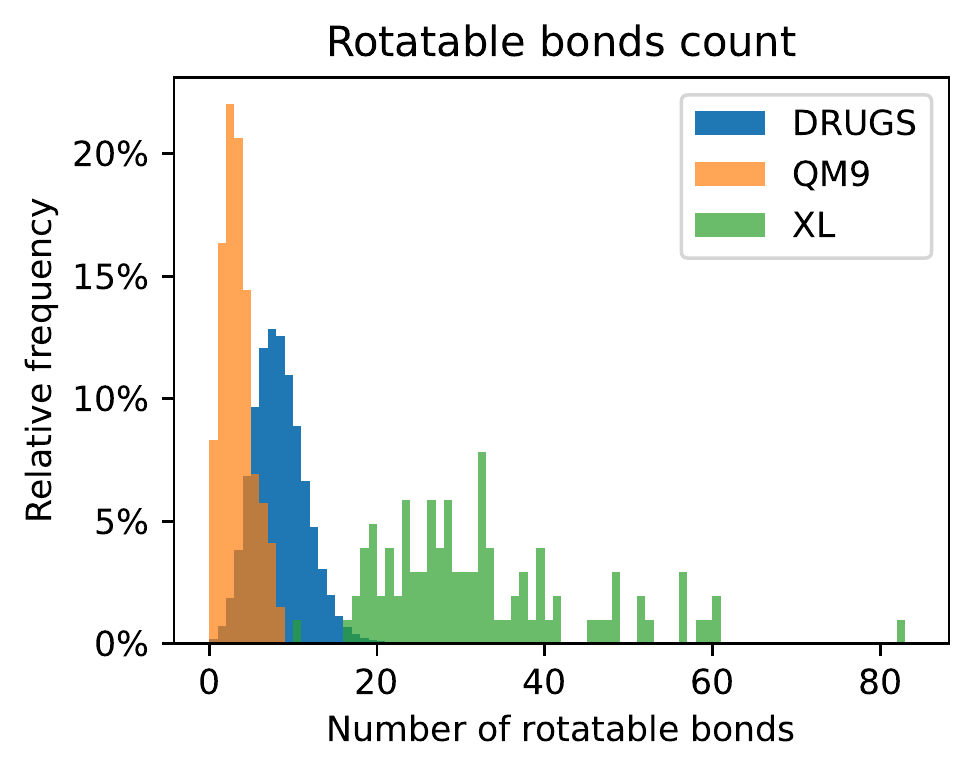}
    \caption{Statistics about the atoms and rotatable bonds counts in the three different datasets.}
    \label{fig:data_stats}
\end{figure}

\paragraph{Boltzmann generator} The torsional Boltzmann generator described in Section \ref{sec:boltzmann} is trained and tested on molecules from GEOM-DRUGS with 3--7 rotatable bonds. The training (validation) set consists of 10000 (400) such randomly selected molecules from the DRUGS training (validation) set. The test set consists of all the 453 molecules present in the DRUGS test set with 3--7 rotatable bonds.

\subsection{Training and tuning details}

\paragraph{Conformer generation} For conformer ensemble generation on GEOM-DRUGS, the torsional diffusion models were trained on NVIDIA RTX A6000 GPUs for 250 epochs with the Adam optimizer (taking from 4 to 11 days on a single GPU). The hyperparameters tuned on the validation set were (in bold the value that was chosen): initial learning rate (0.0003, \textbf{0.001}, 0.003), learning rate scheduler patience (5, \textbf{20}), number of layers (2, \textbf{4}, 6), maximum representation order (1st, \textbf{2nd}), $r_{\max}$ (\textbf{5\AA}, 7\AA, 10\AA) and batch norm (\textbf{True}, False). Finally, for low-temperature sampling, the three relevant parameters $\lambda$, $\psi$, and $\sigma_d$ were chosen with an inference sweep with Bayesian optimizer. All the other default hyperparameters used can be found in the attached code. For GEOM-XL the same trained model was used; for GEOM-QM9 a new model with the same hyperparameters was trained.

\paragraph{Torsional Boltzmann generators} We start from a torsional diffusion model pretrained on GEOM-DRUGS, and train for 250 epochs (6-9 days on a single GPU). A separate model is trained for every temperature. The resampling procedure with 5 steps is run for every molecule every $\max(5, ESS)$ epochs, where $ESS$ is computed for the current set of 32 samples. The only hyperparameter tuned (at temperature 300K) is $\sigma_{\min}$, the noise level at which to stop the reverse diffusion process.


We further improve the training procedure of torsional Boltzmann generators by implementing \textit{annealed training}. The Boltzmann generator for some temperature $T$ is trained at epoch $k$ by using the Boltzmann distribution at temperature $T' = T + (3000 - T)/k$ as the target distribution for that epoch. Intuitively, this trains the model at the start with a smoother distribution that is easier to learn, which gradually transforms into the desired distribution.

\subsection{Evaluation details}

\paragraph{Ensemble RMSD} As evaluation metrics for conformer generation, \cite{ganea2021geomol} and following works have used the so-called Average Minimum RMSD (AMR) and Coverage (COV) for Precision (P) and Recall (R) measured when generating twice as many conformers as provided by CREST. For $K=2L$ let $\{C^*_l\}_{l \in [1, L]}$ and $\{C_k\}_{k \in [1, K]}$ be respectively the sets of ground truth and generated conformers:
\begin{equation}
\begin{aligned}
\text{COV-R} &:= \frac{1}{L}\, \bigg\lvert\{ l \in [1..L]: \exists k \in [1..K], \rmsd(C_k, C^*_l) < \delta \, \bigg\rvert\\
\text{AMR-R} &:= \frac{1}{L} \sum_{l \in [1..L]}  \min_{ k\in [1..K]} \rmsd(C_k, C^*_l)
\end{aligned}
\end{equation}

where $\delta$ is the coverage threshold. The precision metrics are obtained by swapping ground truth and generated conformers.

In the XL dataset, due to the size of the molecules, we compute the RMSDs without testing all possible symmetries of the molecules, therefore the obtained RMSDs are an upper bound, which we find to be very close in practice to the permutation-aware RSMDs.

\paragraph{Runtime evaluation} We benchmark the methods on CPU (Intel i9-9920X) to enable comparison with RDKit. The number of threads for RDKit, \verb|numpy|, and \verb|torch| is set to 8. We select 10 molecules at random from the GEOM-DRUGS test set and generate 8 conformers per molecule using each method. Script loading and model loading times are not included in the reported values.

\paragraph{Boltzmann generator} To evaluate how well the torsional Boltzmann generator and the AIS baselines sample from the conditional Boltzmann distribution, we report their median effective sample size (ESS) \cite{kish1965survey} given the importance sampling weights $w_i$ of 32 samples for each molecule:
\begin{equation}
ESS = \frac{\big( \sum_{i=1}^{32} w_i \big)^2}{\sum_{i=1}^{32} w_i^2}
\end{equation}
This approximates the number of independent samples that would be needed from the target Boltzmann distribution to obtain an estimate with the same variance as the one obtained with the importance-weighted samples.

For the baseline annealed importance samplers, the transition kernel is a single Metropolis-Hastings step with the wrapped normal distributions on $\mathbb{T}^m$ as the proposal. We run with a range of kernel variances: $0.25, 0.5, 0.3, 0.5, 0.75, 1, 1.5. 2$; and report the best result. We use an exponential annealing schedule; i.e., $p_n \propto p_0^{1-n/N}p_N^{n/N}$ where $p_0$ is the uniform distribution and $p_N$ is the target Boltzmann density.

\section{Chapter 4: DiffDock} 

In general, all our code is available at \url{https://github.com/gcorso/DiffDock}. This includes running the baselines, runtime calculations, training and inference scripts for \textsc{DiffDock}, the PDB files of \textsc{DiffDock}'s predictions for all 363 complexes of the test set, and visualization videos of the reverse diffusion.

\subsection{Experimental Setup} \label{appx:experimental_setup}
\paragraph{Data.} We use the molecular complexes in PDBBind \cite{liu2017PDBBind} that were extracted from the Protein Data Bank (PDB) \cite{berman2003PDB}. We employ the time-split of PDBBind proposed by St\"ark et al. \cite{equibind} with 17k complexes from 2018 or earlier for training/validation and 363 test structures from 2019 with no ligand overlap with the training complexes. This is motivated by the further adoption of the same split \cite{Lu2022TankBind} and the critical assessment of PDBBind splits by Volkov et al. \cite{Volkov2022PDBBindSplits} who favor temporal splits over artificial splits based on molecular scaffolds or protein sequence/structure similarity. For completeness, we also report the results on protein sequence similarity splits in Appendix~\ref{app:diffdock_furtherres}. We download the PDBBind data as it is provided by EquiBind from \texttt{https://zenodo.org/record/6408497}. These files were preprocessed with Open Babel before adding any potentially missing hydrogens, correcting hydrogens, and correctly flipping histidines with the \texttt{reduce} library available at \texttt{https://github.com/rlabduke/reduce}.

\paragraph{Metrics.} To evaluate the generated complexes, we compute the heavy-atom RMSD between the predicted and the crystal ligand atoms when the protein structures are aligned. To account for permutation symmetries in the ligand, we use the symmetry-corrected RMSD of sPyRMSD \cite{spyrmsd2020}. For these RMSD values, we report the percentage of predictions that have an RMSD that is less than 2\AA{}.  We choose 2\AA{} since much prior work considers poses with an RMSD less that 2\AA{} as ``good" or successful \cite{Alhossary2015QuickVina2, Hassan2017QVinaW, mcnutt2021gnina}. This is a chemically relevant metric, unlike the mean RMSD as detailed in Section~\ref{sec:generative_modeling} since for further downstream analyses such as determining function changes, a prediction is only useful below a certain RMSD error threshold. Less relevant metrics such as the mean RMSD are provided in Appendix~\ref{app:diffdock_furtherres}.

\subsection{Implementation details} \label{appx:hyperparameters}
\paragraph{Training Details.} We use Adam \cite{kingma2014adam} as optimizer for the diffusion and the confidence model. The diffusion model with which we run inference uses the exponential moving average of the weights during training, and we update the moving average after every optimization step with a decay factor of 0.999. The batch size is 16. We run inference with 20 denoising steps on 500 validation complexes every 5 epochs and use the set of weights with the highest percentage of RMSDs less than 2\AA{} as the final diffusion model. We trained our final score model on four 48GB RTX A6000 GPUs for 850 epochs (around 18 days). The confidence model is trained on a single 48GB GPU. For inference, only a single GPU is required. Scaling up the model size seems to improve performance and future work could explore whether this trend continues further. For the confidence model uses the validation cross-entropy loss is used for early stopping and training only takes 75 epochs. Code to reproduce all results including running the baselines or to perform docking calculations for new complexes is available at \url{https://github.com/gcorso/DiffDock}.

\paragraph{Hyperparameters.} For determining the hyperparameters of \textsc{DiffDock}'s score model, we trained smaller models (3.97 million parameters) that fit into 48GB of GPU RAM before scaling it up to the final model (20.24 million parameters) that was trained on four 48GB GPUs. The smaller models were only trained for 250 or 300 epochs, and we used the fraction of predictions with an RMSD below 2\AA{} on the validation set to choose the hyperparameters. Table~\ref{tab:hyperparameters} shows the main hyperparameters we tested and the final parameters of the large model we use to obtain our results. We only did little tuning for the minimum and maximum noise levels of the three components of the diffusion. For the translation, the maximum standard deviation is 19\AA{}. We also experimented with second-order features for the Tensor Field Network but did not find them to help. The complete set of hyperparameters next to the main ones we describe here can be found in our repository. From the start we have divided the inference schedule into 20 time steps, the effect of using more or fewer steps for inference is discussed in Appendix \ref{app:ablations}. As we found that the large-scale diffusion models overfit the training data on low-levels of noise we stop the diffusion early after 18 steps. At the last diffusion step no noise is added.

The confidence model has 4.77 million parameters and the parameters we tried are in Table~\ref{tab:hyperparameters_confidence_model}. We generate 28 different training poses for the confidence model (for which it predicts whether or not they have an RMSD below 2\AA{}) with a small score model. The score model used to generate the training samples for the confidence model does not need to be the same one that the model will be applied to at inference time.

\begin{table}[htpb]
\caption[SearchSpace]{The hyperparameter options we searched through for \textsc{DiffDock}'s score model. This was done with small models before scaling up to a large model. The parameters shown here that impact model size (bottom half of the table) are those of the large model. The final parameters for the large \textsc{DiffDock} model are marked in \textbf{bold}.}
\label{tab:hyperparameters}
\begin{center}
\begin{small}
\begin{tabular}{lc}
\toprule
Parameter & Search Space  \\    
\midrule
using all atoms for the protein graph & Yes, \textbf{No}\\
using language model embeddings & \textbf{Yes}, No\\
using ligand hydrogens & Yes, \textbf{No}\\
using exponential moving average & \textbf{Yes}, {No}\\
maximum number of neighbors in protein graph & 10, 16, \textbf{24}, 30\\
maximum neighbor distance in protein graph & 5, 10, \textbf{15}, 18, 20, 30\\
distance embedding method & \textbf{sinusoidal}, gaussian \\
dropout & 0, 0.05, \textbf{0.1}, 0.2 \\
learning rates & 0.01, 0.008, 0.003, \textbf{0.001}, 0.0008, 0.0001\\
batch size & 8, \textbf{16}, 24\\
non linearities & \textbf{ReLU} \\
\midrule
convolution layers & 6 \\
number of scalar features &  48 \\
number of vector features &  10 \\
\bottomrule
\end{tabular}
\end{small}
\end{center}
\vskip -0.1in
\end{table}

\begin{table}[htpb]
\caption[SearchSpace]{The hyperparameter options we searched through for \textsc{DiffDock}'s confidence model. The final parameters are marked in \textbf{bold}.}
\label{tab:hyperparameters_confidence_model}
\begin{center}
\begin{small}
\begin{tabular}{lc}
\toprule
Parameter & Search Space  \\    
\midrule
using all atoms for the protein graph & \textbf{Yes}, {No}\\
using language model embeddings & \textbf{Yes}, No\\
using ligand hydrogens & \textbf{No}\\
using exponential moving average & \textbf{No}\\
maximum number of neighbors in protein graph & 10, 16, \textbf{24}, 30\\
maximum neighbor distance in protein graph & 5, 10, \textbf{15}, 18, 20, 30\\
distance embedding method & \textbf{sinusoidal} \\
dropout & 0, 0.05, \textbf{0.1}, 0.2 \\
learning rates & 0.03, 0.003, \textbf{0.0003}, 0.00008\\
batch size & \textbf{16}\\
non linearities & \textbf{ReLU} \\
\midrule
convolution layers & 5 \\
number of scalar features &  24 \\
number of vector features &  6 \\
\bottomrule
\end{tabular}
\end{small}
\end{center}
\vskip -0.1in
\end{table}

\paragraph{Runtime.} Similar to all the baselines, the preprocessing times are not included in the reported runtimes. For \textsc{DiffDock} the preprocessing time is negligible compared to the rest of the inference time where multiple reverse diffusion steps are performed. Preprocessing mainly consists of a forward pass of ESM2 to generate the protein language model embeddings, RDKit's conformer generation, and the conversion of the protein into a radius graph. We measured the inference time when running on an RTX A100 40GB GPU when generating 10 samples. The runtimes we report for generating 40 samples and ranking them are extrapolations where we multiply the runtime for 10 samples by 4. In practice, this only gives an upper bound on the runtime with 40 samples, and the actual runtime should be faster.

\paragraph{Statistical intervals and significance.} In order to provide estimates about the variance of the performance measures reported, in Tables \ref{tab:results_main} and \ref{tab:results_unseen} we report the standard deviation of the performance under 1000 independent resamples with replacement of the test set (bootstrapping). For determining the statistical significance of the superior performance of our method we used the paired two-sample t-test implemented in \texttt{scipy} \cite{virtanen2020scipy}. For Autodock Vina, we took the results from Lu et al. \cite{Lu2022TankBind}, and therefore we were not able to run estimates of the intervals. 

\subsection{Baselines details} \label{appx:baseline_details}
Our scripts to run the baselines are available at \url{https://github.com/gcorso/DiffDock}. For obtaining the runtimes of the different methods, we always used 16 CPUs except for GLIDE as explained below. The runtimes do not include any preprocessing time for any of the methods. For instance, the time that it takes to run P2Rank is not included for TANKBind, and P2Rank + SMINA/GNINA since this receptor preparation only needs to be run once when docking many ligands to the same protein. In applications where different receptors  are processed (such as reverse screening), the experienced runtimes for TANKBind and P2Rank + SMINA/GNINA will thus be higher.

We note that for all these baselines we have used the default hyperparameters unless specified differently below. Modifying some of these hyperparameters (for example the scoring method's exhaustiveness) will change the runtime and performance tradeoffs (e.g., if the searching routine is left running for longer then better poses are likely to be found), however, we leave these analyses to future work.

\paragraph{SMINA} \cite{koes2013smina} improves Autodock Vina with a new scoring-function and user-friendliness. The default parameters were used with the exception of setting \texttt{--num\_modes 10}. To define the search box, we use the automatic box creation option around the receptor with the default buffer of 4\AA{} on all 6 sides.

\paragraph{GNINA} \cite{mcnutt2021gnina} builds on SMINA by additionally using a learned 3D CNN for scoring. The default parameters were used with the exception of setting \texttt{--num\_modes 10}. To define the search box, we use the automatic box creation option around the receptor with the default buffer of 4\AA{} on all 6 sides.

\paragraph{QuickVina-W} \cite{Hassan2017QVinaW} extends the speed-optimized QuickVina 2 \cite{Alhossary2015QuickVina2} for blind docking. We reuse the numbers from St\"ark et al. \cite{equibind} which had used the default parameters except for increasing the exhaustiveness to 64. The files were preprocessed with the \texttt{prepare\_ligand4.py} and \texttt{prepare\_receptor4.py} scripts of the MGLTools library as it is recommended by the QuickVina-W authors.

\paragraph{Autodock Vina} \cite{trott2010autodock} is older docking software that does not perform as well as the other more recent search-based baselines, but it is a well-established tool. We reuse the numbers reported in TANKBind \cite{Lu2022TankBind}

\paragraph{GLIDE} \cite{halgren2004glide} is a strong heavily used commercial docking tool. These methods all use biophysics based scoring-functions. We reuse the numbers from St\"ark et al. \cite{equibind} since we do not have a license. Running GLIDE involves running their command line tools for preprocessing the structures into the files required to run the docking algorithm. As explained by St\"ark et al. \cite{equibind}, the very high runtime of GLIDE with 1405 seconds per complex is partially explained by the fact that GLIDE only uses a single thread when processing a complex. This fact and the parallelization options of GLIDE are explained here \url{https://www.schrodinger.com/kb/1165}. With GLIDE, it is possible to start data-parallel processes that compute the docking results for a different complex in parallel. However, each process also requires a separate software license.

\paragraph{EquiBind} \cite{equibind}, we reuse the numbers reported in their paper and generate the predictions that we visualize with their code at \url{https://github.com/HannesStark/EquiBind}.

\paragraph{TANKBind} \cite{Lu2022TankBind}, we use the code associated with the paper at \url{https://github.com/luwei0917/TankBind}. The runtimes do not include the runtime of P2Rank or any preprocessing steps. In Table~\ref{tab:results_main} we report two runtimes (0.72/2.5 sec). The first is the runtime when making only the top-1 prediction and the second is for producing the top-5 predictions. Producing only the top-1 predictions is faster since TANKBind produces distance predictions that need to be converted to coordinates with a gradient descent algorithm and this step only needs to be run once for the top-1 prediction, while it needs to be run 5 times for producing 5 outputs. To obtain our runtimes we run the forward pass of TANKBind on GPU (0.28 seconds) with the default batch size of 5 that is used in their GitHub repository. To compute the time the distances-to-coordinates conversion step takes, we run the file \texttt{baseline\_run\_tankbind\_parallel.sh} in our repository, which parallelizes the computation across 16 processes which we also run on an Intel Xeon Gold 6230 CPU. This way, we obtain 0.44 seconds runtime for the conversion step of the top-1 prediction (averaged over the 363 complexes of the testset). 

\paragraph{P2Rank} \cite{krivak2018p2rank}, is a tool that predicts multiple binding pockets and ranks them. We use it for running TANKBind and P2Rank + SMINA/GNINA. We download the program from \url{https://github.com/rdk/p2rank} and run it with its default parameters.

\paragraph{EquiBind + SMINA/GNINA} \cite{equibind}, the bounding box in which GNINA/SMINA searches for binding poses is constructed around the prediction of EquiBind with the \texttt{--autobox\_ligand} option of GNINA/SMINA. EquiBind is thus used to find the binding pocket and SMINA/GNINA to find the exact final binding pose. We use \texttt{--autobox\_add 10} to add an additional 10\AA{} on all 6 sides of the bounding box following \cite{equibind}.

\paragraph{P2Rank + SMINA/GNINA.} The bounding box in which GNINA/SMINA searches for binding poses is constructed around the pocket center that P2Rank predicts as the most likely binding pocket. P2Rank is thus used to find the binding pocket and SMINA/GNINA to find the exact final binding pose. The diameter of the search box is the diameter of a ligand conformer generated by RDKit with an additional 10\AA{} on all 6 sides of the bounding box.
\chapter{Further Results}

\section{Chapter 3: Torsional Diffusion} \label{app:torsional_results}

\subsection{Small molecules ensemble RMSD}

We also train and evaluate our model on the small molecules from GEOM-QM9 and report the performance in Table \ref{tab:results_qm9}. For these smaller molecules, cheminformatics methods already do very well and, given the very little flexibility and few rotatable bonds present, the accuracy of local structure significantly impacts the performance of torsional diffusion. RDKit achieves a mean recall AMR just over 0.23\AA, while torsional diffusion based on RDKit local structures results in a mean recall AMR of 0.178\AA. This is already very close lower bound of 0.17\AA \, that can be achieved with RDKit local structures (as approximately calculated by conformer matching). Torsional diffusion does significantly better than other ML methods, but is only on par with or slightly worse than OMEGA, which, evidently, has a better local structures for these small molecules.

\begin{table}[h!]
\caption{Performance of various methods on the GEOM-QM9 dataset test-set ($\delta=0.5$\AA). Again GeoDiff was retrained on the splits from \cite{ganea2021geomol}. } \label{tab:results_qm9}
\begin{tabular}{l|cccc|cccc} \toprule
                & \multicolumn{4}{c|}{Recall} & \multicolumn{4}{c}{Precision}  \\
                  & \multicolumn{2}{c}{Coverage $\uparrow$} & \multicolumn{2}{c|}{AMR $\downarrow$} & \multicolumn{2}{c}{Coverage $\uparrow$} & \multicolumn{2}{c}{AMR $\downarrow$} \\
Method & Mean & Med & Mean & Med & Mean & Med & Mean & Med \\ \midrule
RDKit            & 85.1          & \textbf{100.0} & 0.235          & 0.199          & 86.8          & \textbf{100.0} & 0.232          & 0.205          \\
OMEGA            & 85.5          & \textbf{100.0} & \textbf{0.177} & \textbf{0.126} & 82.9          & \textbf{100.0} & 0.224          & \textbf{0.186} \\
GeoMol           & 91.5          & \textbf{100.0} & 0.225          & 0.193          & 86.7          & \textbf{100.0} & 0.270          & 0.241          \\
GeoDiff          & 76.5          & \textbf{100.0} & 0.297          & 0.229          & 50.0          & 33.5           & 0.524          & 0.510          \\ \midrule
Torsional diffusion             & \textbf{92.8} & \textbf{100.0} & 0.178          & 0.147          & \textbf{92.7} & \textbf{100.0} & \textbf{0.221} & 0.195             \\ \bottomrule
\end{tabular}
\end{table}

\subsection{Ablation experiments} \label{app:ablations}

In Table \ref{tab:ablations} we present a set of ablation studies to evaluate the importance of different components of the proposed torsional diffusion method:
\begin{enumerate}
    \item \textit{Baseline} refers to the model described and tested throughout the paper.
    \item \textit{First order irreps} refers to the same model but with node irreducible representations kept only until order $\ell = 1$ instead of $\ell = 2$; this worsens the average error by about 5\%, but results in a 41\% runtime speed-up. 
    \item \textit{Only D.E. matching} refers to a model trained on conformers obtained by a random assignment of RDKit local structures to ground truth conformers (without first doing an optimal assignment as in Appendix \ref{app:matching}); this performs only marginally worse than full conformer matching.
    \item \textit{Train on ground truth L} refers to a model trained directly on the ground truth conformers without {conformer matching} but tested (as always) on RDKit local structures; although the training and validation score matching loss of this model is significantly lower, its inference performance reflects the detrimental effect of the local structure distributional shift.
    \item \textit{No parity equivariance} refers to a model whose outputs are parity invariant instead of parity equivariant; the model cannot distinguish a molecule from its mirror image and fails to learn, resulting in performance on par with a random baseline.
    \item \textit{Random $\boldsymbol{\tau}$} refers to a random baseline using RDKit local structures and uniformly random torsion angles.
\end{enumerate}

\begin{table}[h!]
\caption{Ablation studies tested on conformer generation on GEOM-DRUGS. Refer to Appendix \ref{app:ablations} for an explanation of each entry. } \label{tab:ablations}
\vspace{5pt}
\begin{tabular}{l|cccc|cccc} \toprule
                & \multicolumn{4}{c|}{Recall} & \multicolumn{4}{c}{Precision}  \\
                  & \multicolumn{2}{c}{Coverage $\uparrow$} & \multicolumn{2}{c|}{AMR $\downarrow$} & \multicolumn{2}{c}{Coverage $\uparrow$} & \multicolumn{2}{c}{AMR $\downarrow$} \\
Method & Mean & Med & Mean & Med & Mean & Med & Mean & Med \\ \midrule
Baseline & \textbf{72.7} & {80.0} & \textbf{0.582} & \textbf{0.565} & \textbf{55.2} & \textbf{56.9} & \textbf{0.778} & \textbf{0.729}    \\ \midrule
First order irreps & 70.1 & 77.9 & 0.605 & 0.589 & 51.4 & 51.4 & 0.817 & 0.783\\
Only D.E. matching & 72.5 & \textbf{81.1} & 0.588 & 0.569 & 53.8 & 56.1 & 0.794 & 0.749 \\
Train on ground truth $L$ & 34.8 & 22.4 & 0.920 & 0.909 & 22.3 & 7.8 & 1.182 & 1.136\\
No parity equivariance & 30.5 & 12.5 & 0.928 & 0.929 & 17.9 & 3.9 & 1.234 & 1.217\\
\midrule
Random $\boldsymbol{\tau}$ & 30.9 & 13.2 & 0.922 & 0.923 & 18.2 & 4.0 & 1.228 & 1.217\\
\bottomrule
\end{tabular}
\end{table}

\subsection{Ensemble properties}

In Table \ref{tab:properties_appendix}, we report the median absolute errors of the Boltzmann-weighted properties of the generated vs CREST ensembles, with and without GFN2-xTB relaxation. For all methods, the errors without relaxation are far too large for the computed properties to be chemically useful---for reference, the thermal energy at room temperature is 0.59 kcal/mol. In realistic settings, relaxation of local structures is necessary for any method, after which errors from global flexibility become important. After relaxation, torsional diffusion obtains property approximations on par or better than all competing methods. 

\begin{table}[h!]
    \caption{Median absolute error of generated v.s. ground truth ensemble properties with and without relaxation. $E, \Delta\epsilon, E_{\min}$ in kcal/mol, $\mu$ in debye.}\label{tab:properties_appendix}
    \vspace{5pt}
    \centering
    \begin{tabular}{l|cccc|cccc}
    \toprule 
    & \multicolumn{4}{c|}{Without relaxation} &  \multicolumn{4}{c}{With relaxation} \\
    Method & $E$ & $\mu$ & $\Delta \epsilon$ & $E_{\min}$ & $E$ & $\mu$ & $\Delta \epsilon$ & $E_{\min}$ \\\midrule
    RDKit & 39.08 & 1.40 & 5.04 & 39.14 & 0.81 & 0.52 & 0.75 & 1.16 \\
    OMEGA & \textbf{16.47} & \textbf{0.78} & \textbf{3.25} & \textbf{16.45} & 0.68 & 0.66 & 0.68 & 0.69 \\
    GeoMol & 43.27 & 1.22 & 7.36 & 43.68 & 0.42 & \textbf{0.34} & 0.59 & 0.40 \\
    GeoDiff & 18.82 & 1.34 & 4.96 & 19.43 & 0.31 & {0.35} & 0.89 & 0.39 \\
    Tor. Diff. & 36.91 & 0.92 & 4.93 & 36.94 & \textbf{0.22} & {0.35} & \textbf{0.54} & \textbf{0.13} \\ \bottomrule
    \end{tabular}
\end{table}

\section{Chapter 4: DiffDock} 

\subsection{Physically plausible predictions}\label{appx:steric_clashes}

\begin{table}[htb]
    \caption{\textbf{Steric clashes.} Percentage of test complexes for which the predictions of the different methods exhibit steric clashes. Search-based methods never produced steric clashes.}
    \label{tab:steric_clashes}
     \begin{small}
     \begin{center}

    \begin{tabular}{lcc}
    \toprule
      & Top-1  & Top-5 \\
    
        Method & \% steric clashes  & \% steric clashes\\
    \midrule
    \textsc{EquiBind}           & 26  & -   \\
    \textsc{TANKBind}           & 6.6 & 3.6   \\ \midrule
    \textbf{\textsc{DiffDock} (10)}  & \textbf{2.8} & \textbf{0}  \\
    \textbf{\textsc{DiffDock} (40)}  & \textbf{2.2} & \textbf{2.2}  \\
    \bottomrule
    \end{tabular}
    \end{center}
    \end{small}
\end{table}

Due to the averaging phenomenon of regression-based methods such as TANKBind and EquiBind, they make predictions at the mean of the distribution. If aleatoric uncertainty is present, such as in case of symmetric complexes, this leads to predicting the ligand to be at an un-physical state in the middle of the possible binding pockets as visualized in Figure~\ref{fig:symmetric_complexes}. The Figure also illustrates how \textsc{DiffDock} does not suffer from this issue and is able to accurately sample from the modes.

In the scenario when epistemic uncertainty about the correct ligand conformation is present, this often results in ``squashed-up" predictions of the regression-based methods as visualized in Figure~\ref{fig:self_intersections}. If there is uncertainty about the correct conformer, the square error minimizing option is to put all atoms close to the mean.

\begin{figure}[t]
    \centering
    \includegraphics[width=\textwidth]{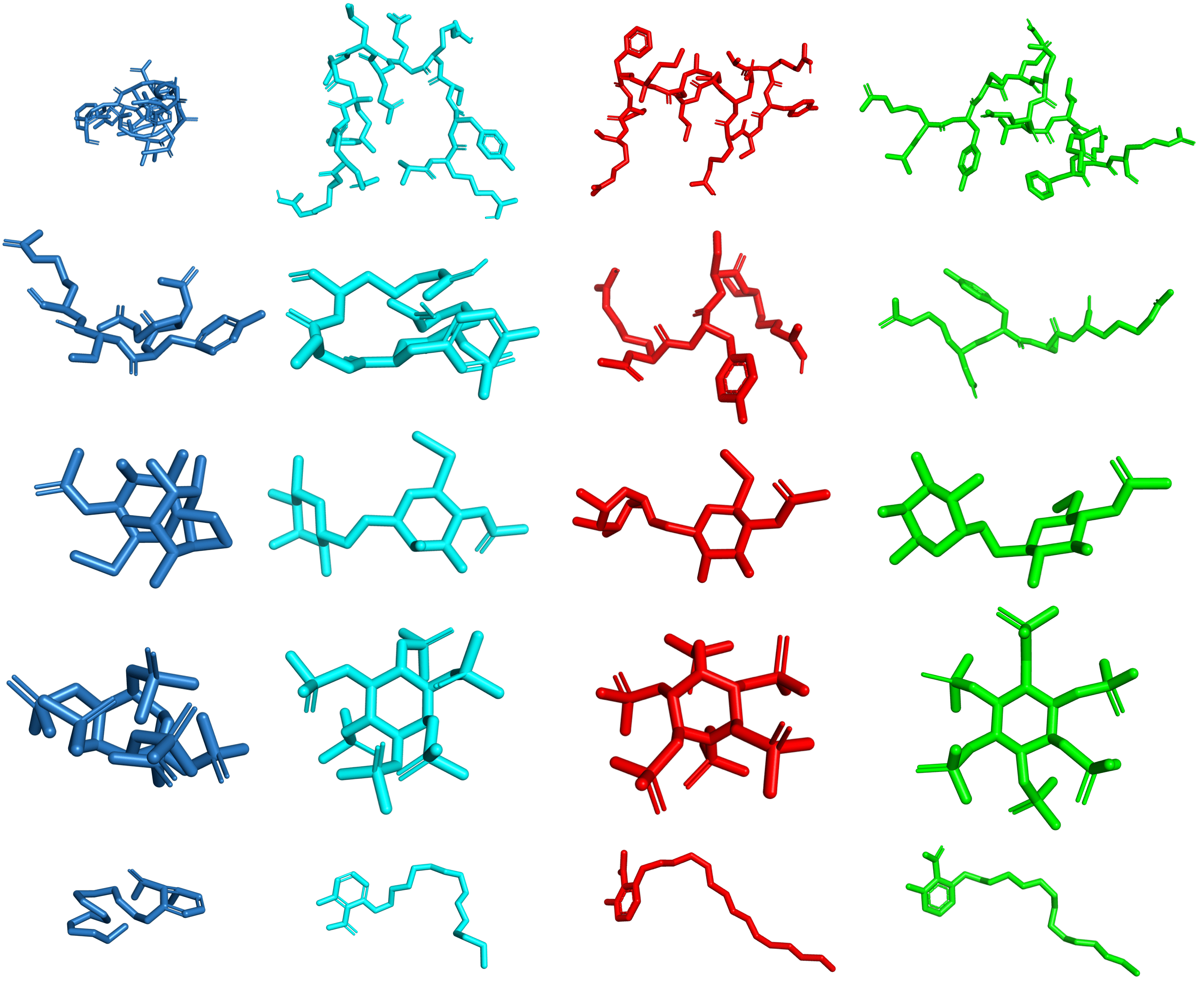}
    \caption{\textbf{Ligand self-intersections.} TANKBind (blue), EquiBind (cyan), \textsc{DiffDock} (red), and crystal structure (green). Due to the averaging phenomenon that occurs when epistemic uncertainty is present, the regression-based deep learning models tend to produce ligands with atoms that are close together, leading to self-intersections. \textsc{DiffDock}, as a generative model, does not suffer from this averaging phenomenon, and we never found a self-intersection in any of the investigated results of \textsc{DiffDock}.}
    \label{fig:self_intersections}
\end{figure}

These averaging phenomena in the presence of either aleatoric or epistemic uncertainty cause the regression-based methods to often generate steric clashes and self intersections. To investigate this quantitatively, we determine the fraction of test complexes for which the methods exhibit steric clashes. We define a ligand as exhibiting a steric clash if one of its heavy atoms is within 0.4\AA{} of a heavy receptor atom. This cutoff is used by protein quality assessment tools and in previous literature \cite{Ramachandran2011stericClashes}. Table~\ref{tab:steric_clashes} shows that \textsc{DiffDock}, as a generative model, produces fewer steric clashes than the regression-based baselines. We generally observe no unphysical predictions from \textsc{DiffDock} unlike the self intersections that, e.g., TANKBind produces (Figure~\ref{fig:self_intersections}) or its incorrect local structures (Figure~\ref{fig:plausible_local_structures}). This is also visible in the randomly chosen examples of Figure~\ref{fig:random_examples} and can be examined in our repository, where we provide all predictions of \textsc{DiffDock} for the test set.

\begin{figure}[t]
    \centering
    \includegraphics[width=\textwidth]{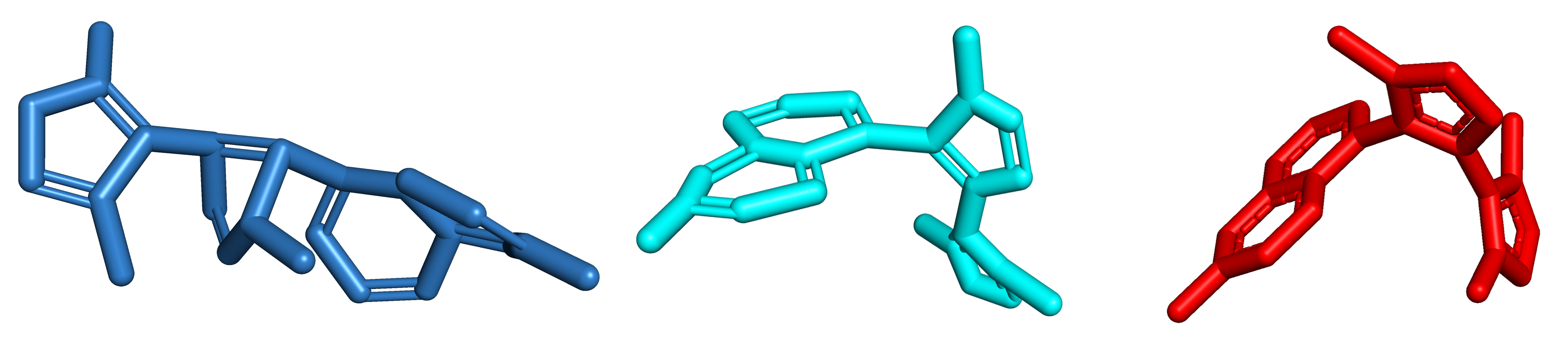}
    \caption{\textbf{Chemically plausible local structures.} TANKBind (blue), EquiBind (cyan), and \textsc{DiffDock} (red) structures for complex 6g2f. EquiBind (without their correction step) produces very unrealistic local structures and TANKBind, e.g., produces non-planar aromatic rings. \textsc{DiffDock}'s local structures are the realistic local structures of RDKit. }
    \label{fig:plausible_local_structures}
\end{figure}

\subsection{Further Results and Metrics} \label{app:diffdock_furtherres}

In this section, we present further evaluation metrics on the results presented in Table~\ref{tab:results_main}. In particular, for both top-1 (Table~\ref{tab:top1complete}) and top-5 (Table~\ref{tab:top5complete}) we report: 25th, 50th and 75th percentiles, the proportion below 2\AA{} and below 5\AA{} of both ligand RMSD and centroid distance. Moreover, while Volkov et al. \cite{Volkov2022PDBBindSplits} advocated against artificial protein set splits and for time-based splits, for completeness, in Table~\ref{tab:results_unseen} and Figure~\ref{fig:histogram_unseen_rec}, we report the performances of the different methods when evaluated exclusively on the portion of the test set where the UniProt IDs of the proteins are not contained in the data that is seen by \textsc{DiffDock} in its training and validation.

\begin{table*}[!h]
    \caption{\textbf{Top-1 PDBBind docking.}  }
    \label{tab:top1complete}
     \begin{small}
     \begin{center}
     \makebox[\textwidth][c]{
    \begin{tabular}{=l+c+c+c+c+c|+c+c+c+c+c}
    
    \toprule
     &\multicolumn{5}{c}{Ligand RMSD} & \multicolumn{5}{c}{Centroid Distance} \\
     &\multicolumn{3}{c}{Percentiles $\downarrow$} & \multicolumn{2}{c}{\begin{tabular}{@{}c@{}}\% below\\threshold $\uparrow$\end{tabular} }  &\multicolumn{3}{c}{Percentiles $\downarrow$} & \multicolumn{2}{c}{\begin{tabular}{@{}c@{}}\% below\\thresh. $\uparrow$\end{tabular} } \\
    
    \textbf{Methods} & 25th & 50th & 75th & 5 \AA{}  &  2 \AA{} & 25th & 50th & 75th & 5 \AA{}  &  2 \AA{} \\
    \midrule
    \textsc{Autodock Vina} & 5.7  &  10.7 &  21.4 & 21.2 & 5.5 & 1.9 &  6.2 & 20.1 & 47.1 & 26.5 \\
     \textsc{QVina-W} & 2.5  &  7.7 &  23.7 & 40.2 & 20.9 & 0.9 &  3.7 & 22.9 & 54.6 & 41.0 \\
     \textsc{GNINA} & 2.4 & 7.7   & 17.9 & 40.8  & 22.9 & 0.8 &  3.7 & 23.1 & 53.6 & 40.2 \\
    \textsc{SMINA} & 3.1 & 7.1 & 17.9 & 38.0 & 18.7 & 1.0 &  2.6 & 16.1 & 59.8 & 41.6 \\
    \textsc{GLIDE} (c.) & 2.6 & 9.3   & 28.1 & 33.6 & 21.8 & 0.8 &  5.6 & 26.9 & 48.7 & 36.1 \\
    \textsc{EquiBind} & 3.8 & 6.2 &  10.3 & 39.1 & 5.5 &  1.3 & 2.6 & 7.4 & 67.5&  40.0 \\ \midrule
    \textsc{TANKBind} & 2.5 & 4.0 & 8.5 & 59.0 & 20.4 & 0.9 & 1.8 & 4.4 & 77.1 &  55.1  \\ 
    \textsc{P2Rank+SMINA} & 2.9 & 6.9 & 16.0 & 43.0 & 20.4 & 0.8 & 2.6 & 14.8 & 60.1 & 44.1   \\ 
    \textsc{P2Rank+GNINA} & 1.7 & 5.5 & 15.9 & 47.8 & 28.8 & 0.6 & 2.2 & 14.6 & 60.9 & 48.3   \\ 
    \textsc{EquiBind+SMINA} & 2.4 & 6.5 & 11.2 & 43.6 & 23.2 & 0.7 & 2.1 & 7.3 & 69.3 & 49.2  \\ 
    \textsc{EquiBind+GNINA} & 1.8 & 4.9 & 13 & 50.3 & 28.8 & 0.6 & 1.9 & 9.9   & 66.5 & 50.8  \\ \midrule 
    \textbf{\textsc{DiffDock} (10)} & 1.5 & 3.6 & \textbf{7.1} & 61.7 & 35.0 & \textbf{0.5} & \textbf{1.2} & 3.3 & \textbf{80.7} & 63.1   \\ 
    \textbf{\textsc{DiffDock} (40)} & \textbf{1.4} & \textbf{3.3} & 7.3 & \textbf{63.2} & \textbf{38.2} & \textbf{0.5} & \textbf{1.2} & \textbf{3.2} & 80.5 & \textbf{64.5}   \\ 
     \bottomrule
    \end{tabular}}
    \end{center}
    \end{small}
\end{table*}

\begin{table*}[!h]
    \caption{\textbf{Top-5 PDBBind docking.}  }
    \label{tab:top5complete}
     \begin{small}
     \begin{center}
     \makebox[\textwidth][c]{
    \begin{tabular}{=l+c+c+c+c+c|+c+c+c+c+c}
    
    \toprule
     &\multicolumn{5}{c}{Ligand RMSD} & \multicolumn{5}{c}{Centroid Distance} \\
     &\multicolumn{3}{c}{Percentiles $\downarrow$} & \multicolumn{2}{c}{\begin{tabular}{@{}c@{}}\% below\\threshold $\uparrow$\end{tabular} }  &\multicolumn{3}{c}{Percentiles $\downarrow$} & \multicolumn{2}{c}{\begin{tabular}{@{}c@{}}\% below\\thresh. $\uparrow$\end{tabular} } \\
    
    \textbf{Methods} & 25th & 50th & 75th & 5 \AA{}  &  2 \AA{} & 25th & 50th & 75th & 5 \AA{}  &  2 \AA{} \\
    \midrule
    \textsc{GNINA} & 1.6 & 4.5 & 11.8 & 52.8 & 29.3 & 0.6 & 2.0 & 8.2 & 66.8 & 49.7   \\
    \textsc{SMINA} & 1.7 & 4.6 & 9.7 & 53.1 & 29.3 & 0.6 & 1.85 & 6.2 & 72.9 & 50.8  \\ \midrule
    \textsc{TANKBind} & 2.1 & 3.4 & 6.1 & 67.5 & 24.5 & 0.8 & 1.4 & 2.9 & 86.8 & 62.0  \\ 
    \textsc{P2Rank+SMINA} & 1.5 & 4.4 & 14.1 & 54.8 & 33.2 & 0.6 & 1.8 & 12.3 & 66.2 & 53.4   \\ 
    \textsc{P2Rank+GNINA} & 1.4 & 3.4 & 12.5 & 60.3 & 38.3 & 0.5 & 1.4 & 9.2 & 69.3 & 57.3  \\ 
    \textsc{EquiBind+SMINA} & 1.3 & 3.4 & 8.1 & 60.6 & 38.6 & 0.5 & 1.3 & 5.1 & 74.9 & 58.9   \\ 
    \textsc{EquiBind+GNINA} & 1.4 & 3.1 & 9.1 & 61.7 & 39.1 & 0.5 & 1.1 & 5.3 & 73.7 & 60.1  \\ \midrule 
    \textbf{\textsc{DiffDock} (10)} & \textbf{1.2} & 2.7 & \textbf{4.9} & 75.1 & 40.7 & 0.5 & 1.0 & 2.2 & 87.0 & 72.3   \\ 
    \textbf{\textsc{DiffDock} (40)} & \textbf{1.2} & \textbf{2.4} & 5.0 & \textbf{75.5} & \textbf{44.7} & \textbf{0.4} & \textbf{0.9} & \textbf{1.9} & \textbf{88.0} & \textbf{76.7}    \\ 
     \bottomrule
    \end{tabular}}
    \end{center}
    \end{small}
\end{table*}

\begin{figure}[!h]
\begin{center}
\includegraphics[width=.48\textwidth]{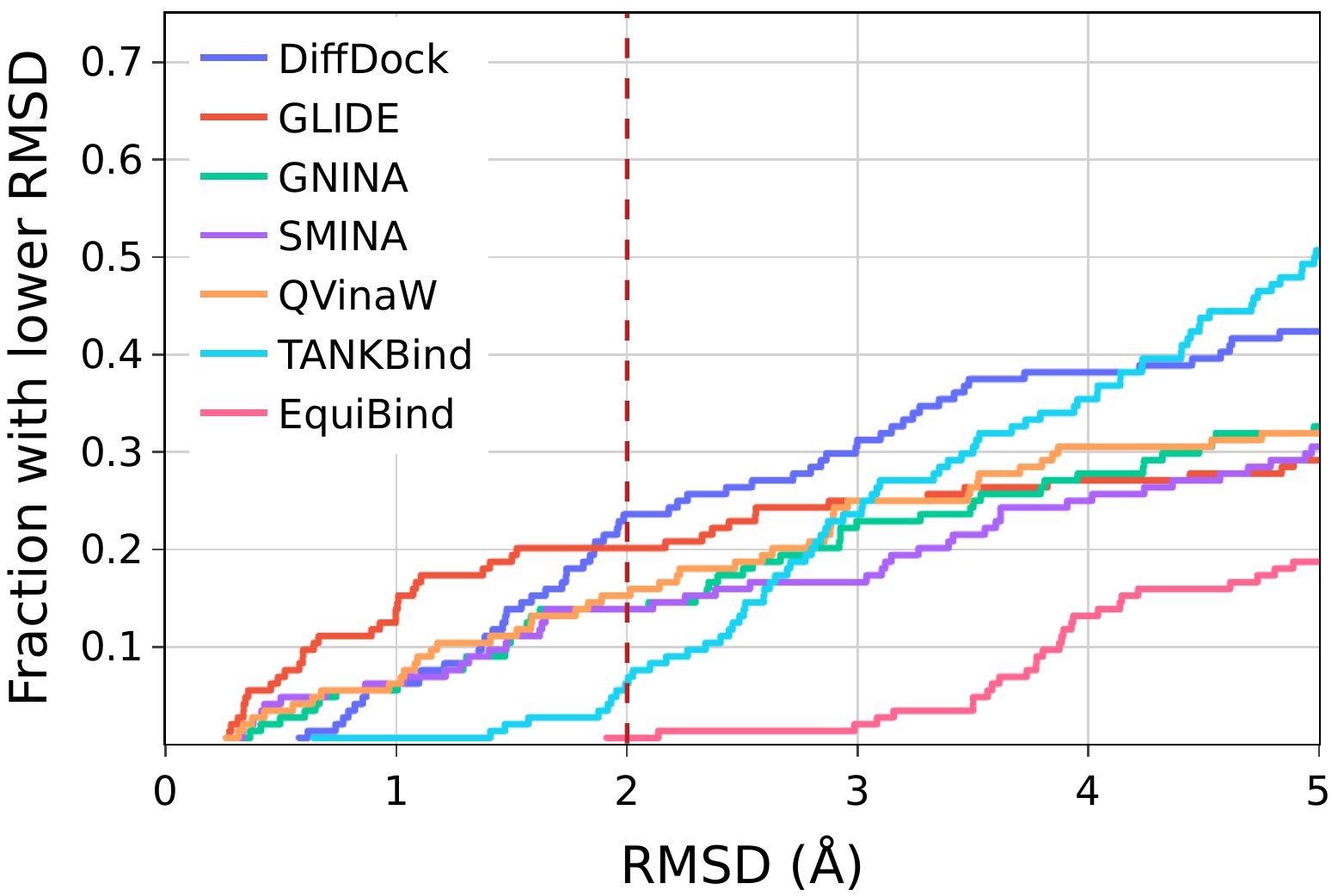}
\includegraphics[width=.48\textwidth]{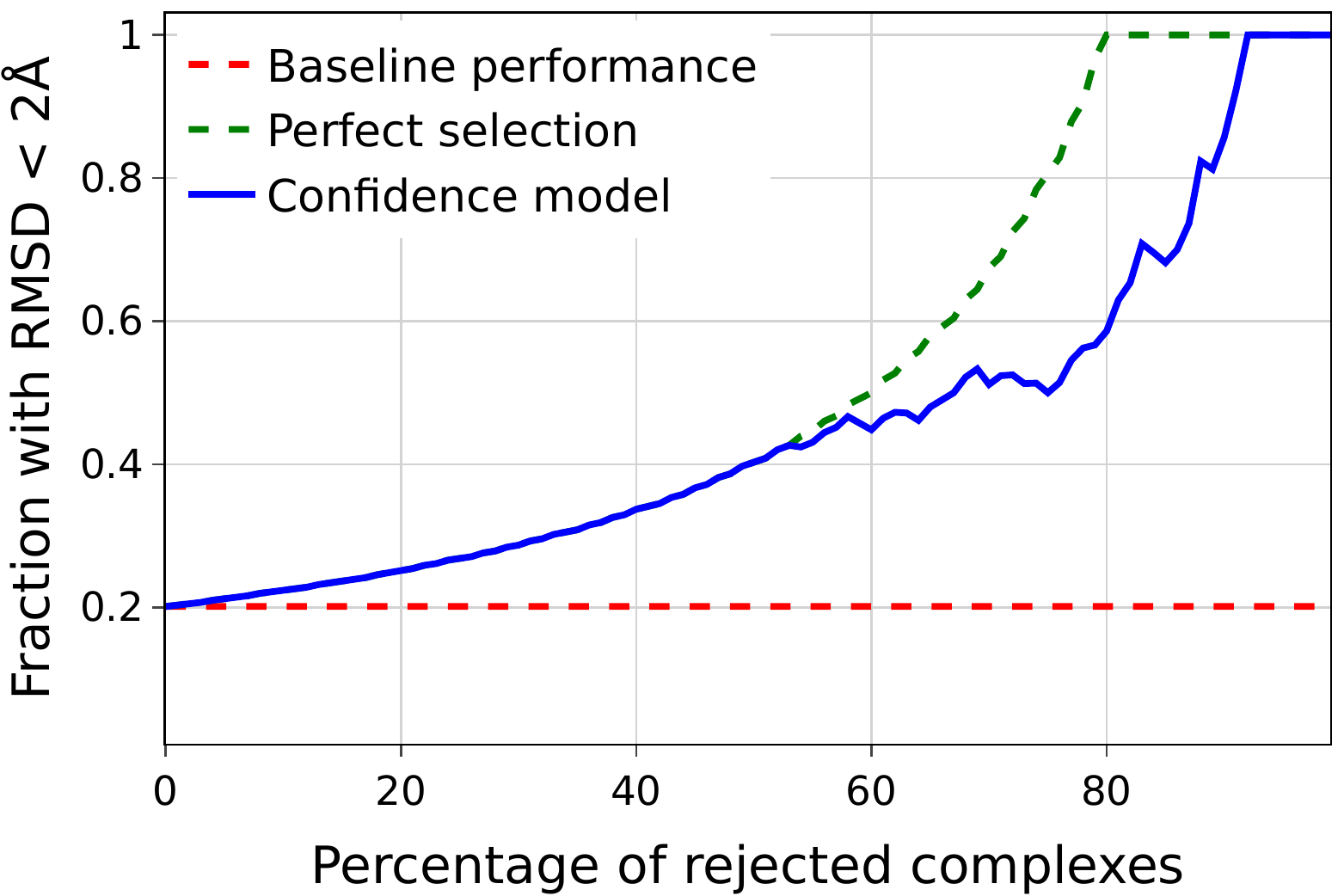}
\caption{\textbf{PDBBind docking on unseen receptors.} \textbf{Left:} cumulative density histogram of the methods' RMSD. \textbf{Right:} Percentage of predictions with RMSD below 2\AA{} when only making predictions for the portion of the dataset where \textsc{DiffDock} is most confident. }
\label{fig:histogram_unseen_rec}
\end{center}
 \vskip -0.1cm
\end{figure}

\begin{table}[!h]
    \caption{\textbf{PDBBind docking on unseen receptors.} Percentage of predictions for which the RMSD to the crystal structure is below 2\AA{} and the median RMSD. ``*" indicates the method run exclusively on CPU, ``-" means not applicable; some cells are empty due to infrastructure constraints. }
    \label{tab:results_unseen}
     \begin{small}
     \begin{center}

    \begin{tabular}{l=c+c|+c+c|c}
    \toprule
      & \multicolumn{2}{c}{Top-1 RMSD} & \multicolumn{2}{c}{Top-5 RMSD}  &  Average\\
    
        Method & \,\%$<$2\, & \,Med.\, & \,\%$<$2\, & \,Med.\, & \,Runtime (s)\, \\
    \midrule
     \textsc{Autodock Vina}             &  1.4 & 16.6 &      &      &   205*     \\
    \textsc{QVinaW}             &  15.3\tiny{$\pm$2.9} & 10.3\tiny{$\pm$2.3} &      &      &  49*     \\
    \textsc{GNINA}              &  14.0\tiny{$\pm$2.9} & 13.6\tiny{$\pm$2.6} & 23.0\tiny{$\pm$3.5} & 7.0\tiny{$\pm$1.1}  &  127    \\
    \textsc{SMINA}              &  14.0\tiny{$\pm$2.9} & 8.5\tiny{$\pm$2.6}  & 21.7\tiny{$\pm$3.5} & 6.7\tiny{$\pm$0.7}  &  126*    \\
    \textsc{GLIDE}              &  19.6\tiny{$\pm$3.3} & 18.0\tiny{$\pm$3.9} &      &      &  1405*   \\
    \textsc{EquiBind}           &  0.7\tiny{$\pm$0.7}  & 9.1\tiny{$\pm$0.6}  &  -   &  -   &  \textbf{0.04}   \\ 
    \textsc{TANKBind}           &  6.3\tiny{$\pm$2.0}  & \textbf{5.0\tiny{$\pm$0.2}}  & 11.1\tiny{$\pm$2.6} & 4.4\tiny{$\pm$0.3}  &  0.7/2.5  \\ \midrule
    \textbf{\textsc{DiffDock} (10)}  &  15.7\tiny{$\pm$3.1} & 6.1\tiny{$\pm$0.7} & 21.8\tiny{$\pm$3.3} & 4.2\tiny{$\pm$0.4} & 10  \\   
    \textbf{\textsc{DiffDock} (40)}  &  \textbf{20.8\tiny{$\pm$3.3}} & 6.2\tiny{$\pm$0.8}  & \textbf{28.7\tiny{$\pm$3.6}} & \textbf{3.9\tiny{$\pm$0.4}}  & 40   \\
    \bottomrule
    \end{tabular}
    \end{center}
    \end{small}
\end{table}

\subsection{Ablation studies} \label{app:diffdock_ablations}

Below we report the performance of our method over different hyperparameter settings. In particular, we highlight the different ways in which it is possible to control the tradeoff between runtime and accuracy in our method. These mainly are: (1) model size, (2) diffusion time, and (3) diffusion samples.

\paragraph{Model size.}  The final \textsc{DiffDock} score model has 20.24 million parameters from its 6 convolution layers with 48 scalar and 10 vector features. In Table~\ref{tab:model_size} we show the results for a smaller score model with 5 convolutions, 24 scalar, and 6 vector features resulting in 3.97 million parameters that can be trained on a single 48GB GPU. The confidence model used is the same for both score models. We find that scaling up the model size helped improve performance which we did as far as possible using four 48GB GPUs for training. Scaling the model size further is a promising avenue for future work.

\paragraph{Protein embeddings.}  As described in Appendix \ref{app:architecture}, the architecture uses as initial features of protein residues the language model embeddings from ESM2 \cite{Lin2022ESM2} in order for the model to more easily reason about the protein sequence. In Table~\ref{tab:model_size} we show that while these provide some improvements they are not necessary to obtain state-of-the-art performance.

\begin{table}[htb]
    \caption{\textbf{Model size and protein embeddings comparison.} All methods receive a small molecule and are tasked to find its binding location, orientation, and conformation. Shown is the percentage of predictions for which the RMSD to the crystal structure is below 2\AA{} and the median RMSD.}
    \label{tab:model_size}
     \begin{small}
     \begin{center}

    \begin{tabular}{l=c+c|+c+c|+c}
    \toprule
      & \multicolumn{2}{c}{Top-1 RMSD (\AA{})} & \multicolumn{2}{c}{Top-5 RMSD (\AA{})} & Average\\ \rule{0pt}{1.5ex}  
    
        Method & \,\%$<$2\, & \,Med.\, & \,\%$<$2\, & \,Med.\, & Runtime (s)  \\
    \midrule
     \textbf{\textsc{DiffDock-small-noESM} (10)}  &  26.2 & 4.7  & 32.0 & 3.2  &  7      \\
     \textbf{\textsc{DiffDock-small-noESM} (40)}  &   28.4 & 3.8  & 37.7 & 2.6  & 28      \\ \midrule
    \textbf{\textsc{DiffDock-small} (10)}  &  26.0 & 4.3  & 33.3 & 3.2  &  7      \\
    \textbf{\textsc{DiffDock-small} (40)}  &  31.1 & 4.0  & 38.0 & 2.7  & 28      \\
    \midrule
     \textbf{\textsc{DiffDock-noESM} (10)}  &  33.9 & 3.8  & 39.4 & 2.8  & 10      \\
     \textbf{\textsc{DiffDock-noESM} (40)}  &  34.2 & 3.5  & 42.7 & 2.4  & 40      \\ \midrule
    \textbf{\textsc{DiffDock} (10)}  &  35.0 & 3.6  & 40.7 & 2.7  &  10     \\
    \textbf{\textsc{DiffDock} (40)}  &  \textbf{38.2} & \textbf{3.3}  & \textbf{44.7} & \textbf{2.4}  & 40      \\
    \bottomrule
    \end{tabular}
    \end{center}
    \end{small}
\vspace{-0.2cm}
\end{table}

\paragraph{Diffusion steps. } Another hyperparameter determining the runtime of the method during inference is the number of steps we take during the reverse diffusion. Since these are applied sequentially \textsc{DiffDock}'s runtime scales approximately linearly with the number of diffusion steps. In the rest of the paper, we always use 20 steps, but in Figure~\ref{fig:diffusion_steps} we show how the performance of the model varies with the number of steps. We note that the model reaches nearly the full performance even with just 10 steps, suggesting that the model can be sped up 2x with a small drop in accuracy.

\begin{figure}[thb]
    \centering
    \includegraphics[width=0.65\textwidth]{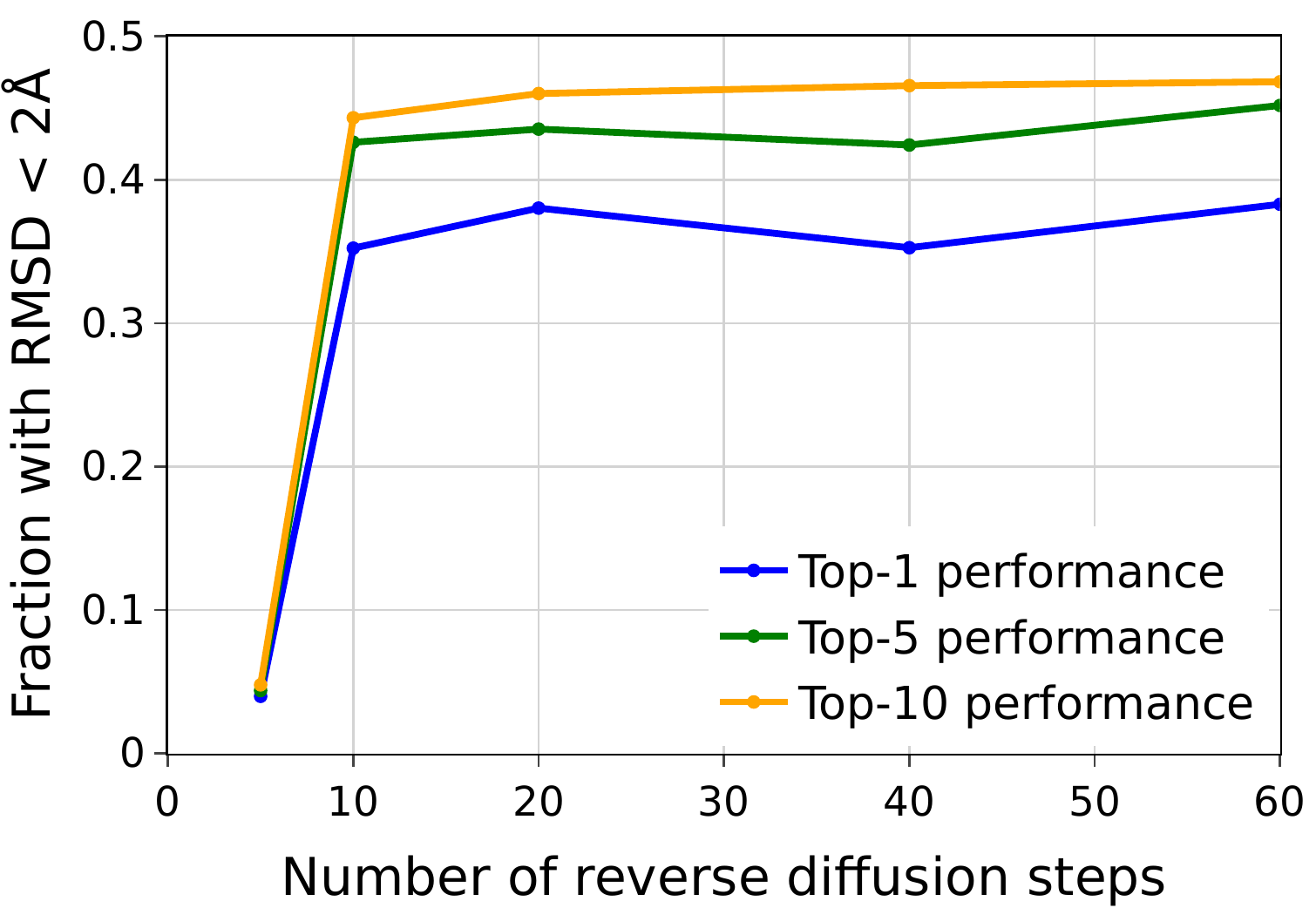}
    \caption{ Ablation study on the number of reverse diffusion steps. }
    \label{fig:diffusion_steps}
\end{figure}

\paragraph{Diffusion samples. } Given a score-based model and a number of steps for the diffusion model, it remains to be determined how many independent samples $N$ to query from the diffusion model and then feed to the confidence model. As expected the more samples the confidence model receives the more likely it is that it will find a pose that it is confident about and, therefore, the higher the performance. The runtime of \textsc{DiffDock} on GPU scales sublinearly until the different samples fit in parallel in the model (depends on the protein size and the GPU memory) and approximately linearly for larger sample sizes (however it can be easily parallelized across different GPUs). In Figure~\ref{fig:results_main} we show how the success rate for the top-1, top-5, and top-10 prediction change as a function of $N$. For example, for the top-1 prediction, the proportion of the prediction with RMSD below 2\AA{} varies between 22\% of a random sample of the diffusion model ($N=1$) to 38\% when the confidence model is allowed to choose between 40 samples.

\subsection{Affinity prediction}

To validate the quality of the predicted poses, we also do some experiments in predicting the binding affinity labels already present in PDBBind. In this section we report some preliminary results on this task that show that a simple approach can already achieve results competitive with the state-of-the-art. We leave a more thorough and sophisticated analysis on how to best use the \textsc{DiffDock} framework for binding affinity to future work. 

\paragraph{Affinity prediction framework.} We train the binding affinity predictor by generating a fixed number of poses with the diffusion model and then feeding them to an affinity prediction model with architecture almost analogous to the confidence model. This affinity prediction model takes in the poses as a single heterogeneous graph with a single receptor but multiple sets of ligand nodes, which have edges to the same receptor but not among themselves. After the final interaction layer, the scalar representations of nodes in each ligand are aggregated with a mean pooling and passed through a set of dense layers (as it is done for the confidence prediction). Then, the representations of the different ligands are aggregated using multiple permutation invariant aggregators (mean, maximum, minimum, and standard deviation) as in Corso et al. \cite{corso2020principal}, and transformed with another set of dense layers producing a single output, the predicted affinity.

\paragraph{Dataset, baselines, and training.} To train we use PDBBind with the same splits used to train the diffusion and confidence models. This provides for each of the complexes an affinity measure that consists of inhibiting concentration ($IC50$), inhibition constant ($K_i$), or dissociation constant ($K_d$) and its conversion to the $-\log K_d/K_i$ metric. As baselines, we use a series of state-of-the-art sequence-based and structure-based methods: TransformerCPI \cite{chen2020transformercpi}, MONN \cite{li2020monn}, IGN \cite{jiang2021interactiongraphnet}, PIGNet \cite{moon2022pignet}, HOLOPTOT \cite{somnath2021multi}, STAMPDPI \cite{wang2022structure} and TANKBind \cite{Lu2022TankBind}. We take the baselines' performances from Lu et al. \cite{Lu2022TankBind}. 

\begin{table}[!htb]
    \caption{\textbf{Binding affinity prediction. } Prediction of $-\log K_d/K_i$ on PDBBind. The baseline numbers are from Lu et al. \cite{Lu2022TankBind}. No hyperparameter tuning was performed for \textsc{DiffDock}'s performance.  }
    \label{tab:results_affinity}
     \begin{small}
     \begin{center}

    \begin{tabular}{lcccc}
    \toprule
    Methods  & RMSE $\downarrow$ & Pearson $\uparrow$ & Spearman $\uparrow$ & MAE $\downarrow$  \\ \midrule
    \textsc{TransCPI} & 1.741 & 0.576    & 0.540      & 1.404 \\
    \textsc{MONN}     & 1.438 & 0.624    & 0.589     & 1.143 \\
    \textsc{PIGNet}   & 2.640  & 0.511    & 0.489     & 2.110  \\
    \textsc{IGN}      & 1.433 & 0.698    & 0.641     & 1.169 \\
    \textsc{HOLOPROT} & 1.546 & 0.602    & 0.571     & 1.208 \\
    \textsc{STAMPDPI} & 1.658 & 0.545    & 0.411     & 1.325 \\
    \textsc{TANKBind} & \textbf{1.346} & \textbf{0.726}    & 0.703     & 1.070 \\  \midrule
    \textbf{\textsc{DiffDock}}  & 1.347 & 0.692 & \textbf{0.718} & \textbf{1.052} \\
    \bottomrule
    \end{tabular}
    \end{center}
    \end{small}
\end{table}

\paragraph{Results.} The results presented in Table~\ref{tab:results_affinity} highlight how even preliminary results with a straightforward way of using \textsc{DiffDock}'s predictions for affinity prediction achieve a performance that is on par with the state-of-the-art. We hope this can motivate future work on better integrating affinity prediction in the method and scaling to larger amounts of data.

\subsection{Visualizations}

\begin{figure}[htb]
    \centering
    \includegraphics[width=0.95\textwidth]{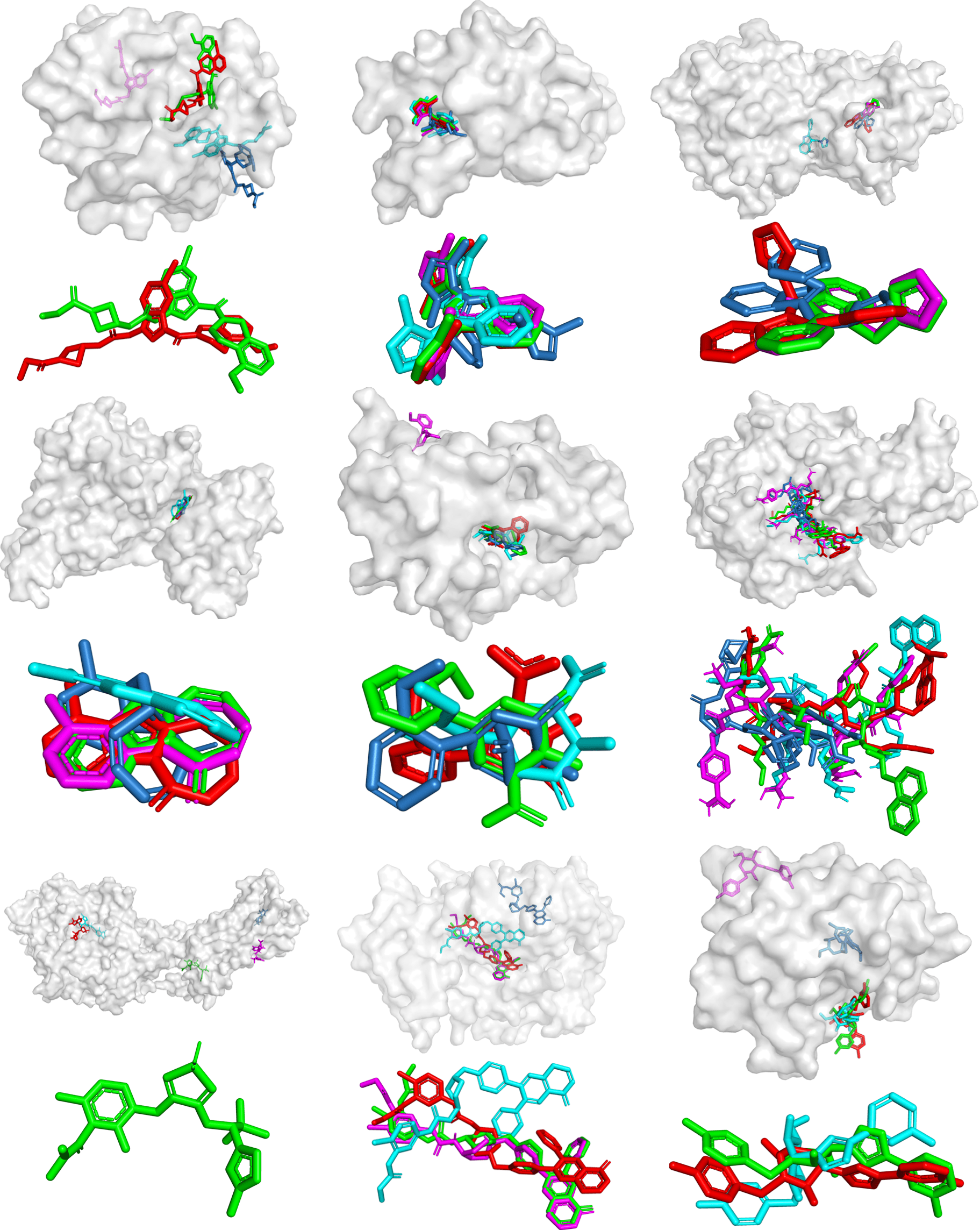}
    \caption{\textbf{Randomly picked examples.} The predictions of TANKBind (blue), EquiBind (cyan), GNINA (magenta), \textsc{DiffDock} (red), and crystal structure (green). Shown are the predictions once with the protein and without it below. The complexes were chosen with a random number generator from the test set. TANKBind often produces self intersections (examples at the top-right; middle-middle; middle-right; bottom-right). \textsc{DiffDock} and GNINA sometimes almost perfectly predict the bound structure (e.g., top-middle). The complexes in reading order are: 6p8y, 6mo8, 6pya, 6t6a, 6e30, 6hld, 6qzh, 6hhg, 6qln.}
    \label{fig:random_examples}
\end{figure}

\begin{figure}[ht]
    \centering
    \includegraphics[width=0.85\textwidth]{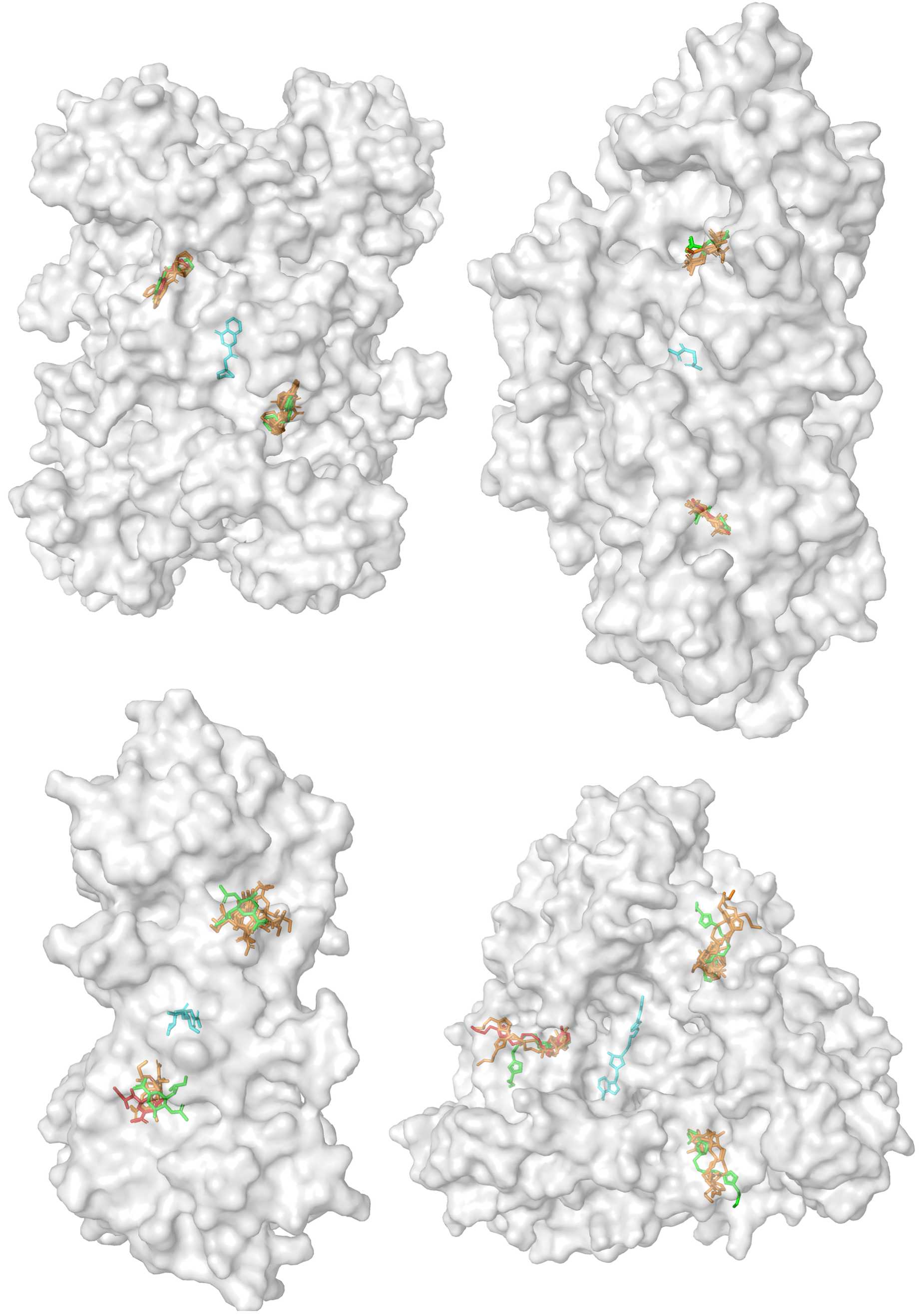}
    \caption{\textbf{Symmetric complexes and multiple modes.} EquiBind (cyan), \textsc{DiffDock} highest confidence sample (red), all other \textsc{DiffDock} samples (orange), and the crystal structure (green). We see that, since it is a generative model, \textsc{DiffDock} is able to produce multiple correct modes and to sample around them. Meanwhile, as a regression-based model, EquiBind is only able to predict a structure at the mean of the modes. The complexes are unseen during training. The PDB IDs in reading order: 6agt, 6gdy, 6ckl, 6dz3.}
    \label{fig:symmetric_complexes}
\end{figure}

\begin{figure}[ht]
    \centering
    \includegraphics[width=\textwidth]{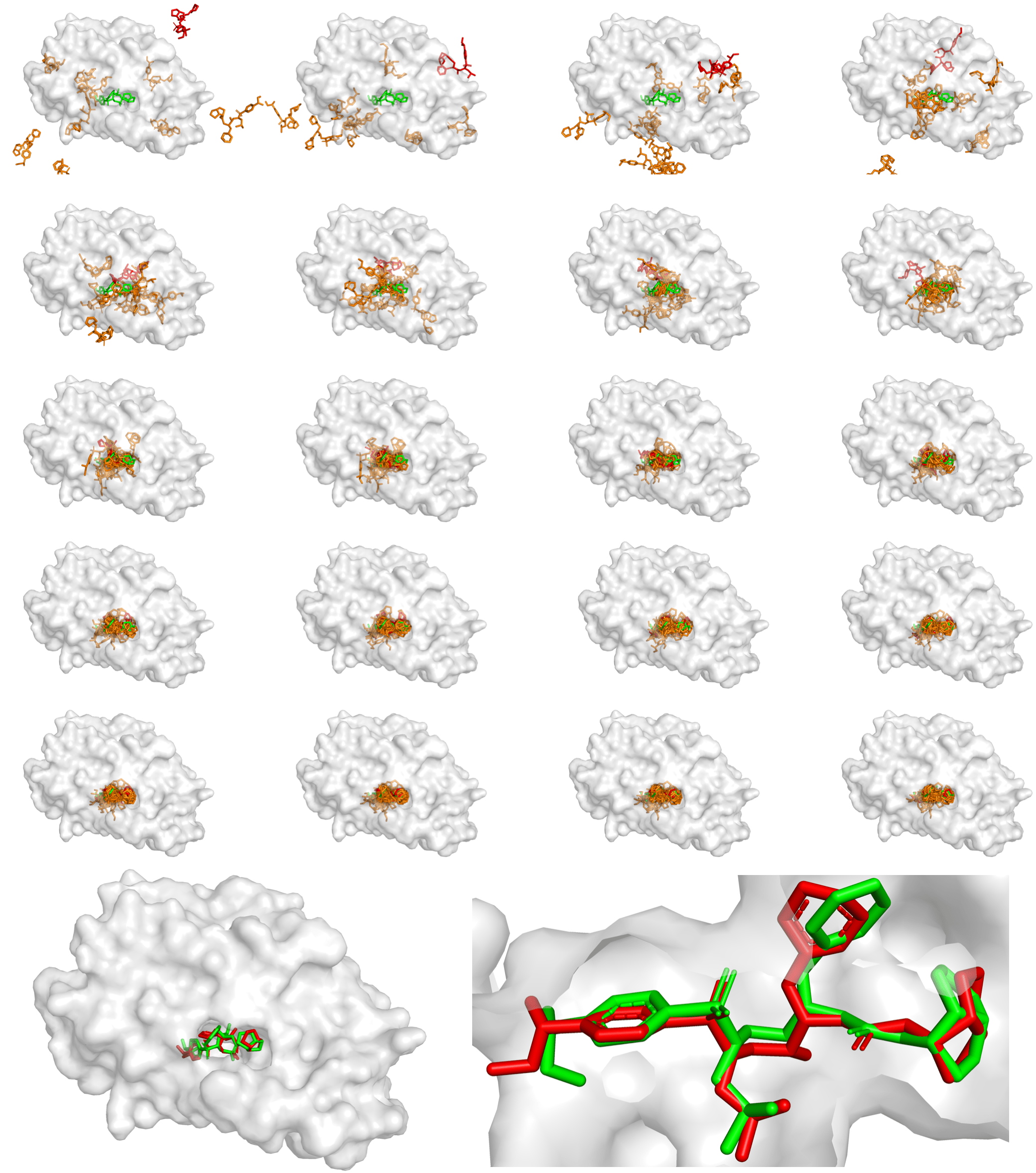}
    \caption{\textbf{Reverse Diffusion.} Reverse diffusion of a randomly picked complex from the test set. Shown are \textsc{DiffDock} highest confidence sample (red), all other \textsc{DiffDock} samples (orange), and the crystal structure (green). Shown are the 20 steps of the reverse diffusion process (in reading order) of \textsc{DiffDock} for the complex 6oxx. Videos of the reverse diffusion are available at \url{https://github.com/gcorso/DiffDock}.}
    \label{fig:reverse_diffusion2}
\end{figure}

\begin{singlespace}
\bibliography{main}
\bibliographystyle{plain}
\end{singlespace}

\end{document}